\documentclass[smallcondensed,numbook,envcountsame,final]{svjour3}
\def\makeheadbox{\relax}

\AtBeginDocument{%
	\setlength{\oddsidemargin}{\dimexpr(\paperwidth-\textwidth)/2-1in}%
	\setlength{\evensidemargin}{\oddsidemargin}%
	\setlength{\topmargin}{%
	\dimexpr(\paperheight-\textheight)/2-\headheight-\headsep-1in}%
}

\usepackage{bbm}
\usepackage{mathtools}

\usepackage{mathptmx}
\usepackage{dsfont}
\usepackage[symbol]{footmisc}
\DeclareMathAlphabet{\mathcal}{OMS}{cmsy}{m}{n}
\SetMathAlphabet{\mathcal}{bold}{OMS}{cmsy}{b}{n}

\DeclareSymbolFont{matha}{OML}{rtxmi}{m}{it}
\DeclareMathSymbol{\varv}{\mathord}{matha}{118}

\usepackage{pst-plot}
\usepackage[width=0.80\textwidth,font=small,labelfont=bf]{caption}
\usepackage{float}

\usepackage[hyperfootnotes=false,pagebackref,colorlinks=true,urlcolor=blue,linkcolor=blue,citecolor=blue]{hyperref}
\renewcommand*{\backref}[1]{}\renewcommand*{\backrefalt}[4]{\ifcase #1 (Not cited.)\or (Cited on page~#2.)\else (Cited on pages~#2.)\fi}

\smartqed

\def\EE{\mathbbm{E}}
\def\PP{\mathbbm{P}}
\def\GG{\mathbbm{W}}
\def\R{{\mathbbm R}}
\def\N{{\mathbbm N}}
\def\C{{\mathbbm C}}
\def\DD{{\mathbbm D}}
\def\B{{\textnormal B}}
\def\L{{\textnormal L}}
\def\D{{\textnormal D}}

\def\e{\T{e}}
\def\i{\T{i}}
\def\bb{\mathsf{b}}
\def\vv{\mathsf{v}}
\def\ww{\mathsf{w}}
\def\gg{\mathsf{g}}
\def\tr{{\rm Tr\,}}
\def\um{{\mathbbm 1}}
\def\T#1{\textnormal{#1}}
\def\ann{\T{\tiny ann}}
\def\x{\,\T{{x}}}
\def\y{\,\T{{y}}}
\def\z{\,\T{{z}}}
\def\ceq{\coloneqq}
\def\eqc{\eqqcolon}

\def\d#1{\!\GG_{N}(\T{d}{#1})\,}
\def\dD#1{\!\DD_{N}(\T{d}{#1})\,}
\def\dl#1{\!\!\T{d}{#1}\,}
\def\dll#1{\T{d}{#1}\,}

\makeatletter
	\newsavebox{\@brx}
	\newcommand{\llangle}[1][]{\savebox{\@brx}{\(\m@th{#1\langle}\)}
	\mathopen{\copy\@brx\kern-0.5\wd\@brx\usebox{\@brx}}}
	\newcommand{\rrangle}[1][]{\savebox{\@brx}{\(\m@th{#1\rangle}\)}
	\mathclose{\copy\@brx\kern-0.5\wd\@brx\usebox{\@brx}}}
	\DeclareRobustCommand*{\bfseries}{
	\not@math@alphabet\bfseries\mathbf
	\fontseries\bfdefault\selectfont
	\boldmath}
	\renewcommand\@biblabel[1]{[#1]}
\makeatother

\renewcommand{\labelenumi}{(\alph{enumi})}
\renewcommand{\theenumi}{(\alph{enumi})}

\begin{document}

\title{The free energy of a quantum Sherrington--Kirkpatrick spin-glass model for weak disorder}
\titlerunning{Quantum spin-glass model}

\dedication{\vspace{-1.9em}\begin{center}Dedicated to \\ Michael AIZENMAN, Joel LEBOWITZ, David RUELLE \\ on the occasions of their big birthdays in 2020.\end{center}}

\author{Hajo~Leschke$^{1,2}$ \and Sebastian~Rothlauf$^{1}$ \and Rainer~Ruder$^{2,1}$ \and Wolfgang~Spitzer$^{2}$}
\authorrunning{H. Leschke et al.}

\institute{
	{}$^{1}$Institut f\"ur Theoretische Physik, Universit\"at Erlangen--N\"urnberg, Staudtstra\ss e 7, 91058 Erlangen, Germany
\and
	{}$^2$Fakult\"at f\"ur Mathematik und Informatik, FernUniversit\"at in Hagen, Universit\"atsstra\ss e 1, 58097 Hagen, Germany
}


\date{\\Date of this version:~26~October~2021\hfill Preprint:~\href{https://arxiv.org/abs/1912.06633}{arXiv:1912.06633} [math-ph] (2019)\\ \phantom{x} \hfill Print (slightly shortened): J.\,Stat.\,Phys.\,\textbf{182}\,(3), 55, 41\,pp. (2021)\\\phantom{x}\hfill \href{https://doi.org/10.1007/s10955-020-02689-8}{Open Access} online since: 06 March 2021}
\maketitle

\begin{abstract}\vspace{-1em}
We extend two rigorous results of {\scshape Aizenman}, {\scshape Lebowitz}, and {\scshape Ruelle} in their pioneering paper of 1987 on the {\scshape Sherrington--Kirkpatrick} spin-glass model without external magnetic field to the quantum case with a ``transverse field" of strength $\bb$. More precisely, if the {\scshape Gauss}ian disorder is weak in the sense that its standard deviation $\vv>0$ is smaller than the temperature $1/\beta$, then the (random) free energy almost surely equals the annealed free energy in the macroscopic limit and there is no spin-glass phase for any $\bb/\vv\geq 0$. The macroscopic annealed free energy (times $\beta$) turns out to be non-trivial and given, for any $\beta\vv>0$, by the global minimum of a certain functional of square-integrable functions on the unit square according to a {\scshape Varadhan} large-deviation principle. For $\beta\vv<1$ we determine this minimum up to the order $(\beta\vv)^{4}$ with the {\scshape Taylor} coefficients explicitly given as functions of $\beta\bb$ and with a remainder not exceeding $(\beta\vv)^{6}/16$. As a by-product we prove that the so-called static approximation to the minimization problem yields the wrong $\beta\bb$-dependence even to lowest order. Our main tool for dealing with the non-commutativity of the spin-operator components is a probabilistic representation of the {\scshape Boltzmann}--{\scshape Gibbs} operator by a {\scshape Feynman--Kac} (path-integral) formula based on an independent collection of {\scshape Poisson} processes in the positive half-line with common rate $\beta\bb$. Its essence dates back to {\scshape Kac} in 1956, but the formula was published only in 1989 by {\scshape Gaveau} and {\scshape Schulman}. 
\end{abstract}

\tableofcontents

\section{Introduction and definition of the model}\label{introduction}

A spin glass is a spatially disordered material exhibiting at low temperatures a complex magnetic phase without spatial long-range order, in contrast to a ferro- or antiferromagnetic phase \cite{FH1991,M1993,N2001}. To this day most theoretical studies of spin glasses are based on models which go back to the classic(al) {\scshape Sherrington--Kirkpatrick (SK)} model \cite{SK1975}. In this simplified model {\scshape (Lenz--)Ising} spins are pairwise and multiplicatively coupled to each other via independent and identically distributed ({\scshape Gauss}ian) random variables and are possibly subject to an external (``longitudinal") magnetic field. The {\scshape SK} model may be viewed as a generalization of the traditional {\scshape Curie--Weiss (CW)} model in which the spin coupling is given by a single (non-random) constant of a suitable sign. In both models the pair interaction is of the somewhat unrealistic \emph{mean-field} type in the sense that it is the same for all spin pairs \cite{NS2003}. This neglect of geometric distances requires the effective strength of the pair interaction  to decrease sufficiently fast with increasing total number of the spins in order to ensure thermostatic behavior in the limit of macroscopically many of them. The notion ``mean field'' indicates the comfortable fact that the {\scshape Bragg--Williams} mean-field approximation of equilibrium statistical mechanics \cite{H1987} yields the exact free energy in this limit \cite{FSV1980}. According to standard textbook wisdom it is easy to calculate the macroscopic free energy of the {\scshape CW} model and to show that it provides a simplified but qualitatively correct description of the onset of ferromagnetism at low temperatures \cite{D1999}. In contrast, for the {\scshape SK} model the calculation has turned out to be much harder due to the interplay between thermal and disorder fluctuations, in particular for low temperatures. Nevertheless, by an ingenious application of the heuristic replica approach, see \cite{FH1991,M1993,N2001}, {\scshape Parisi} found that the macroscopic (quenched) free energy of the {\scshape SK} model is given by the global maximum of a rather complex functional of (probabilistic) distribution functions on the unit interval \cite{P1980a,P1980b}. The unique maximizing distribution function is interpreted as the (functional) spin-glass ``order parameter". The attempt at understanding this {\scshape Parisi} formula became a challenge to mathematical physicists and mathematicians \cite{T1998}. Highly gratifying for him and his intuition \cite{P2009}, the formula was eventually confirmed by a mathematically rigorous proof due to the efforts and insights of {\scshape Guerra}, {\scshape Talagrand}, and others \cite{GT2002,G2003,ASS2003,T2006,T2011b,P2013,AC2015}.

Since magnetic properties cannot be explained at the (sub)microscopic level of atoms and molecules by classical physics alone, some spin glasses require for fundamental and experimental reasons a quantum-theoretical modelling. Of course, the {\scshape SK} model may be viewed as a simplistic quantum model by interpreting the values of the {\scshape Ising} spins as (twice) the eigenvalues of one and the same component of associated three-component spin \emph{operators} each of them with (main) quantum number $1/2$. But a genuine quantum {\scshape SK} model with quantum fluctuations and inherent dynamics needs the presence of different (non-commuting) components of the spin operators. The theory of such a model was pioneered by {\scshape Bray} and {\scshape Moore} \cite{BM1980} and by {\scshape Sommers} \cite{S1981}. More precisely, for a quantum spin-glass model with isotropic ({\scshape Frenkel--)Heisenberg(--Dirac)} spin coupling of mean-field type these authors handled the competition of thermal, disorder, and quantum fluctuations by combining the {\scshape Dyson}--{\scshape Feynman} time-ordering of operator products with the replica approach \cite{BM1980} or with the {\scshape Thouless}--{\scshape Anderson}--{\scshape Pal\-mer} ({\scshape TAP}) approach \cite{S1981}. For the {\scshape TAP} approach see \cite{FH1991,M1993,N2001}. Since these authors did not aim at rigorous results, they applied the so-called \emph{static approximation} to simplify the rather complicated equations derived by them. However, this approximation is still insufficiently understood -- even for higher temperatures.

A simpler genuine quantum {\scshape SK} model is obtained by considering an extremely anisotro\-pic pair interaction where only one component of the spins is coupled which is perpendicular to the direction of the external magnetic field. This model was introduced by {\scshape Ishii} and {\scshape Yamamoto} \cite{IY1985} and approximately studied within the {\scshape TAP} approach. It is usually called the \emph{{\scshape SK} model with (or ``in'') a transverse field}, see \cite{SIC2013} and references therein. It is this model to which we devote ourselves in the present paper. It is characterized by the random energy operator or {\scshape Hamilton}ian
\begin{equation}\label{H_N}
	\boxed{H_{N}\ceq-\bb\!\!\!\sum\limits_{1\leq i \leq N} S_i^{\x}-\frac{\vv}{\sqrt{N}}\sum\limits_{1\leq i<j\leq N}\!\!\!g_{ij}S_i^{\z} S_j^{\z}}
\end{equation}
acting selfadjointly on the $N$-spin {\scshape Hilbert} space (isometrically isomorphic to) 
\[
	\C^2\otimes\cdots\otimes\C^2\eqc (\C^2)^{\otimes N}\cong\C^{2^N}\,,
\]
that is, the $N$-fold tensor product of the two-dimensional complex {\scshape Hilbert} space $\C^{2}$ for a single spin. Here $N\geq 2$ is the total number of a collection of three-component spin-$1/2$ operators where the selfadjoint spin operator $S^{\,\alpha}_{i}/2$ with component $\alpha$ and index (or ``site'') $i$ is given by the tensor product of $N$ factors according to
\[
	S^{\,\alpha}_i\ceq\um\otimes\cdots\otimes\um\otimes S^{\,\alpha}\otimes \um\otimes\cdots\otimes\um\qquad \big(\alpha\in\{\x,\y,\z\},\,\, i\in\{1,\dots,N\}\big)\,.
\]
In this definition the identity operator $\um$ and the operator $S^{\,\alpha}$, as the $i$-th factor, are understood to act (a priori) on $\C^{2}$ and satisfy the (specialized) {\scshape Dirac} identities
\begin{equation}\label{dirac}
 (S^{\,\alpha})^{2}=\um, \quad S^{\x}S^{\y}=\i S^{\z}, \quad S^{\y}S^{\z}=\i S^{\x}, \quad S^{\z}S^{\x}=\i S^{\y}
 \end{equation}
 with $\i\equiv\sqrt{-1}$ denoting the imaginary unit.
With respect to the eigenbasis of $S^{\z}$ these four operators are represented by the $2\times2$ unit matrix and the triple of $2\times2$  {\scshape Pauli} matrices according to
\[
\um=\begin{bmatrix}1&0\\0&1\end{bmatrix},
\qquad S^{\x}=\begin{bmatrix}0&1\\1&0\end{bmatrix},
\qquad S^{\y}=\begin{bmatrix}0&-\i\\\i&0\end{bmatrix},
\qquad S^{\z}=\begin{bmatrix}1&0\\0&-1\end{bmatrix}\,.
\]

The first term in \eqref{H_N} models an ideal (quantum) paramagnet and represents the energy of the spins due to their individual interactions with a constant magnetic field of strength $\bb\geq0$ externally applied along the positive $\x$-direction. The second term in \eqref{H_N} models disorder in spin glasses and represents the energy of the spins due to random mean-field type pair interactions of their $\z$-components. More precisely, we assume the $N(N-1)/2$ coupling coefficients $(g_{ij})_{1\leq i<j\leq N}$ to form a collection of jointly {\scshape Gauss}ian random variables with mean $\EE[g_{ij}]=0$ and covariance $\EE[g_{ij}g_{kl}]=\delta_{ik}\delta_{jl}$ (in terms of the {\scshape Kronecker} delta). The parameter $\vv>0$ is the standard deviation of $\vv g_{ij}$ and stands for the strength of the disorder. At given $\bb$ or $\vv$ quantum fluctuations become more important with increasing $\vv$ or $\bb$, respectively -- due to the non-commutativity  of $S^{\x}_{i}$ and $S^{\z}_{i}$.

We proceed by introducing the basic thermostatic quantity of the model \eqref{H_N}. For any reciprocal (absolute) temperature $\beta\in{\rbrack 0,\infty\lbrack}$, we define the random partition function (or sum) as the trace
\begin{equation}\label{Z_N}
	Z_{N}\ceq\tr{\e^{-\beta H_{N}}}
\end{equation}
of the {\scshape Boltzmann--Gibbs} operator and the (specific {\scshape Gibbs}) \emph{free energy} by
\begin{equation}\label{f_N}
	f_{N}\ceq-\frac{1}{N\beta}\ln(Z_{N})\,,
\end{equation}
which is the random variable of main physical interest, in particular, in the \emph{macroscopic limit} $N\to\infty$. The disorder average $\EE[f_{N}]$ of $f_{N}$ is called the mean or \emph{quenched} free energy and has to be distinguished from the \emph{annealed} free energy,
\begin{equation}\label{f_ann_N}
f_{N}^\ann\ceq-\frac{1}{N\beta}\,\ln\big(\EE[Z_{N}]\big)\,.
\end{equation}
The latter is physically less relevant (for spin glasses with ``frozen-in'' disorder), but mathematically more accessible and provides a lower bound on the quenched free energy by the concavity of the logarithm and the {\scshape Jensen} inequality \cite{J1906} (see also \cite[Lem.\,3.5]{K2002}),
\begin{equation}\label{jensen}
	f^\ann_{N}\leq \EE[f_{N}]\,.
\end{equation}

Over the years the work \cite{IY1985} has stimulated many further approximate and numerical studies devoted to the macroscopic quenched free energy of the quantum {\scshape{SK}} model \eqref{H_N} and the resulting phase diagram in the temperature-field plane, among them \cite{FS1986,YI1987,US1987,K1988,RCC1989,BU1990a,GL1990,BU1990b,MH1993,KK2002,T2007,Y2017,MRC2018}. Not surprisingly, this has led to partially conflicting results, especially for low temperatures. 

From a rigorous point of view, a solid understanding of the low-temperature regime seems still to be out of reach. The main and modest aim of the present paper is therefore to provide the first rigorous explicit results for the opposite regime characterized by $\beta\vv <1$. Since in this regime $\beta\bb\geq 0$ may be arbitrary, we call it the \emph {weak-disorder regime}. In the following sections we firstly compile some properties of $f^{\ann}_{N}$ for arbitrary $\beta\vv > 0$. Next we show that $f^{\ann}_{N}$ has a well-defined macroscopic limit $ f_{\infty}^{\ann}$ with similar and well-understood properties, see Theorem\,\ref{f_ann_exist}, Theorem\,\ref{f_ann_variational}, and Theorem\,\ref{f_ann_second_order} below. In particular, for $\beta\vv <1$ the limit $\beta f^{\ann}_{\infty}$ takes a rather explicit form as a function of $\beta\vv$ and $\beta\bb$. Then we prove that the more important free energies $f_N$ and $\EE[f_N]$ have both $f_{\infty}^{\ann}$ as its (almost sure) macroscopic limit if $\beta\vv <1$, see Theorem\,\ref{difference} and Corollary\,\ref{weak_disorder_free_energy}. For $\beta\vv <1$ we also prove the absence of spin-glass order in the sense that $\lim_{N\to\infty}\big(\langle S^{\z}_{1} S^{\z}_{2}\rangle\big)^{2}=0$, almost surely, see Corollary\,\ref{no_sg} and Remark\,\ref{rem_order_parameter}. Here $\langle\,(\cdot)\,\rangle\ceq\e^{N\beta f_{N}} \tr\e^{-\beta H_{N}}(\cdot)$ denotes the (random) {\scshape Gibbs} expectation induced by $H_{N}$. These results extend two of the pioneering results of {\scshape Aizenman}, {\scshape Lebowitz}, and {\scshape Ruelle} \cite{ALR1987} for the model  \eqref{H_N} with $\bb=0$ to the quantum case $\bb>0$. Unfortunately, for any $\beta\vv >1$ we only have the somewhat weak result that the difference between the macroscopic quenched and annealed free energies is strictly positive if the ratio $\bb/\vv$ is sufficiently small.

To our knowledge, the only other rigorous results for the quantum {\scshape{SK}} model \eqref{H_N} are due to {\scshape Crawford} \cite{C2007} and to {\scshape Adhikari} and {\scshape Brennecke} \cite{AB2020}. {\scshape Crawford} has extended key results of {\scshape Guerra} and {\scshape Toninelli} \cite{GT2002} and {\scshape Carmona} and {\scshape Hu} \cite{CH2006} for the model  \eqref{H_N} with $\bb=0$ to the quantum case $\bb>0$. More precisely, he has proved the existence of the macroscopic (quenched) free energy not only for $\beta\vv <1$, but for \emph{all} $\beta\vv>0$ (without a formula). Furthermore, he has shown that the limit is the same for random variables $(g_{ij})_{1\leq i<j\leq N}$ which are not necessarily {\scshape Gauss}ian but merely independently and identically distributed with $\EE[g_{12}]=0$, $\EE[(g_{12})^{2}]=1$, and $\EE[|g_{12}|^{3}]<\infty$. More recently, {\scshape Adhikari} and {\scshape Brennecke} have provided a variational formula for the macroscopic (quenched) free energy for all $\beta\vv>0$. This formula is still somewhat implicit and given by a suitable $d\to\infty$ limit of a {\scshape Parisi}-like functional for a classical $d$-component vector-spin-glass model, due to {\scshape Panchenko}.
 
\section{The annealed free energy and its deviation from the quenched free energy}\label{sec_annealed_free_energy}
In this section we attend to the annealed free energy $f_{N}^{\ann}$ for arbitrary values of $N\geq 2$, $\beta\vv>0$, and $\beta\bb>0$. According to \eqref{f_ann_N} we have to perform the {\scshape Gauss}ian disorder average of the partition function $Z_{N}$.
In order to do so explicitly, we will use the following {\scshape Poisson--Feynman--Kac (PFK)} probabilistic representation of $Z_{N}$ in terms of $N$ copies of a {\scshape Poisson} process with constant rate (or intensity parameter) $\beta\bb$:
\begin{equation}
Z_{N}=\big(\cosh(\beta\bb)\big)^{ N}\sum\limits_{s}\Big\langle\exp\Big(-\beta\int_0^1\dl{t} h_{N}\big(s\sigma(t)\big)\Big)\Big\rangle_{\!\!\beta\bb}\label{fk}\,.
\end{equation}
Here, the classical {\scshape Hamilton}ian $h_{N}$, characterizing the zero-field {\scshape SK} model \cite{SK1975}, is defined by
\begin{equation}
h_{N}(s)\ceq -\frac{\vv}{\sqrt{N}}\sum\limits_{1\leq i<j\leq N}g_{ij}s_i s_j\,,\label{sk}
\end{equation}
where $s\ceq(s_{1},\dots,s_{N})\in\{-1,1\}\times\dots\times\{-1,1\}\eqc \{-1,1\}^{N} $ denotes one of the $2^N$ classical spin configurations and the notation $\sum_{s}$ indicates summation over all of them. The integrand in \eqref{fk} is obtained from \eqref{sk} by replacing there each $s_i$ by the product $s_{i}\sigma_i(t)$, where
\begin{equation}
\sigma_i(t)\ceq(-1)^{{\cal N}_{i}(t)}\qquad \big(t\in{\lbrack0,\infty\lbrack}\,,\,\, i\in\{1,\dots,N\}\big)\label{teleproc}
\end{equation}
defines the \emph{spin-flip process} with index $i$, in other words, a ``(semi-)random telegraph signal" \cite{K1974,KR2013}. It is a continuous-time-homogeneous pure jump-type two-state {\scshape Markov} process \cite[Ch.\,12]{K2002} steered by a simple {\scshape Poisson} process ${\cal N}_{i}$ in the positive half-line.\footnote{For a concise definition of {\scshape Poisson} (point) processes well suited for our purposes the reader may consult Appendix\,\ref{poisson_appendix}.} The random variable ${\cal N}_{i}(t)$ is $\N_{0}$-valued and {\scshape Poisson} distributed with mean $\beta\bb\,t\geq0$ independent of the index~$i$. The $N$ {\scshape Poisson} processes ${\cal N}_{1},\dots,{\cal N}_{N}$ are assumed to be (stochastically) mutually independent. The angular brackets $\langle\,(\cdot)\,\rangle_{\beta\bb}$ denote the corresponding joint {\scshape Poisson} expectation \emph{conditional} on $\sigma_{i}(1)=1$ for all $i\in\{1,\dots,N\}$.
In \eqref{fk} and in the following we often write $\sigma(t)\ceq\big(\sigma_{1}(t),\dots,\sigma_{N}(t)\big)$  and suppress the $N$-dependence of $\langle\,(\cdot)\,\rangle_{\beta\bb}$ to keep the notation simple. For the validity of the {\scshape PFK} representation \eqref{fk} we refer to \eqref{PFK_trace_multi}  in Appendix\,\ref{pfk_proof}.

For performing the {\scshape Gauss}ian disorder average of the partition function $Z_{N}$ we start out with the disorder mean
\begin{align}
	\EE\big[\beta h_{N}(s)\big]&=0\label{h_N_mean}
	\intertext{and the disorder covariance}
	\EE\big[\beta^{2} h_{N}(s)h_{N}(\widehat{s})\big]&=2\lambda\big(N\big[Q_{N}(s,\widehat{s}\,)\big]^{2}-1\big)
	\label{h_N_cov}
\end{align}
of the classical {\scshape Hamilton}ian in terms of the dimensionless \emph{disorder parameter} $\lambda\ceq\beta^{2}\vv^{2}/4$ and the overlap 
\begin{equation}\label{Q_N}
	Q_{ N}(s,\widehat{s})\ceq\frac{1}{N}\sum\limits_{i=1}^{N}s_{i}\widehat{s}_{i}\qquad \big(s,\widehat{s}\in\{-1,1\}^{N}\big)
\end{equation}
between two classical spin configurations.
Formula~\eqref{fk} then gives
\begin{equation}\label{Z_N_mean}
	\EE[Z_{N}]=\big(2\cosh(\beta\bb)\big)^{ N}\e^{-\lambda}\big\langle{\cal Z}_{N}\big\rangle_{\!\beta\bb}\,.
\end{equation}
Here, the functional ${\cal Z}_{N}:\sigma\mapsto{\cal Z}_{N}(\sigma)$ is a random variable with respect to the underlying $N$ {\scshape Poisson} processes and defined by
\begin{align}\label{Phi_N1}
	{\cal Z}_{N}(\sigma)&\!\!\ceq \e^{\lambda}\,2^{-N}\sum_{s}\EE\Big[\exp\Big(-\beta\int_{0}^{1}\dl{t} h_{N}\big(s\sigma(t)\big)\Big)\Big]\\
	&=\e^{\lambda}\,2^{-N}\sum_{s}\exp\Big(\frac{1}{2}\EE\Big[\Big(-\beta\int_{0}^{1}\dl{t} h_{N}\big(s\sigma(t)\big)\Big)^{2}\Big]\Big)\label{Phi_N2}\\
	&=\exp\Big(N\lambda\int_{0}^{1}\dl{t}\int_{0}^{1}\dl{t'}\big[Q_{N}(\sigma(t),\sigma(t')\big]^{2}\Big)\label{Phi_N3}\\
	&\eqc\exp\Big(N\lambda P_{N}(\sigma)\Big)\,.\label{P_N}
\end{align}
Equation\,\eqref{Phi_N2} reflects the {\scshape Gauss}ianity of the disorder average with \eqref{h_N_mean}. Interchanging now the expectation with the (two-fold) time integration according to the {\scshape Fubini--Tonelli} theorem, using \eqref{h_N_cov}, and observing $(s_{i})^{2}=1$ yields \eqref{Phi_N3}.
  By $0\leq\big[Q_N(s,\widehat{s})\big]^2\leq 1$ the two-fold integral $P_{N}$ is a ${\lbrack0,1\rbrack}$-valued random variable so that we have the crude estimates

\begin{equation}
	1\leq{\cal Z}_N(\sigma)\leq \e^{N\lambda}\,.\label{Phi_N_bounds}
\end{equation}

Somewhat to our surprise, we have not succeeded in calculating $f^{\ann}_{N}$ explicitly, not even for $N\to\infty$, see however, Theorem\,\ref{f_ann_variational} and Theorem\,\ref{f_ann_second_order} below. For the time being we derive certain estimates and properties of $f^{\ann}_{N}$. 

For the formulation of the corresponding theorem we introduce some notation. We begin with the function $\mu:{\lbrack0,1\rbrack}\times{\lbrack0,1\rbrack}\to{\lbrack0,1\rbrack}$ defined by
\begin{equation}
\mu(t,t')\ceq\langle\, \sigma_i(t)\sigma_i(t')\, \rangle_{\!\beta\bb}
	=\frac{\cosh\big(\beta\bb(1-2|t-t'|)\big)}{\cosh(\beta\bb)}\geq\frac{1}{\cosh(\beta\bb)}\label{mu}
\end{equation}
for all $i\in\{1,\dots,N\}$. It is the thermal {\scshape Duhamel--Kubo} auto-correlation function of the $\z$-component of a single spin in the absence of disorder, $\lambda = 0$, see \eqref{Duhamel_Kubo} in Appendix\,\ref{pfk_proof}. Upon explicit integrations we get
\begin{equation}
	\int_0^1\dl{t}\int_0^1\dl{t'}\mu(t,{t'})=\int_{0}^{1}\dl{t}\mu(t,0)=\frac{\tanh(\beta\bb)}{\beta\bb}\eqc m\label{m}
\end{equation}
and
\begin{equation}
\int_0^1\dl{t}\int_0^1\dl{t'}\big(\mu(t,{t'})\big)^2=\int_0^1\dl{t}\big(\mu(t,0)\big)^2=\big(1-\big(\tanh(\beta\bb)\big)^{2}+m\big)/2\eqc p\,,\label{p}
\end{equation}
so that $\sqrt{2p-m}=1/\cosh(\beta\bb)$. Finally, we  introduce two positive sequences by
\begin{equation}
p_{N}\ceq p+(1-p)/N\quad,\quad G_{N}\ceq\ln\big(1+p_{N}(\e^{N\lambda}-1)\big)\quad(N\geq1)\,.\label{G_N}
\end{equation}
\begin{lemma}[Inequalities between $m$ and $p$, and bounds on $G_{N}$]\label{mpG_N_lemma}

\noindent For any $N\geq2$, $\beta\bb>0$, and $\lambda>0$ we have the inequalities
	\begin{equation}\label{bound_mp}
		0<m^{2}< p< m< \min\{1,2p\}\leq 2p<(1+p)m\,,
	\end{equation}
\begin{equation}\label{G_N_weak}
	p_{N}\lambda<\ln\big(1+p_{N}(\e^{\lambda}-1)\big)< G_{N}/N< p_{N}\lambda+(1-p_{N})N\lambda^{2}/2\,,
\end{equation}
\begin{equation}\label{G_N_strong}
	\max\{0,\lambda +\ln(p_{N})/N\} < G_{N}/N < \lambda\,.
\end{equation}
\end{lemma} 
\begin{proof}
The first five inequalities in \eqref{bound_mp} are obvious. The last one is a consequence of the elementary inequalities $\sinh(x)\geq x +x^{3}/6$ and $\tanh(x)\geq x-x^{3}/3$ for $x\geq0$. The first inequality in \eqref{G_N_weak}  follows from the convexity of the exponential and the {\scshape Jensen} inequality for a $\{0,1\}$-valued {\scshape Bernoulli} random variable taking the 
``success" value $1$ with probability $p_{N}\in{\lbrack0,1\rbrack}$. For the second inequality we use $\big(1+p_{N}(\e^{\lambda}-1)\big)^{N}\leq 1+p_{N}(\e^{N\lambda}-1)=\exp(G_{N})$ by the convexity of the $N$-th power $x\mapsto x^{N}$ for $x\geq0$ and the {\scshape Jensen} inequality. The last inequality is an application of 
	\begin{equation}\label{log_inequality}
		\ln\big((1+a(\e^{x}-1)\big)\leq a|x|+(1-a)x^{2}/2\qquad \big(a\in{\lbrack0,1\rbrack}\,,\,\, x\in\R\big)
	\end{equation}
for $a=p_{N}$ and $x=N\lambda$. Inequality \eqref{log_inequality} itself follows from $x\leq|x|$, $\exp(-|x|)\leq 1-|x|+x^{2}/2$, and $\ln(y)\leq y-1$ for $y>0$.
The inequalities \eqref{G_N_strong} are simple consequences of $p_{N}<1$.  
\qed
\end{proof}
Now we are prepared to present our first result.
\begin{theorem}[On the annealed free energy] \label{f_ann_bounds}
\noindent
\begin{enumerate}
\item \label{f_ann_estimates}For any value of the dimensionless disorder parameter $\lambda=\beta^{2}\vv^{2}/4>0$ and any number of spins $N\geq 2$ we have the three estimates
	\begin{equation}
	-\frac{G_{N}}{N}+\frac{\lambda}{N}\leq \beta f^{\ann}_N +\ln\big(2\cosh(\beta\bb)\big)\leq -p\,\lambda\,\frac{N-1}{N}\,,\label{f_ann_N_bounds}
	\end{equation}
	\begin{equation}
	\beta f^{\ann}_N \leq -\lambda \frac{N-1}{N}-\ln(2)\,.\label{f_ann_N_bound2}
	\end{equation}

\item \label{f_ann_N_behavior} The dimensionless quantity $\beta f^{\ann}_{N}$ depends on the disorder strength $\vv$ only via the variable $\lambda>0$. The function $\lambda\mapsto \beta f^{\ann}_{N}$ is concave, is not increasing, and has the following weak- and strong-disorder limits
	\begin{align}
		\lim\limits_{\lambda\downarrow 0}\Big(\frac{1}{\lambda}\big(\beta f^{\ann}_{N}+\ln\big(2\cosh(\beta\bb)\big)\Big)&=-p\,\frac{N-1}{N}\,,\label{f_ann_N_weak_disorder}\\
		\lim\limits_{\lambda\to\infty}\frac{1}{\lambda}\beta f^{\ann}_{N}&=-\frac{N-1}{N}\label{f_ann_N_strong_disorder}\,.
	\end{align}

\item\label{f_ann_N_behavior2} The function $\beta\mapsto\beta f^{\ann}_{N}$ is concave for any $\vv>0$.
\end{enumerate}
\end{theorem}
\begin{proof}
	\begin{enumerate}
\item The claimed inequalities \eqref{f_ann_N_bounds} are equivalent to
\begin{equation}
	Np_{N}\lambda\leq F_{N} \leq G_{N}\,, \label{F_N_bounds}
\end{equation}
where 
\begin{align}\label{F_N1}
F_{N}\ceq&\,\lambda-N \beta  f^{\ann}_N-N\ln\big(2\cosh(\beta\bb)\big)\\
=&\,\ln\big(\big\langle\exp\big(N\lambda P_{N}\big)\big\rangle_{\!\beta\bb}\big)\,,\label{F_N2}
\end{align}
confer \eqref{f_ann_N}, \eqref{Z_N_mean}, and \eqref{P_N}. By an explicit calculation, using \eqref{mu} and \eqref{p}, we obtain
\begin{equation}\label{P_N_mean}
	\langle P_{N}\rangle_{\beta\bb}=p_{N}\,.
\end{equation}
Consequently, the lower estimate in \eqref{F_N_bounds} follows from the convexity of the exponential and the {\scshape Jensen} inequality with respect to the {\scshape Poisson} expectation in \eqref{F_N2}. To obtain the upper estimate in \eqref{F_N_bounds} we recall from \eqref{Phi_N_bounds} that $P_{N}\in{\lbrack0,1\rbrack}$. Therefore we have
\begin{equation}\label{jensen4}
	\big\langle\exp\big(N\lambda P_{N}\big)\big\rangle_{\beta\bb}\leq 1+\big\langle P_{N}\big\rangle_{\beta\bb}\big(\e^{N\lambda}-1\big)
\end{equation}
by the ({\scshape Jensen}) inequality $\e^{ya}\leq 1+a(\e^{y}-1)$, for $y\in\R$ and $a\in{\lbrack0,1\rbrack}$, used already in the proof of the first inequality in \eqref{G_N_weak}, and by taking the {\scshape Poisson} expectation.
Estimate \eqref{f_ann_N_bound2} follows from the disorder average of the inequality $Z_{N}(\beta\bb,\beta\vv)\geq Z_{N}(0,\beta\vv)\ceq\lim_{b\downarrow0}Z_{N}(\beta\bb,\beta\vv)$ for the (random) trace $Z_{N}\equiv Z_{N}(\beta\bb,\beta\vv)$. For later reference we prove this probabilistically and estimate the {\scshape Poisson} expectation in \eqref{fk} from below by restricting it to the single realization without any spin flip (that is, without any jump) in the time interval ${\lbrack0,1\rbrack}$. This realization occurs if, and only if, the random variable $\prod_{i=1}^{N} 1(\sigma_{i})$, defined by $1(\sigma_{i})\ceq 1$ if $\sigma_{i}(t)=1$ for all $t\in{\lbrack0,1\rbrack}$ and $1(\sigma_{i})\ceq0$ otherwise, takes its maximum value 1. The probability of this event is 
\begin{equation}\label{prob_constant_path}
	\Big\langle\prod_{i=1}^{N}1(\sigma_{i})\Big\rangle_{\beta\bb}=\prod_{i=1}^{N}\big\langle 1(\sigma_{i})\big\rangle_{\beta\bb}=\Big(\big\langle 1(\sigma_{1})\big\rangle_{\beta\bb}\Big)^{N}=\big(\cosh(\beta\bb)\big)^{-N}\,.
\end{equation}
\item The first sentence is obvious by \eqref{F_N2}. This equation also shows that $\beta f^{\ann}_{N}$ is concave in $\lambda$, because the right-hand side of \eqref{F_N2} is convex by the {\scshape H\"older} inequality. The monotonicity in $\lambda$ then follows from the concavity and $\beta f^{\ann}_{N}\leq-\ln\big(2\cosh(\beta\bb)\big)$ for all $\lambda>0$ by \eqref{f_ann_N_bounds} with obvious equality in the limiting case $\lambda =0$. 
The claim \eqref{f_ann_N_weak_disorder} follows from \eqref{f_ann_N_bounds} and $\lim_{\lambda\downarrow 0} G_{N}/(\lambda N)=p_{N}$, see \eqref{G_N_weak}. The claim \eqref{f_ann_N_strong_disorder} follows from \eqref{f_ann_N_bound2} and the lower estimate in \eqref{f_ann_N_bounds} by using  $\lim_{\lambda\to\infty} G_{N}/(\lambda N)=1$, see \eqref{G_N_strong}.
\item This concavity follows from definition \eqref{f_ann_N} by the {\scshape H\"older} inequality with respect to the product measure $\EE[\tr(\cdot)]$.
\qed
\end{enumerate}
\end{proof}
\begin{remark}\label{remthm1}

\renewcommand{\labelenumi}{(\roman{enumi})}
\renewcommand{\theenumi}{(\roman{enumi})}

\begin{enumerate}
\item\label{zero-field} Equation\,\eqref{prob_constant_path} shows that the classical limit $\bb\downarrow0$ corresponds to taking into account only the single realization without any spin flip (for each spin).

\item \label{W_N_bounds} The last inequality in \eqref{G_N_strong} weakens the lower estimate in \eqref{f_ann_N_bounds} to $\lambda/N-\lambda$. This weaker lower estimate is quasi-classical in the sense of Lemma\,\ref{quasi_classical_lemma} below. It may also be derived from using the {\scshape Jensen} inequality $\exp\big(\int_{0}^{1}\dl{t}(\dots)\big)\leq\int_{0}^{1}\dl{t}\exp(\dots)$ in the {\scshape PFK} representation \eqref{fk} or from applying the {\scshape Golden--Thompson} inequality directly to the trace \eqref{Z_N}. The inequality $Z_{N}(\beta\bb,\beta\vv)\geq Z_{N}(0,\beta\vv)$, underlying the estimate \eqref{f_ann_N_bound2}, may also be derived directly from \eqref{Z_N} by applying the {\scshape Jensen--Peierls--Bogolyubov} inequality. For such trace inequalities we refer to \cite{S2005b}.\label{Golden_Thompson}

\item \label{p_remark} The parameter $p$ on both sides of \eqref{f_ann_N_bounds} is actually a bijective function of the product $\beta\bb>0$, see \eqref{p}. It is strictly decreasing, approaches its extreme values $1$ and $0$ in the limiting cases $\beta\bb\downarrow 0$ and $\beta\bb\to\infty$, respectively, and attains the value $1/2$ at $\beta\bb=1.19967\dots$, more precisely, at the solution of $\beta\bb\tanh(\beta\bb)=1$, see Figure\,\ref{p_plot}. In the first limiting case the estimates \eqref{f_ann_N_bounds} yield the well-known result for the zero-field {\scshape SK} model, given by the right-hand side of \eqref{f_ann_N_bound2}. 
They also guarantee that $f^{\ann}_{N}$ coincides with the free energy of the ideal paramagnet in the absence of disorder ($\lambda\downarrow0$).
In the opposite limit of extremely strong disorder ($\lambda\to\infty$) the result \eqref{f_ann_N_strong_disorder} shows that  the magnetic field becomes irrelevant in agreement with the zero-field {\scshape SK} model and ``physical intuition'' for the case $\beta\bb\ll\beta\vv$.
\begin{figure}
\begin{center}
	\psset{xunit=3cm,yunit=2.5cm}
	\begin{pspicture}(-0.6cm,-0.5cm)(10.4cm,3.3cm)
	\psaxes[Dx=0.5,Dy=0.5]{->}(3.3,1.2)
	\psset{algebraic,plotpoints=500}
	\psplot[linecolor=blue]{0.001}{3.1}{0.5/(COSH(x))^2+TANH(x)/(2*x)}
	\rput(0,1.3){$p$}
	\rput(3.4,0){$\beta\bb$}
	\end{pspicture}
	\caption{Plot of the parameter $p$ defined in \eqref{p} as a function of $\beta\bb$. See Remark\,\ref{remthm1}\,\ref{p_remark}.}
\label{p_plot}
\end{center}
\end{figure}
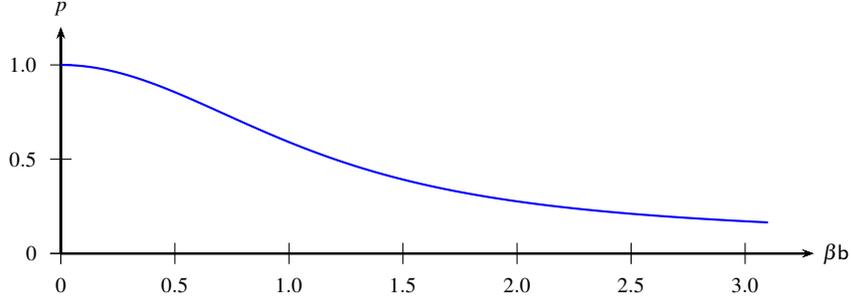
\item \label{rem_W_N} The lower estimate in \eqref{f_ann_N_bounds} may be sharpened by replacing $G_{N}$ with the less explicit bound
\begin{equation}\label{W_N}
\widetilde{G}_{N}\ceq \int_0^1\dl{t}\ln\Big(\!\!\int_{\R}\dl{x}w_N(x)\big[\cosh(\sqrt{4\lambda}\,x)+\mu(t,0)\sinh(\sqrt{4\lambda}\,x)\big]^N\Big)
\end{equation}
where  $w_{N}(x)\ceq\sqrt{N/\pi}\,\exp\big(-Nx^{2}\big)$ defines the centered {\scshape Gauss}ian probability density on the real line $\R$ with variance $1/(2N)$.
In fact, we have the two inequalities  $F_{N}\leq \widetilde{G}_{N}\leq G_{N}$. For their proofs we note that the quantity  $\ln\big(\big\langle\exp\big(K_{N}\big)\big\rangle_{\!\beta\bb}\big)$ as a functional of $\R$-valued {\scshape Poisson} random variables $K_{N}: \sigma\mapsto K_{N}(\sigma)$ is convex by the {\scshape H\"older} inequality. The first inequality now follows from \eqref{F_N1}, \eqref{F_N2}, the {\scshape Jensen} inequality applied to the two-fold integration in \eqref{P_N}, a {\scshape Gauss}ian linearization using $w_{N}$, and an explicit calculation, confer the proof of Lemma\,\ref{comparison} below.
The first step in the proof of the second inequality is the same as in the proof of the first inequality. But then, instead of performing the {\scshape Gauss}ian linearization, we use
 \eqref{jensen4} with $P_{N}$ replaced by $\big[Q_{N}\big(\sigma(t),\sigma(t')\big)\big]^{2}$, combine this with \eqref{P_N_mean}, and finally apply again the {\scshape Jensen} inequality to the two-fold integration, but this time using the concavity of the logarithm.
We note that $\widetilde{G}_{N}$ (like $G_{N}$), for any $N\geq 2$, may be viewed to depend on $\beta\bb$ only via $p$ because the function $\beta\bb\mapsto p$ is bijective. 
\item The upper estimate in \eqref{f_ann_N_bounds} may be sharpened considerably by combining \eqref{F_N1} with \eqref{wolfbound} below. As $N\to\infty$ this becomes optimal for any $\lambda$ according to Theorem\,\ref{f_ann_variational}\,\ref{varadhan_gauss}.
\item In the simple case $N=2$ the quenched and annealed free energies can be calculated exactly by determining the four eigenvalues of the two-spin {\scshape Hamilton}ian\, $H_{2}$ with the results
\begin{align}
	\EE[\beta f_{2}]&=-\frac{1}{2}\EE\Big[\ln\Big(2\cosh\Big(\sqrt{(2\beta\bb)^{2}+2\lambda g_{12}^{2}}\Big)+2\cosh\Big(\sqrt{2\lambda}\,g_{12}\Big)\Big)\Big]\,,\label{f_2}\\
	\beta f_{2}^{\ann}&=-\frac{1}{2}\ln\Big(\EE\Big[2\cosh\Big(\sqrt{(2\beta\bb)^{2}+2\lambda g_{12}^{2}}\Big)\Big]+2\e^{\lambda} \Big)\,.\label{f_ann_2}
\end{align}
For the remaining {\scshape Gauss}ian integrations over the real line there are explicit ``closed-form expressions'' available only in the limiting cases $\beta \bb=0$ and/or $\lambda=0$.

\end{enumerate}
\end{remark}
The basis for our somewhat weak result mentioned near the end of Section\,\ref{introduction} is

\begin{theorem}[On the difference between the quenched and annealed free energies]\label{thm_diff_q_a}

\begin{enumerate}
\item For any $\lambda>0$ we have the crude bounds
	\begin{equation}\label{bounds_q_a}
		0\leq\EE\ [\beta f_{N}]-\beta f^{\ann}_{N}\leq\frac{G_{N}}{N}-\frac{\lambda}{N}\leq\lambda\frac{N-1}{N}\,.
	\end{equation}

\item For any  $\lambda>0$ we also have the lower bound
	\begin{equation}\label{lower_bound_q_a}
			k(\lambda)-\frac{\lambda}{N}-\ln\big(\cosh(\beta\bb)\big)\leq\EE[\beta f_{N}]-\beta f^{\ann}_{N}\,,
	\end{equation}
	where 
	\begin{equation}\label{k}
		k(\lambda)\ceq\max_{q\in{\lbrack0,1\rbrack}}\Big(\lambda\big(1-(1-q)^{2}\big)-\EE\big[\ln\big(\cosh(g_{12}\sqrt{4\lambda q})\big)\big]\Big)\,.
	\end{equation}
\item \label{Thm_cond_k} A simple condition implying strict positivity of the lower bound in \eqref{lower_bound_q_a} is
	\begin{equation}\label{cond_k}
		\beta\bb\min\Big\{\frac{\beta\bb}{2},1\Big\} <\lambda\Big(\frac{N-1}{N}-\sqrt{\frac{8}{\pi\lambda}}\Big)\,. 
	\end{equation}
\end{enumerate}
\end{theorem}
The proof of Theorem\,\ref{thm_diff_q_a}, given below, is based on the lower estimate in \eqref{f_ann_N_bounds}, the estimate \eqref{f_ann_N_bound2}, certain quasi-classical estimates for the (random) free energy $\beta f_{N}\equiv \beta f_{N}(\beta\bb,\beta\vv)$, divided by the temperature, and on the (so-called replica-symmetric) \emph{{\scshape SK} approximation} $k(\lambda)-\lambda-\ln(2)$ to the macroscopic quenched free energy of the (zero-field) {\scshape SK} model \cite{SK1975}. This approximation provides a lower bound on $\EE\big[\beta f_{N}(0,\beta\vv)\big]$ for all $\beta\vv$ and $N\geq 2$ according to {\scshape Guerra}, see \cite[Ineq.\,(5.7)]{G2001}, \cite{G2003}, and also \cite[Thm.\,1.3.7]{T2011a}. Since the quasi-classical estimates are of some independent interest, we firstly compile them in

\begin{lemma}[Quasi-classical estimates for the free energy]\label{quasi_classical_lemma}
	\begin{equation}\label{quasi_classical}
	\beta f_{N}(0,\beta\vv)\leq\beta f_{N}(\beta\bb,\beta\vv)+\ln\big(\cosh(\beta\bb)\big)\leq\beta f_{N}(0,\beta p\vv)\leq-\ln(2)\,.
	\end{equation}
	\begin{equation}\label{maximal}
	0\leq\beta f_{N}(0,\beta\vv)-\beta f_{N}(\beta\bb,\beta\vv)\leq\ln\big(\cosh(\beta\bb)\big)\leq \beta\bb\min\{\beta\bb/2,1\}\,.
	\end{equation}
\end{lemma}

\begin{proof}[Lemma\,\ref{quasi_classical_lemma}]
The first estimate in \eqref{quasi_classical} follows from using the {\scshape Jensen} inequality $\exp\big(\int_{0}^{1}\dl{t}(\dots)\big)\leq\int_{0}^{1}\dl{t}\exp(\dots)$ in the {\scshape PFK} representation \eqref{fk}. Alternatively, it may be viewed as an application of the {\scshape Golden--Thompson} inequality, confer Remark\,\ref{remthm1}\,\ref{Golden_Thompson}. Using in \eqref{fk} the {\scshape Jensen} inequality $\big\langle\exp(\dots)\big\rangle_{\beta\bb}\geq\exp\big(\langle\dots\rangle_{\beta\bb}\big)$ combined with \eqref{mu} 
and \eqref{p} yields the second estimate in \eqref{quasi_classical}. The last estimate in \eqref{quasi_classical} is due to the operator inequality $\exp(-\beta H_{N})\geq 1-\beta H_{N}$ and $\tr H_{N}=0$ for arbitrary $\bb$ and $\vv$.
The first two estimates in \eqref{maximal} follow by combining the first estimate in \eqref{quasi_classical} with the random version underlying \eqref{f_ann_N_bound2}, see its proof and also Remark\,\ref{remthm1}\,\ref{Golden_Thompson}. The last estimate in \eqref{maximal} is due to the elementary inequalities $\cosh(y)\leq\exp\big(\min\{y^{2}/2,|y|\}\big)$ for $y\in\R$. Here the second inequality is obvious and the first one follows by comparing the two associated {\scshape Taylor} series' termwise and using $n!2^{n}\leq(2n)!$ for $n\in\N$.
\qed
\end{proof}

\begin{remark}\label{remcor1}
\renewcommand{\labelenumi}{(\roman{enumi})}
\renewcommand{\theenumi}{(\roman{enumi})}
\begin{enumerate}
	\item The estimates \eqref{quasi_classical} are quasi-classical, because the quantum fluctuations lurking behind the anti-commutativity $S^{\z}_{i}S^{\x}_{i}=-S^{\x}_{i}S^{\z}_{i}$ , equivalently behind the randomness of the {\scshape Poisson} process ${\cal N}_{i}$, are neglected by the first estimate and taken somewhat into account by the second one in terms of an effective disorder parameter $\beta p\vv\leq\beta\vv$. The last estimate in \eqref{quasi_classical} corresponds to the limiting cases $\vv\downarrow0$ and $\bb\to\infty$.
	The estimates \eqref{maximal} control in a simple way the influence of the transverse magnetic field on the values attainable by the free energy.
	\item\label{special} The estimates \eqref{quasi_classical} imply especially $\beta f_{N}+\ln\big(2\cosh(\beta\bb)\big)\leq 0$. This estimate can be derived, without $\beta f_{N}(0,\beta p\vv)$, directly from \eqref{fk} by using $\exp(y)\geq 1+y$ for $y\in\R$ and $\sum_{s}h_{N}\big(s\sigma(t)\big)=0$. Alternatively, it may be viewed as an application of the {\scshape Jensen--Peierls--Bogolyubov} inequality, similarly as in Remark\,\ref{remthm1}\,\ref{Golden_Thompson}.\label{Golden_Thompson3}
\end{enumerate}
\end{remark}
\begin{proof}[Theorem\,\ref{thm_diff_q_a}]
\begin{enumerate}
	\item The first bound in \eqref{bounds_q_a} is \eqref{jensen}. For the second bound in \eqref{bounds_q_a} we take the disorder average of the second estimate in \eqref{quasi_classical} and combine the result with the lower estimate in \eqref{f_ann_N_bounds}. This gives
	\begin{equation}
		\EE[\beta f_{N}]-\beta f^{\ann}_{N}\leq\frac{G_{N}}{N}-\frac{\lambda}{N}+\EE\big[\beta f_{N}(0,\beta p\vv)\big]+\ln(2)\,.
	\end{equation}
	Finally, we use the disorder average of the last (crude) estimate in \eqref{quasi_classical}
	\begin{equation}\label{crude_bound}
		\EE\big[\beta f_{N}(0,\beta p\vv)\big]\leq-\ln(2)\,.
	\end{equation}
	The simplified last bound in \eqref{bounds_q_a} is due to \eqref{G_N_strong}.
\item We begin by claiming the lower estimate
\begin{equation}\label{Guerra}
	\EE[\beta f_{N}]\geq k(\lambda)-\lambda-\ln\big(2\cosh(\beta\bb)\big)\,.
\end{equation}
It simply follows by combining the disorder average of the second estimate in \eqref{maximal} with {\scshape Guerra}'s lower bound on $\EE\big[\beta f_{N}(0,\beta\vv)\big]$, that is, with \eqref{Guerra} for $\bb=0$. The claim \eqref{lower_bound_q_a} now follows by combining \eqref{Guerra} with \eqref{f_ann_N_bound2}.

\item We have $k(\lambda)\geq \lambda-\EE\big[\ln\big(\cosh(g_{12}\sqrt{4\lambda})\big)\big]\geq \lambda-\sqrt{8\lambda/\pi}$. Here, the first inequality follows from restricting the maximization in \eqref{k} to $q=1$. The second inequality follows from $\ln(\cosh(y))\leq |y|$ for $y=g_{12}\sqrt{4\lambda}$ combined with $\EE[|g_{12}|]=\sqrt{2/\pi}$. The last estimate in \eqref{maximal} now completes the proof.
\qed
 
\end{enumerate}
\end{proof}
\begin{remark}\label{rem_quasiclassical}
\renewcommand{\labelenumi}{(\roman{enumi})}
\renewcommand{\theenumi}{(\roman{enumi})}

\begin{enumerate}
	\item The upper bound in \eqref{bounds_q_a} has an underlying random version, namely $\beta f_{N}-\beta f^{\ann}_{N}\leq(G_{N}-\lambda)/N$. It simply follows by combining the lower estimate in \eqref{f_ann_N_bounds} with the estimate in Remark\,\ref{remcor1}\,\ref{Golden_Thompson3}.
	\item \label{rem_k}We recall some more or less well-known properties of the function $k: \lambda\mapsto k(\lambda)$.
	First, we have the bounds $\max\{0,\lambda-\sqrt{8\lambda/\pi}\}\leq k(\lambda)\leq \lambda$. The lower bound $0$ follows from restricting the maximization in \eqref{k} to $q=0$, the other lower bound has just been derived in the proof of Theorem\,\ref{thm_diff_q_a}\,\ref{Thm_cond_k}, and the upper bound follows from using $\cosh(y)\geq1$ in \eqref{k}.
	Second, we have the equivalence $4\lambda\leq 1 \Leftrightarrow k(\lambda)=0$. It follows by considering the first two derivatives with respect to $q$ of the function to be maximized in \eqref{k}. These can be studied via {\scshape Gauss}ian integration by parts.
	Third, for $4\lambda>1$ the function $k$ is strictly increasing according to its first derivative $k'(\lambda)=\big(q(\lambda)\big)^{2}\geq 0$ where $q(\lambda)$ is the unique strictly positive solution of $q=\EE\big[\big(\tanh(g_{12}\sqrt{4\lambda q})\big)^{2}\big]$, sometimes called the (zero-field) {\scshape SK} equation.
	\item Obviously, any sharpening of \eqref{crude_bound} or \eqref{Guerra} improves the bound in \eqref{bounds_q_a} or \eqref{lower_bound_q_a}, respectively. 
\end{enumerate}

\end{remark}


\section{The topics of Section\,\ref{sec_annealed_free_energy} in the macroscopic limit}\label{sec_macroscopic_annealed_free_energy}
From now on we are mainly interested in the macroscopic limit $N\to\infty$. The next theorem is the main result of this section and analogous to Theorem \,\ref{f_ann_bounds}.

\begin{theorem}[On the macroscopic annealed free energy]\label{f_ann_exist}
\begin{enumerate}
	\item For any $\lambda> 0$  the macroscopic limit of the annealed free energy exists, is given by
	\begin{equation}
	f_{\infty}^{\ann}\ceq\lim\limits_{N\to\infty}f^{\ann}_{N}=\sup\limits_{N\geq 2}\Big(f^{\ann}_N-\frac{\lambda}{N\beta}\Big)\,, \label{f_ann_inf}
	\end{equation}
	and obeys the three estimates
	\begin{align}
	-\inf_{N\geq2}\frac{G_{N}}{N}\leq &\beta f^{\ann}_\infty+\ln\big(2\cosh(\beta\bb)\big)\leq -p\,\lambda\,,\label{f_ann_inf_bounds}\\
	&\beta f^{\ann}_{\infty}\leq -\lambda-\ln(2)\,. \label{f_ann_inf_bound2}
	\end{align}

	\item The dimensionless limit $\beta f^{\ann}_{\infty}$  depends on the disorder strength $\vv$ only via the dimensionless variable $\lambda>0$. The function $\lambda\mapsto \beta f^{\ann}_{\infty}$ is concave, is not increasing,
	and has the following weak- and strong-disorder limits
	\begin{align}
		\lim\limits_{\lambda\downarrow 0}\Big(\frac{1}{\lambda}\big(\beta f^{\ann}_{\infty}+\ln\big(2\cosh(\beta\bb)\big)\Big)&=-p\,,\label{weak_disorder}\\
		\lim\limits_{\lambda\to\infty}\frac{1}{\lambda}\beta f^{\ann}_{\infty}&=-1\label{strong_disorder}\,.
	\end{align}
	
	\item\label{strong_field} The difference between the (macroscopic) annealed free energy and the ideal paramagnetic free energy vanishes for any $\vv>0$ in the high-field limit according to  
	\begin{equation}\label{high_field}
	                        \lim\limits_{\bb\to\infty}\big(\beta f^{\ann}_{\infty}+\ln\big(2\cosh(\beta\bb)\big)=0\,.
	\end{equation}                       
	\item\label{f_ann_inf_behavior2} The function $\beta\mapsto\beta f^{\ann}_{\infty}$ is concave for any $\vv>0$.
	\end{enumerate}
\end{theorem}

\begin{remark}\label{rem_inf_G_N}
\renewcommand{\labelenumi}{(\roman{enumi})}
\renewcommand{\theenumi}{(\roman{enumi})}
\begin{enumerate}
	\item\label{inf_bounds} For the lower estimate in \eqref{f_ann_inf_bounds}, Lemma\,\ref{mpG_N_lemma} implies the bounds
	\begin{align}
		p\lambda<\ln\big(1+p(\e^{\lambda}-1)\big)&\leq\inf_{N\geq 2}G_{N}/N\leq G_{M}/M\label{inf_weak_lower}\\
	        &\leq p\lambda -(1-p)\lambda^{2}/2+(1-p)\lambda(M\lambda +2/M)/2\,, \label{inf_weak_upper}\\
                 \max\big\{0,\lambda+\ln(p)/2\big\}&\leq\inf_{N\geq2}G_{N}/N\leq G_{M}/M<\lambda=\lim_{N\to\infty}G_{N}/N\label{inf_strong}
	\end{align}
	 for any $\lambda>0$, $\beta\bb>0$, and natural number $M\geq2$. Another upper bound is
\begin{equation}\label{inf_p}
	 \inf_{N\geq 2}G_{N}/N\leq\ln(1+ (3/2)p\e^{2\lambda})/2\,,
\end{equation}
which follows from \eqref{inf_strong} by choosing $M\geq 2/p$ and observing $p\leq 1$. It is smaller than $\lambda$ if and only if $p<(2/3)(1-\e^{-2\lambda})$.
\item The infimum over $M\geq2$ in \eqref{inf_weak_upper} is attained and depends on $\lambda$. The smaller $\lambda$ is, the larger is the minimizing $M$. In particular, if $\lambda<1/3$ then, and only then, the minimizing $M$ is larger than 2.
\item Returning to $\inf_{N\geq2}G_{N}/N $ itself, the minimizing $N\geq2$  depends on $\lambda$ and, via $p$, also on $\beta\bb$. In the weak- and strong-disorder limits we have
\begin{equation}\label{g_limits}
	\lim_{\lambda\downarrow0}g(\lambda)=p\,,\qquad\qquad\lim_{\lambda\to\infty}g(\lambda)=1
\end{equation}
for $g(\lambda)\ceq\inf_{N\geq2}G_{N}/(\lambda N)$. The strong-disorder limit is obvious from \eqref{inf_strong}. For the proof of the weak-disorder limit we use \eqref{inf_weak_upper} to obtain $\limsup_{\lambda\downarrow0}g(\lambda)\leq p +(1-p)/M=p_{M}$. Taking now the infimum over $M\geq2$ and observing $p\leq g(\lambda)$ from \eqref{inf_weak_lower} completes the proof of \eqref{g_limits}.
For intermediate values of $\lambda$ in the sense that $2\lambda<\ln(1+2/p)$, equivalently  $G_{4}/4<G_{2}/2$, the infimum in \eqref{f_ann_inf_bounds} is attained at some $N\geq 3$.
A numerical approach suggests that the minimizer is $N=2$ for all $\lambda\geq1$ if $\beta\bb\leq1/2$.
In this context we recall from \eqref{F_N_bounds} the inequality $F_{2}\leq G_{2}$, for all $\lambda$ and $\beta\bb$, and from \eqref{f_ann_2} and \eqref{F_N1} that $F_{2}$ can be calculated exactly.
However, while $-F_{2}/2$ provides a sharper lower bound in \eqref{f_ann_inf_bounds} than $-G_{2}/2$, it is a less explicit function of $\lambda$ and $\beta\bb$.
\end{enumerate}
\end{remark}
\newpage
\begin{proof}[Theorem\,\ref{f_ann_exist}]
\begin{enumerate}
	\item For the proof of \eqref{f_ann_inf} we show that the sequence $(F_N)_{N\geq 2}$, as defined in \eqref{F_N1}, is sub-additive. Indeed, for two arbitrary natural numbers $N_{1},N_{2}\geq 2$ and classical spin configurations $s,\widehat{s}\in\{-1,1\}^{N_{1}+N_{2}}$  we have
	\begin{align}
		\big[Q_{N_{1}+N_{2}}(s,\widehat{s})\big]^{2}
		\leq\frac{N_{1}}{N_{1}+N_{2}}\Big(\frac{1}{N_{1}}\sum\limits_{i=1}^{N_{1}}s_{i}\widehat{s}_{i}\Big)^{2}
		+\frac{N_{2}}{N_{1}+N_{2}}\Big(\frac{1}{N_{2}}\sum\limits_{i=N_{1}+1}^{N_{1}+N_{2}}s_{i}\widehat{s}_{i}\Big)^{2}
	\end{align}
	by the convexity of the square and the {\scshape Jensen} inequality. By combining this with \eqref{F_N2}, \eqref{Phi_N3}, and \eqref{Q_N} and by using the independence of the involved {\scshape Poisson} processes we get the claimed sub-additivity, $F_{N_{1}+N_{2}}\leq F_{N_{1}}+F_{N_{2}}$. According to {\scshape Fekete} \cite{F1923} (see also \cite[Lem.\,10.21]{K2002}) this establishes the convergence result
	\begin{equation}
		\lim_{N\to\infty}\frac{F_{N}}{N}=\inf_{N\geq 2}\frac{F_{N}}{N}\geq p\lambda > 0\,,
	\end{equation}
	where we have used also \eqref{F_N_bounds} and \eqref{G_N}. By \eqref{F_N1} this gives the claim \eqref{f_ann_inf}.
	The estimates \eqref{f_ann_inf_bounds} and \eqref{f_ann_inf_bound2} follow from \eqref{f_ann_inf} applied to \eqref{f_ann_N_bounds} and \eqref{f_ann_N_bound2}, respectively.

	\item The first sentence follows from \eqref{f_ann_inf} and the corresponding one in Theorem\,\ref{f_ann_bounds}\,\ref{f_ann_N_behavior}. The function $\lambda\mapsto\beta f^{\ann}_{\infty}$ is concave and not increasing, because it is the pointwise limit of a family of such functions according to \eqref{f_ann_inf} and Theorem\,\ref{f_ann_bounds}\,\ref{f_ann_N_behavior}.
	The claim \eqref{weak_disorder} follows from \eqref{f_ann_inf_bounds} and \eqref{g_limits}. The claim \eqref{strong_disorder} follows from \eqref{f_ann_inf_bound2}, the lower estimate in \eqref{f_ann_inf_bounds}, and \eqref{g_limits}.

         \item This follows from \eqref{f_ann_inf_bounds}, \eqref{inf_p}, and $ \lim\limits_{\bb\to\infty} p=0$.
	 \item The function $\beta\mapsto\beta f^{\ann}_{\infty}$ is concave, because it is the pointwise limit of a family of such functions according to \eqref{f_ann_inf} and Theorem\,\ref{f_ann_bounds}\,\ref{f_ann_N_behavior2}.\qed
\end{enumerate}
\end{proof}

\begin{remark}\label{remthm2}
	\renewcommand{\labelenumi}{(\roman{enumi})}
	\renewcommand{\theenumi}{(\roman{enumi})}
	\begin{enumerate}

	\item The lower estimate in \eqref{f_ann_inf_bounds} becomes sharpened (for intermediate values of $\lambda$) when $G_{N}$ is replaced by $\widetilde{G}_{N}$. This is a consequence of Remark\,\ref{remthm1}\,\ref{rem_W_N}.
	\item\label{high} The high-field relation \eqref{high_field} complements the strong-disorder relation \eqref{strong_disorder}. A relation analogous to \eqref{high_field} also holds for the macroscopic quenched free energy $\lim_{N\to\infty}\EE[f_{N}]$. This follows by combining \eqref{high_field} with \eqref{jensen} and the disorder average of the estimate in Remark\,\ref{remcor1} \ref{special} in the limit $N\to\infty$.
	Here we rely on the fact that $\lim_{N\to\infty}\EE[f_{N}]$ exists for all $\lambda>0$. For the classical limit $\bb=0$ this has been shown by {\scshape Guerra} and {\scshape Toninelli} in their seminal paper \cite{GT2002}. Its extension to the quantum case $\bb>0$ is due to {\scshape Crawford} \cite{C2007} by building on \cite{GT2003}. For $\lambda<1/4$ this extension also follows from our main result $\Delta_{\infty}\ceq\lim_{N\to\infty}\big(\EE[f_{N}]-f^{\ann}_{N}\big)=0$ obtained by probabilistic arguments in Theorem\,\ref{difference} below.
	Returning to the high-field relation for the macroscopic quenched free energy, we note that it is consistent with the inequality
	\begin{equation}\label{eitheror}
		\lim_{N\to\infty}\EE[\beta f_{N}]\leq\min\big\{-\ln\big(2\cosh(\beta\bb)\big),\lim_{N\to\infty}\EE[\beta f_{N}(0,\beta\vv)]\big\}
	\end{equation}
	following from the disorder averages of \eqref{quasi_classical} and \eqref{maximal} for any $\bb>0$. Equality in \eqref{eitheror}, for given $\vv>0$, only holds in the limiting cases $\bb\downarrow 0$ and $\bb\rightarrow\infty$, as follows from \eqref{quasi_classical} and the strict concavity of $\lim_{N\to\infty}\EE[\beta f_{N}]$ in $\bb\in\R$.
	This contrasts the quantum random energy model (QREM) for which equality holds in the analog of \eqref{eitheror} for all $\bb$ (and $\vv$) according to {\scshape Goldschmidt} in \cite{G1990}. Although the QREM is much simpler than the quantum {\scshape SK} model \eqref{H_N}, a rigorous proof of his statement was achieved only recently by {\scshape Manai} and {\scshape Warzel} in \cite{MW2020}.
	\item In the limit $N\to\infty$ the bounds in Theorem\,\ref{thm_diff_q_a} take the form 
	\begin{equation}\label{diff_bound}
	\max \big\{0,\, k(\lambda)-\ln\big(\cosh(\beta\bb)\big)\big \}\leq\beta\Delta_{\infty}\leq\inf_{N\geq 2}\frac{G_{N}}{N}\,.
	\end{equation}
	The upper bound in \eqref{diff_bound} is due to the disorder average of \eqref{quasi_classical} and the lower estimate in \eqref{f_ann_inf_bounds}. At the expense of weakening this bound for intermediate values of $\lambda$, it can be made somewhat more explicit with the help of \eqref{inf_weak_upper} and \eqref{inf_strong} or \eqref{inf_p} for small and large $\lambda$, respectively. In any case, by \eqref{inf_weak_lower} the upper bound in \eqref{diff_bound} is not sharp enough to vanish for $\lambda<1/4$. But it vanishes for any $\lambda\in{]0,\infty[}$ in the high-field limit $ \bb\to\infty $ due to  \eqref{inf_p}.
	\item According to the monotonicity mentioned in Remark\,\ref{rem_quasiclassical}\,\ref{rem_k} the lower bound in \eqref{diff_bound} is strictly positive for sufficiently large $\lambda>1/4$, in particular, if $\beta\bb\min\{\beta\bb/2,1\} <\lambda-\sqrt{8\lambda/\pi}$, see \eqref{cond_k}.
	It follows from \eqref{diff_bound} that the difference $\Delta_{\infty}\geq0$ between the macroscopic quenched and annealed free energies is strictly positive for any pair $(\beta,\bb)\in{\rbrack0,\infty\lbrack}\times{\lbrack0,\infty\lbrack}$ provided that $\vv>0$ is sufficiently large. Physically more important is the situation of a fixed $\vv>0$. Then strict positivity holds for any $\bb\geq 0$ provided that $\beta>0$ is sufficiently large and, conversely, for any $\beta>1/\vv$ provided that $\bb>0$ is sufficiently small.
	Nevertheless, the lower bound in \eqref{diff_bound} is not sharp enough to characterize the (maximum) region with $\Delta_{\infty}>0$ in the $(1/\beta,\bb)$-plane for $\bb>0$,\footnote{The region characterized by $\Delta_{\infty}>0$ is larger than the region implied by \eqref{diff_bound}. For example, for $1/(\beta \vv)<1/\sqrt{4\ln(2)}= 0.60056\dots$ there is a region, where the ``annealed entropy'' $\beta^{2}\partial f^{\ann}_{\infty}/\partial\beta$ is negative as follows from Theorem\,\ref{f_ann_exist}\,\ref{f_ann_inf_behavior2}, the lower estimate in \eqref{f_ann_inf_bounds}, and \eqref{f_ann_inf_bound2}. This region does not completely belong to the region implied by \eqref{diff_bound}.} but it is so for $\bb=0$. The latter can be seen by combining the main result in \cite{ALR1987} (or Theorem\,\ref{difference}) with \eqref{Guerra} for $\bb=0$ and the equivalence in Remark\,\ref{rem_quasiclassical}\,\ref{rem_k}.
	These facts are illustrated in Figure\,\ref{regionplot}, where we also have included the result of Theorem\,\ref{difference} and a cartoon of the \emph{critical line}, that is, the border line between the spin-glass and the paramagnetic phase as obtained by approximate arguments and/or numerical methods, for example in \cite{FS1986,YI1987,US1987,GL1990,T2007,Y2017,MRC2018}.\label{phases}
\begin{figure}[H]
\begin{center}
	\psset{xunit=7cm,yunit=1cm}
	\begin{pspicture}(-0.4cm,-0.5cm)(10.6cm,3.3cm)
	\pscustom[fillstyle=solid,fillcolor=gray!25,linestyle=none]{
		\psline(0,0)(0,2.64983574029153)
		\psline(0,2.64983574029153)(0.075,2.64983574029153)
		\fileplot{region_data.csv}
		\psline(1.0,0.0)(0,0)}
	\pscustom[fillstyle=solid,fillcolor=gray!5,linestyle=none]{
		\fileplot{region_data.csv}
		\psline(1.0,0.0)(1.0,2.64983574029153)
		\psline(1.0,2.64983574029153)(0.075,2.64983574029153)
		}
	\pscustom[fillstyle=solid,fillcolor=gray!70,linestyle=none]{
		\psline(1.0,0.0)(1.0,2.64983574029153)(1.3,2.64983574029153)(1.3,0)(1.0,0)
		}
	\fileplot[plotstyle=line,linecolor=blue,linewidth=1pt]{region_data.csv}
	\psline[linecolor=blue,linestyle=solid,linewidth=1pt](1,0)(1,2.64983574029153)
	\psaxes[Dx=0.2]{->}(1.4,3)
	\rput(1.465,0){$1/(\beta \vv) $}
	\rput(0,3.2){$\bb/\vv$}
	\psset{algebraic,plotpoints=500}
	\psplot[linestyle=dashed,linecolor=red]{0}{1}{1.51*sqrt(1-x^2)}
	\rput(0.10,0.45){ $\Delta_{\infty}>0$}
	\rput(0.55,2){\large ?}
	\rput(1.15,1.3){$\Delta_{\infty}=0$}
	\textcolor{red}{
		\rput(0.178,1){spin glass}
		\rput(0.965,2){paramagnet}
	}
	\end{pspicture}
\caption{
	In the temperature-field quarter-plane there is one region where the difference $\Delta_{\infty}=\lim_{N\to\infty}\big(\EE[f_{N}]-f^{\ann}_{N}\big)\geq0$ is strictly positive (light gray) according to \eqref{diff_bound} and another one where it is zero (heavy gray) according to Theorem\,\ref{difference}. 
	The (red) dashed line is a cartoon of the critical line between the spin-glass and the paramagnetic phase as obtained by approximate arguments and/or numerical methods, see Remark\,\ref{remthm2}\,\ref{phases}.
	The region with $\Delta_{\infty}>0$ is larger than the light gray region, but we do not know how large. It should at least contain the critical line. 
}
\label{regionplot}
\end{center}
\end{figure}
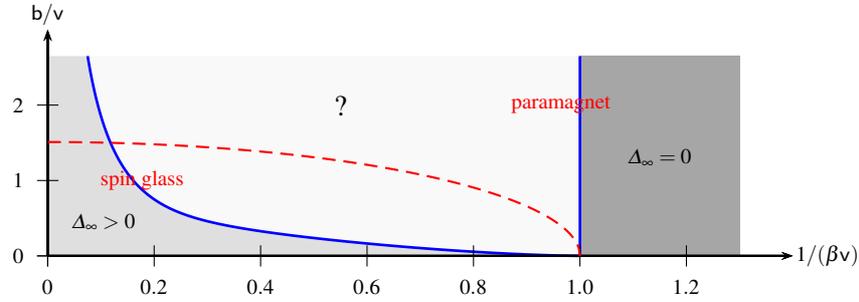
	 \item For any $\bb\geq0$ the lower bound in \eqref{diff_bound} is good enough to coincide with the upper one for asymptotically large $\lambda$ in the sense that
	$\lim_{\lambda\to\infty}\beta\Delta_{\infty}/\lambda=1$. It follows from \eqref{g_limits} and $\lim_{\lambda\to\infty}k(\lambda)/\lambda=1$ according to the simple bounds in Remark\,\ref{rem_quasiclassical}\,\ref{rem_k}. Combining this observation with \eqref{strong_disorder} yields $\lim_{\lambda\to\infty}\lim_{N\to\infty}\EE[\beta f_{N}]/\lambda=0$, reflecting the finiteness of the (specific) macroscopic quenched ground-state energy which, in its turn, directly follows from \eqref{Guerra} and Remark\,\ref{rem_quasiclassical}\,\ref{rem_k}. Finally we note that the lower bound in \eqref{diff_bound} also implies $\lim_{\beta\to\infty}\Delta_{\infty}=\infty$ in agreement with \eqref{f_ann_inf_bound2}.
	\end{enumerate}
\end{remark}

\section{A dual pair of variational formulas for the macroscopic annealed free energy}

In the last section we have seen that the macroscopic annealed free energy $f^{\ann}_{\infty}$ exists and obeys explicitly given lower and upper bounds which become sharp in the limits of weak and strong disorder, $\lambda\downarrow 0$ and $\lambda\to\infty$, respectively. Furthermore, the bounds in \eqref{f_ann_inf_bounds} coincide asymptotically also in the limits of low and high field, that is, $\bb\downarrow 0$ and $\bb\rightarrow\infty$. However, so far we have no formula which makes the $\lambda$-dependence of $\beta f^{\ann}_{\infty}$ more ``transparent'' for general $\lambda>0$ and $\bb>0$. This will be achieved, to some extent, in the present section. More precisely, we will show that $\beta f^{\ann}_{\infty}$ may be viewed as the global minimum of a non-linear functional with a simple $\lambda$-dependence and defined on the {\scshape Hilbert} space of real-valued functions being square-integrable over the unit square. This follows from an asymptotic evaluation of the right-hand side of \eqref{F_N2} as $N\to\infty$ by large-deviation techniques due to {\scshape Varadhan} \cite{V1966} and others, see \cite{DS1989,DZ2010} and also \cite{K2002}.

For the formulation of the corresponding theorem we need some preparations. We begin by introducing some further notation. We consider the (separable) real {\scshape Hilbert} space $\L^{2}\ceq\L^{2}\big({\lbrack0,1\rbrack}\times{\lbrack0,1\rbrack}\big)\cong\L^{2}\big({\lbrack0,1\rbrack}\big)\otimes\L^{2}\big({\lbrack0,1\rbrack}\big)$ of all {\scshape Lebesgue} square-integrable functions $\psi:{\lbrack0,1\rbrack}\times{\lbrack0,1\rbrack}\to\R$, $(t,t')\mapsto\psi(t,t')$ with scalar product $\langle\psi,\varphi\rangle\ceq \int_{0}^{1}\dl{t}\int_{0}^{1}\dl{t'}\psi(t,t')\varphi(t,t')$ and norm $\|\psi\|\ceq\langle\psi,\psi\rangle^{1/2}$ for $\psi,\varphi\in\L^{2}$. If $\psi\in\L^{2}$, then obviously its absolute value $|\psi|$ also belongs to $\L^{2}$, where $|\psi|(t,t')\ceq |\psi(t,t')|$ (for almost all $(t,t')\in{\lbrack0,1\rbrack}\times{\lbrack0,1\rbrack}$ with respect to the two-dimensional {\scshape Lebesgue} measure). The tensor product ${\sigma_{i}\otimes\sigma_{j}}$ of two random functions \eqref{teleproc} is defined pointwise by ${(\sigma_{i}\otimes\sigma_{j})(t,t')}\ceq\sigma_{i}(t)\sigma_{j}(t')$ so that ${\sigma_{i}\otimes\sigma_{j}}\in\L^{2}$, obviously.
In particular, we consider the sequence ${(\sigma_{i}\otimes\sigma_{i})_{i\geq 1}}$ of independent and identically distributed $\L^{2}$-valued random variables and its empirical (or sample) averages $\xi_{N}\ceq\sum_{i=1}^{N}{\sigma_{i}\otimes\sigma_{i}}/N$ with mean $\langle\xi_{N}\rangle_{\beta\bb}=\mu$ and variance $\langle\xi_{N}^{2}\rangle_{\beta\bb}-\mu^{2}=(1-\mu^{2})/N $ for all $N\in\N$. Finally, we introduce the non-linear functional $\Lambda:\L^{2}\to\R$, $\psi\mapsto\Lambda(\psi)$, generating the cumulants (or logarithmic moments) of $ \xi_{1}= {\sigma_{1}\otimes\sigma_{1}}$, by
\begin{equation}\label{cumulant}
	\Lambda(\psi)\ceq\ln\Big(\big\langle\exp\big(\langle\psi,\xi_{1}\rangle\big)\big\rangle_{\beta\bb}\Big)
\end{equation}
and its {\scshape Legendre}--{\scshape Fenchel} transform $\Lambda^{\!*}:\L^{2}\to\R\cup\{\infty\}$, $\varphi\mapsto\Lambda^{\!*}(\varphi)$ by
\begin{equation}\label{transform}
\Lambda^{\!*}(\varphi)\ceq\sup\limits_{\psi\in\L^{2}}\big(\langle\psi,\varphi\rangle-\Lambda(\psi)\big)
\end{equation}
with its effective domain $\D^{*}\ceq\{\varphi\in\L^{2}:\Lambda^{*}(\varphi)<\infty\}$. We stress that $\Lambda$, $\Lambda^{*}$, and $\D^{*}$ depend on $\beta\bb>0$, equivalently on $p$. But we suppress this dependence to simplify the notation, as we have done with $\mu$, $m$, and $p$ introduced in \eqref{mu}, \eqref{m}, and \eqref{p}.

\begin{lemma}[Some properties of the functionals $\Lambda$ and $\Lambda^{*}$] \label{some_properties}

\noindent
For any $\psi,\varphi\in\L^{2}$, with $\mu\in\L^{2}$ as defined in  \eqref{mu}, and with $1\in\L^{2}$ denoting the (constant) unit function we have
\begin{enumerate}
	\item for $\Lambda$ the inequalites
	\begin{equation}\label{cumulant_bounds}
		-\infty<\langle\psi,\mu\rangle\leq\Lambda(\psi)\leq\Lambda(|\psi|)\leq\langle|\psi|,1\rangle\leq\|\psi\|<\infty\,,
	\end{equation}
	\begin{equation}\label{cumulant_bound2}
		\langle\psi,1\rangle-\ln\big(\cosh(\beta\bb)\big)\leq\Lambda(\psi)\,,
	\end{equation}
	\item for $\Lambda^{*}$ the equality and inequalities
	\begin{equation}\label{transform_bounds}
		0=\Lambda^{*}(\mu)\leq\Lambda^{*}(\varphi)\,,\qquad\Lambda^{*}(1)\leq \ln\big(\cosh(\beta\bb)\big)\,.
	\end{equation}
\end{enumerate}
Moreover, we have the properties:
\begin{enumerate}
\setcounter{enumi}{2}
	\item \label{Lambda_convexity}The functional $\Lambda$ is convex which is reflected by the inequalities
	\begin{equation}\label{jensen3}
        \langle\Lambda'(\psi),\varphi-\psi\rangle\leq \Lambda(\varphi)-\Lambda(\psi)\leq\langle\Lambda'(\varphi),\varphi-\psi\rangle\,.
	\end{equation}

	Here, the non-linear mapping $\Lambda': \L^{2}\to\L^{2}, \psi\mapsto\Lambda'(\psi)$, is defined by the $\L^{2}$-function 
	\begin{equation}\label{Lambda'}
	\Lambda'(\psi)\ceq\e^{-\Lambda(\psi)}\big\langle\e^{\langle\psi,\xi_{1}\rangle}\xi_{1}\big\rangle_{\beta\bb}\eqc\big\langle\xi_{1}\big\rangle_{\beta\bb,{\psi}}\,\,,\qquad  |\Lambda'(\psi)|\leq 1\,.
	\end{equation}
	This function is in its argument $(t,t')$ continuous and exchange symmetric ($t\leftrightarrow t'$). Moreover, for $\psi\geq 0$ (pointwise) it satisfies the equality and inequalities
	\begin{equation}\label{Lambda'_bound}
			\mu=\Lambda'(0)\leq\Lambda'(\psi)\qquad (\psi\geq0)\,,
	\end{equation}
	\begin{equation}\label{Lambda'_bounds}
	0<p\leq\langle\Lambda'(\psi),\mu\rangle\leq m \leq  \langle\Lambda'(\psi),1\rangle\leq 1\qquad (\psi\geq0)\,.
	\end{equation}
	 \item \label{gateaux_diff} The functional $\Lambda$ is {\scshape Lipschitz} continuous with constant $1$,
	\begin{equation} \label{Lambda_continuity}
		|\Lambda(\psi)-\Lambda(\varphi)|\leq \langle|\psi-\varphi|,1\rangle\leq\|\psi-\varphi\|\,,
	\end{equation}
	 and also weakly sequentially continuous, that is, sequentially continuous with respect to the weak topology on $\L^{2}$. Furthermore, the functional $\Lambda$ has at any $\psi\in \L^{2}$ the linear and continuous {\scshape G\^{a}teaux} differential $\L^{2}\to\R, \varphi\mapsto\langle\Lambda'(\psi),\varphi\rangle$ because
	\begin{equation}\label{gateaux}
         \frac{\dl{}}{\dl{a}} \Lambda(\psi+a\varphi)\big|_{a=0}=\langle\Lambda'(\psi),\varphi\rangle\qquad (a\in\R)\,.
	\end{equation}
	\item The mapping $\Lambda'$ is {\scshape Lipschitz} continuous with constant $1$,
	\begin{equation}\label{Lambda'_continuity}
		\|\Lambda'(\psi)-\Lambda'(\varphi)\|\leq \|\psi-\varphi\|\,.
	\end{equation}
	\item The finiteness $\Lambda^{*}(\varphi)<\infty$, equivalently $\varphi\in\D^{*}$, implies $|\varphi|\leq1$, $0<\langle1,\varphi\rangle$, and $0\leq \langle\rho\otimes\rho,\varphi\rangle$ for all $\rho\in\L^{2}\big({\lbrack0,1\rbrack}\big)$. In particular, $\Lambda^{*}(0)=\infty$.
	\label{phibound}

	\item The functional $\Lambda^{*}$ is convex and so is its non-empty set $\D^{*}$. Furthermore, $\Lambda^{*}$ is weakly lower semi-continuous and its lower-level sets $\D^{*}_{r}\ceq\{\varphi\in\L^{2}:\Lambda^{*}(\varphi)\leq r\}$ for $r\in{\lbrack0,\infty\lbrack}$\,, are non-empty, convex, weakly sequentially compact, and weakly compact.
	\label{Lambda*_prop}
\end{enumerate} 

\end{lemma}

\begin{proof}
\begin{enumerate}
	\item The first and last inequality in \eqref{cumulant_bounds} are obvious. The second one is the {\scshape Jensen} inequality combined with \eqref{mu}, the fourth one follows from $\langle|\psi|,\xi_{1}\rangle\leq \langle|\psi|,1\rangle$, and the fifth one is the {\scshape Schwarz} inequality. The third inequality follows from the {\scshape Taylor} series of the exponential in \eqref{cumulant} by estimating termwise according to $\big\langle\langle\psi,\xi_{1}\rangle^{n}\big\rangle_{\beta\bb}\leq\big\langle\langle|\psi|,\xi_{1}\rangle^{n}\big\rangle_{\beta\bb}$ for $n\in\N$. This estimate is due to $\psi\leq|\psi|$ and the positivity implied by the inequalities
	\begin{equation}\label{multipoint_inequality}
		\Big\langle\prod_{k\in I}\sigma_{1}(t_{k})\Big\rangle_{\beta\bb}\geq\mu(t_{i},t_{j})\,\,\Big\langle\prod\limits_{k\in I\setminus\{i,j\}}\sigma_{1}(t_{k})\Big\rangle_{\beta\bb}\,\geq0\,.
	\end{equation}
	They are proved within a more general setting in Appendix\,\ref{poisson_appendix}, see \eqref{cond_positivity}. Here, $t_{1},\dots,t_{2n}$ denote $2n$ arbitrary points of the time interval ${\lbrack0,1\rbrack}$ and $i,j$ denote two arbitrary elements of the index set $I\ceq\{1,\dots,2n\}$. 
	For the proof of \eqref{cumulant_bound2} we restrict the {\scshape Poisson} expectation in definition \eqref{cumulant} to the single realization without any spin flip in ${\lbrack0,1\rbrack}$ by inserting $1(\sigma_{1})$, confer the proof of \eqref{f_ann_N_bound2}. By this we get $\Lambda(\psi)\geq\langle\psi,1\rangle+\ln\big(\langle 1(\sigma_{1})\rangle_{\beta\bb}\big)=\langle\psi,1\rangle-\ln\big(\cosh(\beta\bb)\big)$.
	\item By definitions \eqref{transform} and \eqref{cumulant} we have $\Lambda^{*}(\varphi)\geq \langle0,\varphi\rangle-\Lambda(0)=0$ for all $\varphi$, in particular,  $\Lambda^{*}(\mu)\geq 0$. On the other hand, \eqref{transform} also gives $\Lambda^{*}(\mu)\leq0 $ by $\Lambda(\psi)\geq\langle\psi,\mu\rangle$ from \eqref{cumulant_bounds}. The second inequality in \eqref{transform_bounds} follows from using \eqref{cumulant_bound2} in \eqref{transform}.
	\item The functional $\Lambda$ is convex by the {\scshape H\"older} inequality.
	The first inequality in \eqref{jensen3} follows from the {\scshape Jensen} inequality with respect to the expectation $\langle\,(\cdot)\,\rangle_{\beta\bb, \psi}$  and the {\scshape Fubini--Tonelli} theorem. The second inequality then  follows from interchanging $\psi$ and $\varphi$.
	The exchange symmetry ($t\leftrightarrow t'$) of the function $(t,t')\mapsto\big(\Lambda'(\psi)\big)(t,t')=\langle\xi_{1}(t,t')\rangle_{\beta\bb,\psi}=\big\langle\sigma_{1}(t)\sigma_{1}(t')\big\rangle_{\beta\bb,\psi}$ is obvious.
	The proof of the continuity of this function is postponed until the proof of \ref{proof_Lipschitz}. 
	The estimate in \eqref{Lambda'} is due to the triangle inequality  $|\langle\,(\cdot)\,\rangle_{\beta\bb,\psi}|\leq \langle|\,(\cdot)\,|\rangle_{\beta\bb,\psi}$ combined with $|\xi_{1}|=1$. It ensures that $\Lambda'(\psi)\in\L^{2}$.
	The (pointwise) equality in \eqref{Lambda'_bound} is obvious. The inequality there follows from the {\scshape Taylor} series of the exponential under the expectation in \eqref{Lambda'} and using \eqref{multipoint_inequality} with $n+1$ instead of $n$.
	The inequalities \eqref{Lambda'_bounds} are immediate consequences of the estimates $0<\mu\leq\Lambda'(\psi)\leq 1$, see \eqref{Lambda'} and \eqref{Lambda'_bound}, and the definitions \eqref{m} and \eqref{p}. 
	\item \label{proof_Lipschitz} Inequality \eqref{Lambda_continuity} follows from \eqref{jensen3} and the inequality in \eqref{Lambda'}.
	Moreover, for any sequence $(\varphi_{n})_{n\geq 1}\subset\L^{2}$ weakly converging to some $\psi\in\L^{2}$ we have $a\ceq\sup_{n\geq1}\|\varphi_{n}\|<\infty$ as a consequence of the {\scshape Banach}--{\scshape Steinhaus} theorem, see \cite[Lem.\,2.46 and Lem.\,2.22]{BC2017}. Therefore we get $\langle\psi, \xi_{1}\rangle=\lim_{n\to\infty}\langle\varphi_{n},\xi_{1}\rangle$ and $\langle\varphi_{n},\xi_{1}\rangle\leq\|\varphi_{n}\|\leq a$ for all realizations of the underlying {\scshape Poisson} process ${\cal N}_{1}$. The claimed weak sequential continuity $\Lambda(\psi)=\lim_{n\to\infty}\Lambda(\varphi_{n})$ now follows from the {\scshape Lebesgue} dominated-convergence theorem with respect to the {\scshape Poisson} expectation.
	For the proof of the (global) {\scshape G\^{a}teaux} differentiability we replace $\varphi$ in \eqref{jensen3}  by $\psi+a\varphi$ with	$a\in{\rbrack0,\infty\lbrack}$ and get $\langle\Lambda'(\psi),\varphi\rangle\leq\big(\Lambda(\psi+a\varphi)-\Lambda(\psi)\big)/a\leq\langle\Lambda'(\psi+a\varphi),\varphi\rangle$. The proof of \eqref{gateaux} is now completed by observing that $\lim_{a\downarrow 0}\langle\Lambda'(\psi+a\varphi),\eta\rangle=\langle\Lambda'(\psi),\eta\rangle$ for all $\varphi,\eta\in\L^{2}$. The latter follows from \eqref{Lambda'}, \eqref{Lambda_continuity}, and by dominated convergence with respect to the {\scshape Poisson} expectation and to the integration underlying the scalar product of $\L^{2}$. For the postponed proof of the continuity of $(t,t')\mapsto\eta(t,t')\ceq\big\langle\xi_{1}(t,t')\exp\big(\langle\psi,\xi_{1}\rangle\big)\big\rangle_{\beta\bb}$ for given $\psi\in\L^{2}$ we apply the {\scshape Schwarz} inequality to the {\scshape Poisson} expectation and obtain for the time being
	\begin{equation}\label{continuity_t_t'}
		\big|\eta(t+u,t'+u')-\eta(t,t')\big|^{2}\leq2\,\e^{\Lambda(2\psi)}\big(1-\big\langle\xi_{1}(t+u,t'+u')\xi_{1}(t,t')\big\rangle_{\beta\bb}\big)\,.
	\end{equation}
	The last {\scshape Poisson} expectation with $t,t'\in{\lbrack0,1\rbrack}$ and $t+u,t'+u'\in{\lbrack0,1\rbrack}$ is bounded from below by $\mu(t+u,t)\mu(t'+u',t')$ according to \eqref{multipoint_inequality} with $n=2$.
	 Thus, the left-hand side of \eqref{continuity_t_t'} tends to zero for all $t,t'\in{\lbrack0,1\rbrack}$ as $(u,u')$ tends to $(0,0)$. This implies the continuity of the  $\L^{2}$-function   $\Lambda'(\psi)=\eta\exp\big(-\Lambda(\psi)\big)$ on the unit square.
	\item The claimed inequality \eqref{Lambda'_continuity} is equivalent to 
	\begin{equation}\label{Lambda'_continuity2}
		\big\langle\Lambda'(\psi)-\Lambda'(\varphi),\eta\big\rangle\leq \|\psi-\varphi\|\|\eta\|\qquad(\psi,\varphi,\eta\in\L^{2})\,.
	\end{equation}
	In fact, \eqref{Lambda'_continuity} follows from \eqref{Lambda'_continuity2} for $\eta=\Lambda'(\psi)-\Lambda'(\varphi)$. Conversely, \eqref{Lambda'_continuity2} follows from \eqref{Lambda'_continuity} by the {\scshape Schwarz} inequality. For the proof of \eqref{Lambda'_continuity2} we write
	\begin{equation}
	\big\langle\Lambda'(\psi)-\Lambda'(\varphi),\eta\big\rangle=\int_{0}^{1}\dl{a}\,\,\frac{\dl{}}{\dl{a}}\big\langle\Lambda'\big(\varphi+a(\psi-\varphi)\big),\eta\big\rangle=\int_{0}^{1}\dl{a}\,\,\frac{\dl{}}{\dl{a}}\big\langle\langle\xi_{1},\eta\rangle\big\rangle_{a}
	\end{equation}
	using the $a$-expectation $\langle\,(\cdot)\,\rangle_{\!a}\ceq\langle\,(\cdot)\,\rangle_{\beta\bb,\varphi+a(\psi-\varphi)}$, see \eqref{Lambda'}. The integrand turns out to be the $a$-covariance of  the centered random variables $A\ceq\langle\xi_{1},\psi-\varphi\rangle-\big\langle\langle\xi_{1},\psi-\varphi\rangle\big\rangle_{\!a}$  and  $B\ceq\langle\xi_{1},\eta\rangle-\big\langle\langle\xi_{1},\eta\rangle\big\rangle_{\!a}$ and has an $a$-independent upper bound according to
	\begin{equation}
	\frac{\dl{}}{\dl{a}}\big\langle\langle\xi_{1},\eta\rangle\big\rangle_{a}=\langle AB\rangle_{\!a}\leq\big(\langle A^{2}\rangle_{a}\langle B^{2}\rangle_{a}\big)^{1/2}\leq \|\psi-\varphi\|\|\eta\|\,.
	\end{equation}
	Here, the first estimate is the {\scshape Schwarz} inequality with respect to $\langle\,(\cdot)\,\rangle_{a}$. For the second estimate we use $\langle B^{2}\rangle_{a}=\big\langle\big(\langle\xi_{1},\eta\rangle\big)^{2}\big\rangle_{a}-\big(\big\langle\langle\xi_{1},\eta\rangle\big\rangle_{a}\big)^{2}\leq \|\eta\|^{2}$ which follows from the positivity of squared real numbers, the {\scshape Schwarz} inequality for the scalar product, and $\|\xi_{1}\|=1$. We also use $\langle A^{2}\rangle_{a}\leq\|\psi-\varphi\|^{2}$ which follows analogously.

	\item For the first claim we use $\Lambda(\psi)\leq\langle|\psi|,1\rangle$ from \eqref{cumulant_bounds} in \eqref{transform} to obtain $\Lambda^{*}(\varphi)\geq\langle\psi,\varphi\rangle-\langle|\psi|,1\rangle$. Now we pick an arbitrary $\epsilon>0$ and apply the last inequality to a function $\varphi\in\L^{2}$ satisfying the lower estimate $|\varphi|\geq 1+\epsilon$  on some {\scshape Borel}-measurable set $B\subseteq{\lbrack0,1\rbrack}\times{\lbrack0,1\rbrack}$ of strictly positive {\scshape Lebesgue} area $|B|\ceq\langle\chi_{B},1\rangle>0$, where $\chi_{B}$ denotes the indicator function of $B$. For such a $\varphi$ we choose $\psi= a \varphi \chi_{B}\in\L^{2}$ with $a\in{\rbrack0,\infty\lbrack}$ and get $\Lambda^{*}(\varphi)\geq a\langle|\varphi|(|\varphi|-1),\chi_{B}\rangle\geq a(1+\epsilon)\epsilon|B|>0$. Taking the supremum over $a>0$ gives $\Lambda^{*}(\varphi)=\infty$ which is equivalent to the first claim.
	For the second claim we restrict the supremum in \eqref{transform} to $\psi=-a1$ with $a>0$ and obtain $\Lambda^{*}(\varphi)\geq-\Lambda(-a1)$ if $\langle1,\varphi\rangle\leq0$. Since the ${\lbrack0,1\rbrack}$-valued random variable $\langle1,\xi_{1}\rangle=\big(\int_{0}^{1}\dl{t}\sigma_{1}(t)\big)^{2}$ has the strictly positive variance $2(1-2m+p)/(\beta\bb)^{2}\,\big[\!>2(1-m)^{2}/(\beta\bb)^{2}\big]$, the supremum over $a>0$ gives $\Lambda^{*}(\varphi)=\infty$ which is equivalent to the second claim.
	For the third claim we note that $\Lambda(-a\rho\otimes\rho)\leq 0$ by \eqref{cumulant} and $\langle \rho\otimes\rho,\xi_{1}\rangle\geq 0$.  By \eqref{transform} we therefore get $\Lambda^{*}(\varphi)\geq -a\langle \rho\otimes\rho,\varphi\rangle$ for all $\varphi\in\L^{2}$. Taking the supremum over $a>0$ gives $\Lambda^{*}(\varphi)=\infty$ if $\langle \rho\otimes\rho, \varphi\rangle <0$. This is equivalent to the third claim. 
	\item The functional $\Lambda^{*}$ has the claimed two properties because it is, by definition \eqref{transform}, the pointwise supremum of a family of affine and weakly continuous functionals, see \cite[Prop.\,13.13]{BC2017}. The sets $\D^{*}$ and  $\D^{*}_{r}$ are convex because $\Lambda^{*}$ is convex and they are not empty because $\mu\in\D^{*}_{0}\subseteq\D^{*}_{r}\subset\D^{*}$.
	Since $\D^{*}$ is contained in the closed unit-ball $\overline{\B}_{1}\ceq\{\varphi\in\L^{2}:\|\varphi\|\leq1\}$ by the previous claim \ref{phibound}, it is bounded.\footnote{Since the unit-ball $\overline{\B}_{1}$ is closed with respect to the $\L^{2}$-norm and also convex (by the convexity of that norm), it is also weakly closed, see \cite[Thm.\,3.34]{BC2017}, and weakly (sequentially) compact \cite[Cor.\,2.38]{BC2017}.} Consequently, every sequence in $\D^{*}_{r}$ is bounded and therefore has a sub-sequence $ (\varphi_{n})_{n\geq1} $weakly converging to some $\psi\in\overline{\B}_{1}$, see \cite[Lem.\,2.45]{BC2017}. Since $\Lambda^{*}(\psi)\leq\liminf_{n\to\infty}\Lambda^{*}(\varphi_{n})\leq r$ by the weak (sequential) lower semi-continuity of $\Lambda^{*}$, we actually have $\psi\in \D^{*}_{r}$. In words, $\D^{*}_{r}$ is weakly sequentially closed. To conclude, $\D^{*}_{r}$ is weakly sequentially compact and, hence, by the {\scshape \v{S}mulian--Eberlein} equivalence also weakly compact, see \cite[Cor.\,2.38]{BC2017}.\qed
\end{enumerate}
\end{proof}

\begin{remark}\label{remlem4}

\renewcommand{\labelenumi}{(\roman{enumi})}
\renewcommand{\theenumi}{(\roman{enumi})}

\begin{enumerate}
	\item The fourth inequality in \eqref{cumulant_bounds} may be sharpened in a $\beta\bb$-dependent way according to
	\begin{equation}\label{jensen5}
	       \Lambda(|\psi|)\leq\langle\ln(\cosh(|\psi|)+\mu \sinh(|\psi|)),1\rangle \leq\langle|\psi|,1\rangle\,.
	\end{equation}
	The first inequality follows from the convexity of $\Lambda$ and the {\scshape Jensen} inequality applied to the two-fold integration underlying the $\L^{2}$-scalar product. The second inequality is due to $\mu\leq 1$.
        \item\label{wlsc} Since $\Lambda$ is weakly sequentially continuous, it is in particular weakly sequentially lower semi-continuous  (w.\,s.\,l.\,s.\,c.). By its convexity, $\Lambda$ is therefore even weakly lower semi-continuous, see \cite[Thm.\,9.1]{BC2017}.   
	
	\item Since the {\scshape G\^{a}teaux} differential  $\varphi\mapsto\langle\Lambda'(\psi),\varphi\rangle$ in Lemma\,\ref{some_properties}\,\ref{gateaux_diff} is linear and continuous, $\Lambda'(\psi)$  is even the {\scshape Fr\'{e}chet} derivative (or gradient) of $\Lambda$ at $\psi$, see \cite[Fact\,2.62]{BC2017}.
	\item \label{monoconvex} The inequalities \eqref{jensen3} imply monotonicity of $\Lambda'$ in the sense that $\langle\Lambda'(\psi)-\Lambda'(\varphi),\psi-\varphi\rangle\geq0$. In fact, this monotonicity is even equivalent to the convexity of $\Lambda$, see \cite[Prop.\,17.7]{BC2017}.

	\item We also have a pointwise monotonicity of $\Lambda$ in the sense that $0\leq\Lambda(\psi)\leq\Lambda(\varphi)$ if $0\leq\psi\leq\varphi$. This is a consequence of $0\leq\langle\mu,\varphi-\psi\rangle\leq\langle\Lambda'(\psi),\varphi-\psi\rangle\leq\Lambda(\varphi)-\Lambda(\psi)$. Here, the first inequality is obvious by $0<\mu$, the second one follows from \eqref{Lambda'_bound}, and the third one is \eqref{jensen3}.
	\item We only have the inequalities $0\leq\langle\Lambda'(\psi),1\rangle\leq 1$ and $0\leq\langle\Lambda'(\psi),\mu\rangle\leq m$ instead of \eqref{Lambda'_bounds} for general, not necessarily positive $\psi\in\L^{2}$. The first positivity follows from $0\leq\big(\int_{0}^{1}\dl{t}\sigma_{1}(t)\big)^{2}=\langle\sigma_{1}\otimes\sigma_{1},1\rangle=\langle\xi_{1},1\rangle$.
	The proof of the second positivity contains an additional argument according to $0\leq \int\dl{\sigma_{2}}\big(\int_{0}^{1}\dl{t}\!\sigma_{1}(t)\sigma_{2}(t)\big)^{2}=\int\dl{\sigma_{2}}\langle\sigma_{1}\otimes\sigma_{1},\sigma_{2}\otimes\sigma_{2}\rangle=\big\langle\sigma_{1}\otimes\sigma_{1},\langle\sigma_{2}\otimes\sigma_{2}\rangle_{\beta\bb}\big\rangle=\langle\xi_{1},\mu\rangle$. Here, we are using $\int\dl{\sigma_{2}}(\,\cdot\,)\ceq\langle(\,\cdot\,)\rangle_{\beta\bb}$ as another notation for the (partial) expectation with respect to the {\scshape Poisson} process steering the spin-flip process $\sigma_{2}$. \label{Lambda'_scalar_bounds}
\end{enumerate}
\end{remark}
Eventually we are prepared to present the main result of this section.
\begin{theorem}[Dual variational formulas for the macroscopic annealed free energy]\label{f_ann_variational}
\begin{enumerate}
	\item\label{varadhan_gauss}For any $\lambda>0$ the limit $\beta f_{\infty}^{\ann}$ of the dimensionless annealed free energy satisfies the following two equivalent variational formulas:
	\begin{align}
	\beta f_{\infty}^{\ann}+\ln\big(2\cosh(\beta\bb)\big)&=\inf\limits_{\varphi\in\D^{*}}\big(\Lambda^{\!*}(\varphi)-\lambda\langle\varphi,\varphi\rangle\big)\label{varadhan}\\
		&=\inf\limits_{\psi\in\L^{2}}\Big(\frac{1}{4\lambda}\langle\psi,\psi\rangle-\Lambda(\psi)\Big)\label{gauss}\,.
	\end{align}
	\item The infimum in \eqref{varadhan} and the infimum in \eqref{gauss} are attained and each (global) minimizer $\varphi_{\lambda}\in\D^{*}\subset\L^{2}$ in \eqref{varadhan}  (at given $\beta\bb>0$) corresponds to a minimizer $\psi_{\lambda}\in\L^{2}$ in \eqref{gauss}, and vice versa, through the relation $\psi_{\lambda}=2\lambda\varphi_{\lambda}$. In \eqref{gauss}, and hence in \eqref{varadhan}, one may restrict the infimization to positive and exchange symmetric $\psi\in\L^{2}$ without losing generality.
\label{correspondence}
\item	Each minimizer $\psi_{\lambda}\in\L^{2}$ in \eqref{gauss} is a continuous and exchange symmetric function solving the {\scshape Euler--Lagrange} critical equation\label{critical}
\begin{equation}\label{critical_equation}
	\psi=2\lambda\Lambda'(\psi)
\end{equation}
and obeying the (pointwise) bounds $2\lambda\mu\leq\psi_{\lambda}\leq 2\lambda\,1$ (so that $2\lambda\sqrt{p}\leq\|\psi_{\lambda}\|\leq2\lambda$). Corresponding properties hold for the minimizer in \eqref{varadhan} by \ref{correspondence}.
\item	For any $\lambda<1/2$ there exists only one solution of \eqref{critical_equation} and hence, by \ref{critical}, only one minimizer in \eqref{gauss} and hence, by \ref{correspondence}, only one minimizer in \eqref{varadhan}.
\end{enumerate}
\end{theorem}

\begin{proof}
\begin{enumerate}
\item At first we note that the infimum in \eqref{varadhan} is lower bounded by $-\lambda$, as follows from $\Lambda^{*}(\varphi)\geq 0$ and $\langle\varphi,\varphi\rangle\leq 1$ for $\varphi\in\D^{*}$, see Lemma\,\ref{some_properties}\ref{phibound}. By \eqref{varadhan} and \eqref{inf_strong} this lower bound may be recognized as a weakened version of the one in \eqref{f_ann_inf_bounds}. From this lemma and Remark\,\ref{remlem4}\,\ref{wlsc} we also recall that $\Lambda$ is convex and weakly lower semi-continuous, two properties which are well-known to be shared by the $\L^{2}$-norm and its square (due to the {\scshape Schwarz} and the {\scshape Jensen} inequality). The equality \eqref{gauss} therefore easily follows from {\scshape Legendre--Fenchel} duality, see \cite[Thm.\,2.2]{T1979} or \cite[Cor.\,14.20]{BC2017}).
	The proof of \eqref{varadhan} requires more work. We begin by rewriting \eqref{F_N2} as
	\begin{equation}\label{F_N3}
		F_{N}=\ln\big(\big\langle\exp\big(N\lambda\langle\xi_{N},\xi_{N}\rangle\big)\big\rangle_{\!\beta\bb}\big)  \qquad(N\geq2)
	\end{equation}
	using \eqref{P_N}, \eqref{Q_N}, and  the empirical averages $\xi_{N}$ introduced above Lemma\,\ref{some_properties}.
	In view of \eqref{F_N1} and \eqref{f_ann_inf} we need to show that 
	\begin{equation}\label{varadhan2}
		\lim\limits_{N\to\infty}\frac{F_{N}}{N}=\sup\limits_{\varphi\in\D^{*}}\big(\lambda\langle\varphi,\varphi\rangle-\Lambda^{*}(\varphi)\big)\,.
	\end{equation}
	To this end, we observe that the sequence $(\xi_{N})_{N\geq 1}$ satisfies a large-deviation principle (LDP) with convex (good) rate functional $\Lambda^{*}$ with respect to the weak topology on $\L^{2}$. Equivalently, this LDP is satisfied by the sequence of {\scshape Borel}\footnote{The norm and weak topologies on $ \L^{2}$ induce the same {\scshape  Borel} sigma algebra of events \cite{E1977}.} probability measures  $(\DD_{N})_{N\geq 1}$  on $\L^{2}$, where $\DD_{N}$ is the distribution of $\xi_{N}$ characterized by its {\scshape Lap\-lace} functional $\psi\mapsto\int_{\L^{2}}\dD{\varphi}\,\exp\big(\langle\psi,\varphi\rangle\big)=\exp\big(N\Lambda(\psi/N)\big)$ on $\L^{2}$. This follows from Lemma\,\ref{some_properties} and \cite[Thm.\,3.3.11]{DS1989} or \cite[Sec.\,6.1]{DZ2010}.\footnote{This functional LDP is one natural extension from $\R^{d}$- to $\L^{2}$-valued random variables of the pioneering refinement of the weak law of large numbers due to {\scshape Cram\'er} (1938) and {\scshape H.~Chernoff} (1952).}
	That given, the {\scshape Varadhan} integral lemma \cite[Thm.\,4.3.1]{DZ2010} (see also \cite[Thm.\,27.10]{K2002}), as an extension of the classic asymptotic method of {\scshape Lap\-lace}, then yields $\lim_{N\to\infty}F_{N}/N=\sup_{\varphi\in\B_{1}}\big(\lambda\langle\varphi,\varphi\rangle-\Lambda^{*}(\varphi)\big)$. Here we have used the norm-continuity of the squared norm $\varphi\mapsto\langle \varphi,\varphi\rangle=\|\varphi\|^{2}$ and the fact that the measure $\DD_{N}$ is supported on the closed unit-ball $\overline{\B}_{1}\subset\L^{2}$ because $\|\xi_{N}\|\leq 1$. Since $\lambda\langle\varphi,\varphi\rangle-\Lambda^{*}(\varphi)=-\infty$ for all $\varphi\in\overline{\B}_{1}\setminus\D^{*}$, the desired equality \eqref{varadhan2} follows.
\item At first we show that the infimum $I\ceq\inf_{\psi\in\L^{2}}\Omega(\psi)>-\lambda$ in \eqref{gauss} is attained. Here, 
	\begin{equation}\label{Omega} 
		\Omega(\psi)\ceq \frac{1}{4\lambda}\|\psi\|^{2}-\Lambda(\psi)\qquad\big(\psi\in\L^{2},\,\,\,\lambda\in]0,\infty[\,\,\big)
	\end{equation}
	defines the underlying non-linear functional $\Omega:\L^{2}\to\R$. 
	We will use two properties of $\Omega$.
	By $\Lambda(\psi)\leq\|\psi\|$ from \eqref{cumulant_bounds} we have $\Omega(\psi)\geq\|\psi\|\big(\|\psi\|-4\lambda\big)/(4\lambda)$ so that $\Omega$ is (super-)coercive, that is, $\lim_{\|\psi\|\to\infty}\Omega(\psi)/\|\psi\|=\infty$.{}
	Furthermore, we claim that $\Omega$ is w.\,s.\,l.\,s.\,c. in the sense of Remark\,\ref{remlem4}\,\ref{wlsc}, because it is the sum of two functionals with this property. For the first functional, the squared norm, this is well-known. Namely, if  $(\varphi_{n})_{n\geq1}\subset\L^{2}$ converges weakly to $\psi\in\L^{2}$, then we get $\|\psi\|^{2}\leq\liminf_{n\to\infty}\|\varphi_{n}\|^{2}$ from the obvious inequality $\|\psi\|^{2}\leq \|\varphi_{n}\|^{2}+2\|\psi\|^{2}-2\langle\psi,\varphi_{n}\rangle$. The second functional, $-\Lambda$, is w.\,s.\,l.\,s.\,c., because $\Lambda$ is even weakly sequentially continuous by Lemma\,\ref{some_properties}\,\ref{gateaux_diff}.
	Now we choose an arbitrary sequence $(\eta_{j})_{j\geq1}\subset\L^{2}$ infimizing $\Omega$ in the sense that $\lim_{j\to\infty}\Omega(\eta_{j})=I$. Since the (converging) sequence $\big(\Omega(\eta_{j}) \big)_{j\geq1}\subset{[-\lambda,\infty[}$ is bounded, the coerciveness of $\Omega$ implies the same for the underlying sequence $(\eta_{j})_{j\geq1}$, see \cite[Prop.\,11.20]{BC2017}. Therefore this has at least one sub-sequence $(\varphi_{n})_{n\geq1}$, $\varphi_{n}\ceq\eta_{j(n)}$, weakly converging to some limit in $\L^{2}$, see \cite[Lem.\,2.45]{BC2017}. We name this (unknown) limit $\psi_{\lambda}$, use the fact that $\Omega$ is w.\,s.\,l.\,s.\,c., and conclude by a well-known beautiful argument going back to {\scshape Bolzano} and {\scshape Weierstrass}, confer for example \cite{BB1992}:
	\begin{equation}
		I\leq\Omega(\psi_{\lambda})\leq\liminf_{n\to\infty}\Omega(\varphi_{n})\leq\limsup_{j\to\infty}\Omega(\eta_{j})=\lim_{j\to\infty}\Omega(\eta_{j})=I\,.
	\end{equation}
	To summarize, each weak accumulation point of any infimizing sequence is a minimizer.
	The other claims about the minimizers follow from the simplicity of the squared norm $\|(\cdot)\|^{2}$ and from the finiteness, convexity, weak sequential continuity, and {\scshape G\^{a}teaux} differentiability of $\Lambda$ according to \cite[Sec.\,2.1]{T1979} or by extending an argument in the proof of \cite[Thm.\,A.1]{CET2005} from the {\scshape Euclid}ean\, $\R^{d}$ to the real {\scshape Hilbert}ian\, $\L^{2}$. One may restrict to positive $\psi$ in \eqref{gauss} because $\langle\psi,\psi\rangle=\langle|\psi|,|\psi|\rangle$ and $\Lambda(\psi)\leq\Lambda(|\psi|)$ by \eqref{cumulant_bounds}. One may restrict to exchange symmetric $\psi$ because $\langle\psi,\psi\rangle\geq\langle\psi_{+},\psi_{+}\rangle$ and $\Lambda(\psi)=\Lambda(\psi_{+})$. Here, $\psi_{+}(t,t')\ceq \big(\psi(t,t')+\psi(t',t)\big)/2$ defines the exchange symmetric part $\psi_{+}$ of $\psi$.
\item At first we assert the two inequalities
	\begin{equation}\label{Omega_bounds}
		\lambda\|\Omega'(\varphi)\|^{2}\leq\Omega(\varphi)-\Omega(\psi_{\lambda})\leq \frac{1}{4\lambda}\|\varphi-\psi_{\lambda}\|^{2}\,.
	\end{equation}
	Here $\varphi\in\L^{2}$ is arbitrary, $ \psi_{\lambda}$ is an arbitrary minimizer of $\Omega$, and $\Omega'(\varphi)=(2\lambda)^{-1}\varphi-\Lambda'(\varphi)$ is the {\scshape Fr\'{e}chet} gradient of $\Omega$ at $\varphi$.
	Both inequalities follow from the second inequality in \eqref{jensen3}, rewritten as $\lambda\|\Omega'(\varphi)\|^{2}-\|\psi-2\lambda\Lambda'(\varphi)\|^{2}/(4\lambda)\leq \Omega(\varphi)-\Omega(\psi)$. To obtain the first inequality in \eqref{Omega_bounds} we take the supremum over $\psi$. For a minimizing $\varphi$ this then yields $\|\Omega'(\varphi)\|^{2}=0$, hence a solution of \eqref{critical_equation}.
	The claimed continuity and exchange symmetry just reflect the corresponding properties of $\Lambda'(\psi)$ for any $\psi\in\L^{2}$, according to Lemma\,\ref{some_properties}\,\ref{Lambda_convexity}.
	The claimed bounds follow from \eqref{critical_equation} combined with \eqref{Lambda'} and \eqref{Lambda'_bound}.
	The second inequality in \eqref{Omega_bounds} will be used only later in the proof of Theorem\,\ref{f_ann_second_order} in the next section. It follows from the above rewritten 	inequality in \eqref{jensen3} with $\varphi=\psi_{\lambda}$ by using the just obtained $\Omega'(\psi_{\lambda})=0$ and renaming $\psi$ to $\varphi$.
\item For two solutions $\psi$, $\widetilde{\psi}\in\L^{2}$ of \eqref{critical_equation} we have
	$\big\|\psi-\widetilde{\psi}\big\|=2\lambda\big\|\Lambda'\big(\psi\big)-\Lambda'\big(\widetilde{\psi}\big)\big\|\leq 2\lambda\big\|\psi-\widetilde{\psi}\big\|$ by \eqref{Lambda'_continuity}. It follows that $1\leq2\lambda$ if  $\psi\neq\widetilde{\psi}$. This implication is equivalent to the claim.
	\qed
\end{enumerate}
\end{proof}

\begin{remark}\label{remthm3}

\renewcommand{\labelenumi}{(\roman{enumi})}
\renewcommand{\theenumi}{(\roman{enumi})}
\begin{enumerate}
\item As to the proof of \eqref{varadhan2}, it is interesting to notice that the inequality
\begin{equation}\label{wolfbound}
	\frac{F_{N}}{N}\geq\sup_{\psi\in\L^{2}}\Big(\Lambda(\psi)-\frac{1}{4\lambda}\langle\psi,\psi\rangle\Big)=\sup_{\varphi\in\D^{*}}\big(\lambda\langle\varphi,\varphi\rangle-\Lambda^{*}(\varphi)\big)\quad(N\geq2)
\end{equation}
can be derived easily without large-deviation techniques.
One only has to use in \eqref{F_N3} the obvious inequality $\langle\xi_{N},\xi_{N}\rangle\geq 2\langle\xi_{N},\eta\rangle-\langle\eta,\eta\rangle$ for arbitrary $\eta\in\L^{2}$, combined with properties of the exponential, the definition \eqref{cumulant}, and the fact that the random variables $\sigma_{1}\otimes\sigma_{1},\dots,\sigma_{N}\otimes\sigma_{N}$ are independent and identically distributed. The equality in \eqref{wolfbound} is  \eqref{gauss}. 
\item \label{phimu} From \eqref{varadhan} and \eqref{gauss} we immediately rediscover that $\beta f^{\ann}_{\infty}$ depends on $\vv$ only via $\lambda>0$ and that $\beta f^{\ann}_{\infty}$ is not increasing in $\lambda$ because it is the pointwise infimum of a family of such functions. Since $\inf_{\psi\in\L^{2}}\Omega(\psi)=\inf_{\eta\in\L^{2}}\big(\lambda\langle\eta,\eta\rangle-\Lambda(2\lambda\eta)\big)$ by scaling, the concavity of $\beta f^{\ann}_{\infty}$ in $\lambda$ is seen to follow similarly from the convexity of $\Lambda$.
From \eqref{varadhan} and \eqref{gauss} also the upper estimates in \eqref{f_ann_inf_bounds} and \eqref{f_ann_inf_bound2} can easily be rediscovered. 
Namely, the upper estimate in \eqref{f_ann_inf_bounds} follows from \eqref{varadhan} by restricting to the single function $\varphi=\mu$ and from \eqref{gauss} by using $\Lambda(\psi)\geq\langle\psi,\mu\rangle$, see \eqref{cumulant_bounds}.
The estimate \eqref{f_ann_inf_bound2} follows from \eqref{varadhan} or \eqref{gauss} by restricting to the constant function $\varphi=1$ or $\psi=2\lambda\,1$ and by observing \eqref{transform_bounds} or \eqref{cumulant_bound2}, respectively. 
The lower estimate in \eqref{f_ann_inf_bounds} does not seem to be obtainable so easily from \eqref{varadhan} or \eqref{gauss}. However, the weaker lower estimate $-\lambda$, see \eqref{inf_strong}, immediately follows from \eqref{varadhan} according to the beginning of the proof of Theorem\,\ref{f_ann_variational}\,\ref{varadhan_gauss}. Alternatively, it follows from \eqref{gauss} by $\Lambda(\psi)\leq\|\psi\|$, see \eqref{cumulant_bounds}.
\item 	The {\scshape Lipschitz} continuity \eqref{Lambda'_continuity}  implies  $\big\langle\Lambda'(\psi)-\Lambda'(\varphi),\psi-\varphi\big\rangle\leq\|\psi-\varphi\|^{2}$ by the {\scshape Schwarz} inequality. Applying this to \eqref{Omega} yields
	\begin{equation}
	\big\langle\Omega'(\psi)-\Omega'(\varphi),\psi-\varphi\big\rangle\geq\Big(\frac{1}{2\lambda}-1\Big)\|\psi-\varphi\|^{2}
	\end{equation}
and hence strict monotonicity of $\Omega'$, equivalently \cite[Prop.\,17.7]{BC2017}, strict convexity of $\Omega$ on $\L^{2}$, if $2\lambda<1$. This shows again, without referring to the critical equation \eqref{critical_equation}, that  $\Omega$ has exactly one minimizer for $2\lambda<1$. 
\item The variational formula~\eqref{gauss} may be derived \emph{informally}, and without referring to \eqref{varadhan}, from rewriting \eqref{F_N3} as follows
	\label{gauss_white_noise}
	\begin{equation}\label{F_N4}
		F_{N}=\ln\Big(\int\d{\phi}\exp\big(N\Lambda(\sqrt{4\lambda}\,\phi)\big)\Big)\,.
	\end{equation}
Here we have interchanged the {\scshape Poisson} expectation with the ``linearizing'' integration over the (generalized) sample paths ${\lbrack0,1\rbrack}\times{\lbrack0,1\rbrack}\ni(t,t')\mapsto \phi(t,t')$ of centered two-time {\scshape Gauss}ian white noise with (generalized) covariance $\int\d{\phi}\phi(t,t')\phi(u,u')=\delta(t-u)\delta(t'-u')/(2N)$ in terms of the {\scshape Dirac} delta. The symbolic equation ``\,$\d{\phi}=\dll{\phi}\exp\big(-N\langle\phi,\phi\rangle\big)$'' for the {\scshape Gauss}ian probability measure $\GG_{N}$ in \eqref{F_N4} suggests  the {\scshape Laplace} method for the asymptotic evaluation as $N\to\infty$. The resulting variational formula then turns into \eqref{gauss} by the replacement $\phi\mapsto\psi/\sqrt{4\lambda}$. In the physics literature such a derivation often goes under the name {\scshape Hubbard--Stratonovich} trick or transformation, in particular when a {\scshape Dyson}--{\scshape Feynman} ``time-ordered exponential'' or ``product integral'' [instead of the {\scshape PFK} representation \eqref{fk}] is employed in order to ``disentangle'' non-commuting {\scshape Hil\-bert}-space operators. See, for example \cite{M1975,BM1980,S1981,KK2002}. For a rigorous approach to {\scshape Feynman}'s disentangling formalism we refer to the monograph \cite{JL2000}.
\item \label{landau} Upon dividing by $\beta$, the functional \eqref{Omega} in \eqref{gauss} may be viewed as a simple {\scshape Landau}--{\scshape Ginz\-burg}--{\scshape Wilson} free energy \cite{H1987} obtained by integration over the paths $\sigma_{1}:t\mapsto\sigma_{1}(t)$ of a single spin-flip process in the sense of \eqref{teleproc}:
\begin{equation}
\beta^{-1}\Omega(\psi)=-\beta^{-1}\ln\big(\big\langle\exp\big(-\beta {\cal H}_{1}(\psi,\sigma_{1})\big)\big\rangle_{\beta\bb}\big)\,.
\end{equation}
Here, the effective {\scshape Hamilton}ian\, ${\cal H}_{1}(\psi,\sigma_{1})$ associates to any $\psi\in\L^{2}$ and any $\sigma_{1}$ a non-instantaneous self-interaction energy by $\beta{\cal H}_{1}(\psi,\sigma_{1})\ceq \langle\psi,\psi\rangle/(4\lambda)-\langle\psi,\sigma_{1}\otimes\sigma_{1}\rangle \geq-\lambda$. The corresponding effective {\scshape Gibbs} expectation is given by  $\langle(\,\cdot\,)\rangle_{\beta\bb,\psi}$, see \eqref{Lambda'}. The critical equation $\Omega'(\psi)=0$, see \eqref{critical_equation} combined with \eqref{Lambda'}, then identifies the (positive) self-interaction function $\psi$ ``self-consistently'' (up to a factor $2\lambda$) with the ``perturbed'' auto-correlation function of the spin $\langle\sigma_{1}\otimes\sigma_{1}\rangle_{\beta\bb,\psi}\geq\langle\sigma_{1}\otimes\sigma_{1}\rangle_{\beta\bb}=\mu$. The combination of Theorem\,\ref{f_ann_variational} and Theorem\,\ref{difference} below therefore constitutes, for $4\lambda<1$, a rigorous theory of the ``self-consistency equations" in the respective physics literature, among them \cite{FS1986,YI1987,US1987,GL1990,MH1993}. In this literature the condition $\langle\psi_{\lambda},1\rangle=1/2$ is sometimes used to determine the critical line between the paramagnetic and the spin-glass phase in the temperature-field plane, see for example \cite{US1987}. By the (pointwise) inequalities $2\lambda\mu\leq\psi_{\lambda}\leq2\lambda\,1$ from Theorem\,\ref{f_ann_variational}\,\ref{critical} this condition implies $1\leq4\lambda\leq1/m$ which indicates the absence of spin-glass order for $4\lambda<1$ in agreement with Corollary\,\ref{no_sg} below, see also Section\,\ref{conclusion}. In this connection we mention that the parameter $m$ defined in \eqref{m} may be recognized as the zero-field single-spin (or ``local") \z\z-susceptibility for $\lambda=0$.
\end{enumerate}
\end{remark}

\section{The macroscopic annealed free energy for weak disorder}

Unfortunately, we do not know explicitly a single minimizer in \eqref{varadhan} or \eqref{gauss} if $\lambda>0$.\footnote[4]{In the limiting case $\lambda\downarrow0$ we simply have $\varphi_{0}=\mu$, see \eqref{transform_bounds}, and therefore $\psi_{0}=0\cdot\mu=0$.}
In this section we therefore compare the global minimum $\Omega(\psi_{\lambda})=\beta f^{\ann}_{\infty}+\ln\big(2\cosh(\beta\bb)\big)$, see Theorem\,\ref{f_ann_variational}\,\ref{correspondence}, to its simple upper bound $\Omega(2\lambda\mu)$, that is, to the functional \eqref{Omega} evaluated at $\psi=2\lambda\mu$. Fortunately, it turns out that $\Omega(2\lambda\mu)$ not only shares with $\Omega(\psi_{\lambda})$ the properties of convexity and monotonicity in $\lambda$, but also constitutes a very good approximation to $\Omega(\psi_{\lambda})$ for small $\lambda$. More precisely, their respective asymptotic expansions, as $\lambda\downarrow0$, coincide up to the second order.\footnote[3]{Coincidence up to the first order follows already by an argument in Remark\,\ref{remthm3}\,\ref{phimu}.} The corresponding second-order coefficient turns out to be a rather complicated function of $\beta\bb>0$ taking only negative values larger than $-0.14$, see Figure\,\ref{c0_plot} below. The main drawback of $\Omega(2\lambda\mu)$ is the fact that it does not yield the true behavior of $\Omega(\psi_{\lambda})$ for large $\lambda$. However, due to our main result Theorem\,\ref{difference} in the next section, it is (presently) only the weak-disorder regime, $4\lambda<1$, for which $\Omega(\psi_{\lambda})$ is known to be physically relevant. We begin with
\begin{lemma}[Some properties of the upper bound $\Omega(2\lambda\mu)$] \label{lem_Omega_mu}

\noindent
	The function $\lambda\mapsto\Omega(2\lambda\mu)=p\lambda-\Lambda(2\lambda\mu)$
	\begin{enumerate}
		\item is concave and strictly decreasing,
		\item obeys for any $\lambda>0$ the estimates
		\begin{align}
			p\lambda-\ln\big(1+(p/m)(\e^{2m\lambda}-1)\big)\leq\Omega(2\lambda\mu)&\leq-p\lambda\,,\label{Omega_mu_bounds}\\
			-(2m-p)\lambda\leq\Omega(2\lambda\mu)&\leq -(2m-p)\lambda+\ln\big(\cosh(\beta\bb)\big)\,,\label{Omega_mu_bound2}
		\end{align}
		\item has the second-order {\scshape Taylor} formula
			\begin{equation}\label{Omega_mu_second_order}
			 \Omega(2\lambda\mu)=-p\lambda-2c_{0}\lambda^{2}-\frac{4}{3}r(\lambda)\lambda^{3}
		\end{equation}
		with the variance
		\begin{equation}\label{c}
		c_{0}\ceq\big\langle\langle\mu,\xi_{1}\rangle^{2}\big\rangle_{\beta\bb}-p^{2}=\frac{m-p}{4(\beta\bb)^{2}}+\frac{2p-m}{6}-\Big(\frac{2p-m}{2}\Big)^{2}
		\end{equation}
		and some continuous function $r:{\lbrack0,\infty\lbrack}\to{\lbrack-m^{3},m^{3}\rbrack}$, $\lambda\mapsto r(\lambda)$,
		\item has the strong-disorder limit
			\begin{equation}\label{Omega_mu_limit}
				\lim\limits_{\lambda\to\infty}\frac{1}{\lambda}\Omega(2\lambda\mu)=-(2m-p)\,.
			\end{equation}
	\end{enumerate}
\end{lemma}
\begin{proof}
	\begin{enumerate}
		\item The concavity follows from the convexity of $\Lambda$. The strict monotonicity follows from concavity and the fact that the first derivative at $\lambda=0$ equals $-			p<0$.
		\item In terms of the ${\lbrack0,m\rbrack}$-valued random variable $q\ceq\langle\mu,\xi_{1}\rangle$, see Remark\,\ref{remlem4}\,\ref{Lambda'_scalar_bounds}, we have 
		\begin{equation}\label{elementary}
			\e^{\langle 2\lambda\mu,\xi_{1}\rangle}=\e^{2m\lambda q/m}\leq 1+(q/m)(e^{2m\lambda}-1)
		\end{equation}
		by the elementary ({\scshape Jensen}) inequality used in the proof of the first inequality in \eqref{G_N_weak}. Taking now the {\scshape Poisson} expectation and
		the logarithm gives an upper bound on $\Lambda(2\lambda\mu)$ which, in turn, yields the lower estimate in \eqref{Omega_mu_bounds}. The upper estimate in 			\eqref{Omega_mu_bounds} is $\Lambda(2\lambda\mu)\geq 2p\lambda$ from \eqref{cumulant_bounds}. The lower estimate in \eqref {Omega_mu_bound2} follows from 			weakening that in \eqref {Omega_mu_bounds} by using $p/m\leq 1$. The upper estimate in \eqref {Omega_mu_bound2} follows from
		 \eqref{cumulant_bound2}. 
		\item We use the expectation $\langle(\,\cdot\,)\rangle_{\lambda}\ceq\langle(\,\cdot\,)\rangle_{\beta\bb,2\lambda\mu}$, see \eqref{Lambda'}. 
	Since $q\in{\lbrack0,m\rbrack}$, the convex function $\lambda\mapsto\Lambda(2\lambda\mu)=\ln\big(\langle\e^{2\lambda q}\rangle_{0}\big)$ is arbitrarily often differentiable and its second-order {\scshape Taylor} formula (at $\lambda=0$ with remainder in {\scshape Lagrange} form) affirms that for each $\lambda>0$ there exists some (unknown) number $a\in{\rbrack0,1\lbrack}$ such that
	\begin{equation}\label{taylor}
		\Lambda(2\lambda\mu)= (2\lambda)\langle q\rangle_{0}+\frac{(2\lambda)^{2}}{2!}\big\langle(q-\langle q\rangle_{0})^{2}\big\rangle_{0}+\frac{ (2\lambda)^{3} }{3!}r(\lambda)	\end{equation}
	with the third cumulant $r(\lambda)\ceq\big\langle(q-\langle q\rangle_{a\lambda})^{3}\big\rangle_{a\lambda}\in{\lbrack-m^{3},m^{3}\rbrack}$. The claim now follows from $\langle q\rangle_{0}=p$ and an explicit calculation of $\langle q^{2}\rangle_{0}$. The latter is based on \eqref{mu}, \eqref{prod_exp_cond} in Appendix\,\ref{poisson_appendix} for $ \langle\, \sigma_1(t_{1})\sigma_1(t_{2})\sigma_1(t_{3})\sigma_1(t_{4})\, \rangle_{0}$, and a straightforward but somewhat tedious integration over the four-dimensional unit-cube ${\lbrack0,1\rbrack}^{4}$. 
	\item The limit \eqref{Omega_mu_limit} follows from \eqref{Omega_mu_bound2}.
	\qed
	\end{enumerate}
\end{proof}

\begin{remark}\label{rem_lem_second_order}
\renewcommand{\labelenumi}{(\roman{enumi})}
\renewcommand{\theenumi}{(\roman{enumi})}
\begin{enumerate}
	\item \label{c0_remark} Since the explicit expression \eqref{c} for the variance $c_{0}$ is rather complicated, we mention its simple bounds according to
	\begin{equation}\label{simple_bounds}
	\frac{(m-p)^{2}}{\cosh(\beta\bb)}\leq c_{0}\leq (m-p)p\,.
	\end{equation}
	In the notation of the proof of Lemma\,\ref{lem_Omega_mu} the lower bound follows from $c_{0}=\big\langle(q-p)^{2}\big\rangle_{0}\geq\big\langle(q-p)^{2} 1(\sigma_{1})\big\rangle_{0}=(m-p)^{2}\langle1(\sigma_{1})\rangle_{0}$, confer the proof of \eqref{f_ann_N_bound2}. The upper bound simply follows from $q^{2}=qq\leq mq$. The lower bound implies the expected strict positivity of $c_{0}$ (for $\beta\bb>0$). The upper bound shows that $c_{0}$ vanishes (only) in the limiting cases $\beta\bb\downarrow0$ and $\beta\bb\to\infty$, and is smaller than $p/4$. As a function of $\beta\bb$ the variance $c_{0}$ is continuous and attains its maximum value $0.0695\dots$ at $\beta\bb=0.9089\dots$ (according to a simple numerical computation), see Figure\,\ref{c0_plot}. For small $\beta\bb$ we have the asymptotic expansion $c_{0}=\frac{4}{5}(\beta\bb)^{2}-\frac{13}{36}(\beta\bb)^{4}+\cdots$ as $\beta\bb\downarrow0$.
\begin{figure}
\begin{center}
	\psset{xunit=2.5cm,yunit=38cm}
	\begin{pspicture}(-0.9cm,-0.5cm)(10.1cm,3.5cm)
	\psaxes[Dx=0.5,Dy=0.025]{->}(3.8,0.085)
	\psset{algebraic,plotpoints=500}
	\psplot[linecolor=blue]{0.02}{3.6}{TANH(x)/(8*x^3)-0.5/(2*x*COSH(x))^2+1/(6*(COSH(x))^2)-1/(2*(COSH(x))^2)^2}
	\rput(0,0.09){ $c_{0}$}
	\rput(3.95,0){$\beta\bb$}
	\end{pspicture}
\caption{Plot of the variance $c_{0}$ given by \eqref{c} as a function of $\beta\bb$. See Remark\,\ref{rem_lem_second_order}\,\ref{c0_remark}.}
\label{c0_plot}
\end{center}
\end{figure}
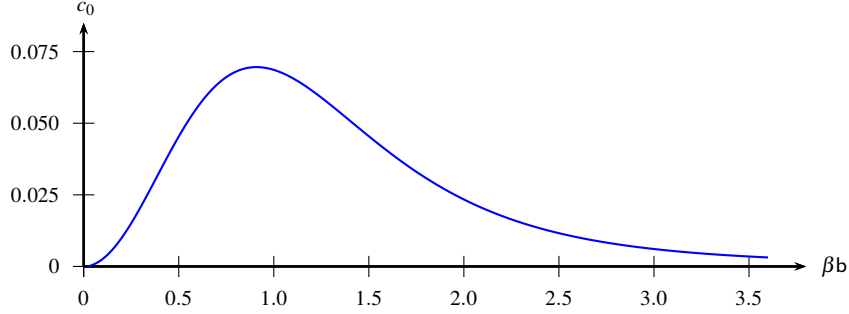
	\item The rather explicit lower estimate in \eqref{Omega_mu_bounds} shares with $\Omega(2\lambda\mu)$ the properties of concavity and monotonicity. It also has the same leading asymptotic behaviors in the limits of small and large $\lambda$ and $\bb$. But the second-order coefficient of its small-$\lambda$ {\scshape Taylor} series is (already) smaller and given by $-2(m-p)p$, confer \eqref{Omega_mu_second_order} and \eqref{simple_bounds}. Nevertheless, the (positive) difference between $\Omega(2\lambda\mu)$ and the lower estimate in \eqref{Omega_mu_bounds} does not exceed $2(m-p)\lambda\min\{m\lambda,1\}$ for all $\lambda$. This follows from \eqref{Omega_mu_bounds} combined with \eqref{log_inequality}, respectively with $p/m\leq 1$.

	\item Obviously, \eqref{Omega_mu_limit} does not reflect the true strong-disorder limit \eqref{strong_disorder} of $\Omega(\psi_{\lambda})$ because $2m-p\leq m^2/p <1$. But by 	generalizing $\Omega(2\lambda\mu)$ to the one-parameter variational expression $\min_{x\in{\lbrack0,1\rbrack}}\Omega\big(2\lambda(\mu+x(1-\mu))\big)\leq\Omega(2\lambda\mu)$ this limit may be included (for $x=1$) without changing the first two terms on the right-hand side of \eqref{Omega_mu_second_order}.
	\end{enumerate}
\end{remark}
The next theorem shows that $\Omega(2\lambda\mu)$ is a very good approximation to the global minimum $\Omega(\psi_{\lambda})=\beta f^{\ann}_{\infty}+\ln\big(2\cosh(\beta\bb)\big)$ for small $\lambda$, see Theorem\,\ref{f_ann_variational}\,\ref{correspondence} and \eqref{Omega}.

\begin{theorem}[The macroscopic annealed free energy up to second order in $\lambda$] \label{f_ann_second_order}

\noindent
We have the following error estimates
\begin{equation}\label{Omega_diff}
	0\leq\Omega(2\lambda\mu)-\Omega(\psi_{\lambda})\leq 4\lambda^{3}\qquad (\lambda>0)
\end{equation}
and the two-term asymptotic expansion for weak disorder
\begin{equation}\label{asymptotic}
	\beta f^{\ann}_{\infty} +\ln\big(2\cosh(\beta\bb)\big)=\Omega(\psi_{\lambda})=-p\lambda-2c_{0}\lambda^{2}+\mathcal{O}(\lambda^{3})\qquad(\lambda\downarrow0)
\end{equation}
with $c_{0}$ given by \eqref{c} and the usual understanding  of the {\scshape Landau} big-Oh notation that $\mathcal{O}(\lambda^{3})$ stands for some function of $\lambda$ with $\limsup_{\lambda\downarrow0}|\mathcal{O}(\lambda^{3})|/\lambda^{3}<\infty$.
\end{theorem}

\begin{proof}
The first inequality in \eqref{Omega_diff} is obvious because $\psi_{\lambda}$ is a minimizer of $\Omega$. For the second inequality we observe $\|2\lambda\mu-\psi_{\lambda}\|=\|2\lambda\Lambda'(0)-2\lambda\Lambda'(\psi_{\lambda})\|\leq2\lambda\|0-\psi_{\lambda}\|\leq(2\lambda)^{2}$ by \eqref{Lambda'_bound}, Theorem\,\ref{f_ann_variational}\,\ref{critical}, and \eqref{Lambda'_continuity}. Using this estimate in the second inequality of \eqref{Omega_bounds} yields \eqref{Omega_diff}. Combining \eqref{Omega_diff} with \eqref{Omega_mu_second_order} yields \eqref{asymptotic}.
\qed
\end{proof}

\begin{remark}\label{remcor}
\renewcommand{\labelenumi}{(\roman{enumi})}
\renewcommand{\theenumi}{(\roman{enumi})}
\begin{enumerate}
	\item The first two terms on the right-hand side of \eqref{asymptotic} reflect the concavity and monotonicity of $\Omega(\psi_{\lambda})$ in $\lambda$. Both of the (positive) coefficients $p$ and $c_{0}$ vanish in the limit $\bb\to\infty$ in agreement with \eqref{high_field} and \eqref{Omega_mu_bounds}. Recall also Remark\,\ref{remthm2}\,\ref{high}.
	\item The estimates \eqref{Omega_diff} can be sharpened according to
	\begin{equation}\label{Omega_diff2}
		\lambda\|\Omega'(2\lambda\mu)\|^{2}\leq\Omega(2\lambda\mu)-\Omega(\psi_{\lambda})\leq\min\big\{4\lambda^{3},\inf_{N\geq 2}G_{N}/N-p\lambda\big\}\,.
	\end{equation}
	The lower estimate in \eqref{Omega_diff2} follows from the lower estimate in \eqref{Omega_bounds} and is strictly positive, because $2\lambda\mu$ is not a minimizer of $\Omega$ (for $\lambda>0$).
	The second upper estimate in \eqref{Omega_diff2} follows from the lower estimate in \eqref{f_ann_inf_bounds} and the upper estimate in \eqref{Omega_mu_bounds}.
	For given $\lambda>0$ it sharpens the one in \eqref{Omega_diff} for sufficiently small and sufficiently large $\beta\bb$, see Remark\,\ref{rem_inf_G_N}\,\ref{inf_bounds}.
	We also have the (simplified) relative-error estimates
	\begin{equation}\label{relative_deviation}
		\|\Omega'(2\lambda\mu)\|^{2}\leq\frac{\Omega(2\lambda\mu)-\Omega(\psi_{\lambda})}{|\Omega(\psi_{\lambda})|}\leq \min\Big\{\frac{4\lambda^{2}}{p},1-p\Big\}\,.
	\end{equation}
	The lower and the first upper estimate follow by combining \eqref{Omega_diff2} with the estimates $p\lambda\leq|\Omega(\psi_{\lambda})|\leq \inf_{N\geq 2}G_{N}/N\leq\lambda$ from \eqref{f_ann_inf_bounds} and Remark\,\ref{rem_inf_G_N}\,\ref{inf_bounds}.The second upper estimate follows from writing the relative error as $\Omega(2\lambda\mu)/|\Omega(\psi_{\lambda})|+1$ and using again $\Omega(2\lambda\mu)\leq-p\lambda$ as well as $|\Omega(\psi_{\lambda})|\leq \lambda$.
	\item In view of \eqref{critical_equation} combined with \eqref{Lambda'_continuity} the (unique) minimizer $\psi_{\lambda}$ of $\Omega$ for $2\lambda<1$ may be determined (numerically) with arbitrary precision by the successive approximations
	\begin{equation}
		\psi_{\lambda}^{(n+1)}\ceq 2\lambda\Lambda'\big(\psi_{\lambda}^{(n)}\big)\qquad\big(\psi_{\lambda}^{(1)}\ceq2\lambda\mu\,,\quad n\in\N\big)\,.
	\end{equation}
	The (norm-)convergence of this minimizing sequence is exponentially fast according to
	\begin{equation}\label{fixed_point_diff}
		\|\psi^{(n)}_{\lambda}-\psi_{\lambda}\|\leq (2\lambda)^{n-1}\|\psi^{(1)}_{\lambda}-\psi_{\lambda}\|\leq(2\lambda)^{n}\|\psi_{\lambda}\|\leq(2\lambda)^{n+1}\,.
	\end{equation}
	Here, the first inequality follows from \eqref{Lambda'_continuity} by mathematical induction. For the next inequalities see the proof of Theorem\,\ref{f_ann_second_order}.\footnote{This is, of course, consistent with the {\scshape Banach} fixed-point theorem.} Using $\varphi=\psi^{(n)}_{\lambda}$ in \eqref{Omega_bounds} combined with \eqref{fixed_point_diff} yields an approximation to the macroscopic annealed free energy with an error not exceeding $(2\lambda)^{2n+1}/2$. For $n=1$ we get back to \eqref{Omega_diff}.
	\item By Theorem\,\ref{f_ann_second_order} we know that $\Omega\big(\psi^{(1)}_{\lambda}\big)$ coincides with $\Omega(\psi_{\lambda})$ up to the order $\lambda^{2}$, as $\lambda\downarrow0$.
	By \eqref{fixed_point_diff} we see that $\psi^{(2)}_{\lambda}$ coincides with the minimizer $\psi_{\lambda}$ up to the same order. Therefore it is of interest to determine $\psi^{(2)}_{\lambda}$ up to that order. To this end, we recall that for each non-zero $\eta\in\L^{2}$ the mapping $\psi\mapsto\langle\eta,\psi\rangle\eta$ defines a positive rank-one operator on $\L^{2}$, which we denote by $|\eta\rangle\langle\eta|$ following {\scshape Dirac}. The {\scshape Poisson} average $\text E\ceq\big\langle|\xi_{1}\rangle\langle\xi_{1}|\big\rangle_{\beta\bb}$ of the projection $|\xi_{1}\rangle\langle\xi_{1}|$ is a positive integral operator with a continuous ${[0,1]}$-valued kernel given by $\big\langle\xi_{1}(t,t')\xi_{1}(u,u')\big\rangle_{\beta\bb}$ for $(t,t'),(u,u')\in{[0,1]}\times{[0,1]}$. By an extension of \eqref{taylor} we now have
	
\begin{equation}
	\Lambda(2\lambda\varphi)=(2\lambda)\langle \varphi,\mu\rangle+\frac{(2\lambda)^{2}}{2!}\langle\varphi,(\text E-|\mu\rangle\langle\mu|)\varphi\rangle+{\cal O}(\lambda^{3})
\end{equation}

for any $\varphi\in\L^{2}$ with the derivative
\begin{equation}
	\Lambda'(2\lambda\varphi)=\mu+2\lambda(\text E-|\mu\rangle\langle\mu|)\varphi+{\cal O}(\lambda^{2})\,.
\end{equation}
Hence, we arrive at
\begin{equation}\label{psi_2}
	\psi^{(2)}_{\lambda}=2\lambda\Lambda'(2\lambda\mu)=2\lambda\mu+(2\lambda)^{2}(\text E-p\mathds{1})\mu+2\lambda{\cal O}(\lambda^{2})=\psi_{\lambda}+{\cal O}(\lambda^{3})
\end{equation}
with $\mathds{1}$ denoting the identity operator on $\L^{2}$. Like the first-order term in \eqref{psi_2} also the second-order term is a continuous, exchange symmetric, and positive $\L^{2}$-function, in agreement with Theorem\,\ref{f_ann_variational}\,\ref{critical}. The first two properties are directly inherited from the integral kernel of the operator $\text E$. The positivity follows from the (pointwise) inequality $\text E\psi\geq\langle\mu,\psi\rangle\mu$ for all $\psi\geq0$ due to  \eqref{multipoint_inequality}. The function $\text E\mu$ can be calculated explicitly. By \eqref{c} we have, in particular, $\langle\mu,\text E\mu\rangle=c_0+p^2$. Further properties of $\text E$ are given by the operator inequalities $0\leq|\mu\rangle\langle\mu|\leq\text E\leq\mathds{1}$ and the equality $\textrm{tr}\,\text E=1$ for its trace. Consequently, the uniform norm of the operator difference $\text A\ceq\text E-|\mu\rangle\langle\mu|\geq0$ obeys $\|\text A\|\leq\textrm{tr}\,\text A=1-p$.
\end{enumerate}
\end{remark}
Sometimes variational problems in function spaces like \eqref{gauss} are drastically simplified by restricting the set of all allowed variational functions to the one-parameter subset of functions of the form $\psi=y1$ where $1$ is the constant unit function (in $\L^{2}$ for the present case) and $y\in\R$ is arbitrary. This is often called, for an obvious reason, the \emph{static approximation}, confer for example \cite{BM1980,S1981,YI1987,US1987,KK2002,T2007}. In view of Lemma\,\ref{lem_Omega_mu} and Theorem\,\ref{f_ann_second_order} it is not surprising that this approximation in the present case does not yield the true behavior for small $\lambda$, not even up to the first order in $\lambda$. As opposed to that, $\Omega(2\lambda1)$ has the same strong-disorder limit as $\Omega(\psi_{\lambda})$ which, however, does not reflect the true behavior of the macroscopic (quenched) free energy in this limit. The main properties of the static approximation to the macroscopic annealed free energy are compiled in
\begin{corollary}[On the static approximation and its insufficiency for weak disorder] \label{constant_psi}

\noindent
Let the restriction of the infimization in \eqref{gauss} to the one-dimensional subspace of all constant $\L^{2}$-functions be denoted as
\begin{equation}\label{J1}
	J(\lambda)\ceq\inf\limits_{x\in\R}\Omega(2\lambda x1)=\inf\limits_{x\in\R}\big(\lambda x^{2}-\Lambda(2\lambda x1)\big)\,.
\end{equation}
Then the function $\lambda\mapsto J(\lambda)$
\begin{enumerate}
	\item is concave and not increasing,
	\item obeys for any $\lambda>0$ the three estimates
		\begin{align}
			-\inf_{N\geq2}G_{N}/N \leq&\,\,J(\lambda)\leq -m^{2}\lambda\,,\label{J1_bounds}\\
			&\,\,J(\lambda)\leq -\lambda+\ln\big(\cosh(\beta\bb)\big)\,,\label{J1_bound2}
		\end{align}
	\item obeys for any $\lambda\leq r/\big(2(1-m)\big)$, with arbitrary $r\in{\rbrack0,1\lbrack}$\,, the lower estimate
		\begin{equation}
			-\frac{m^{2}\lambda}{1-r} \leq J(\lambda)\,,\label{J1_bound3}
		\end{equation}
	\item has the weak- and strong-disorder limits
		\begin{equation}
		\lim\limits_{\lambda\downarrow 0}\frac{J(\lambda)}{\lambda}=-m^{2}\,,\qquad\qquad
		\lim\limits_{\lambda\to\infty}\frac{J(\lambda)}{\lambda}=-1\label{J1_limits}\,.
	\end{equation}
\end{enumerate}

\end{corollary}

\begin{proof}
\begin{enumerate}
		\item The claim holds because it is true for any function defined by an arbitrary restriction of the infimization in \eqref{gauss}, confer Remark\,\ref{remthm3}\,\ref{phimu}. More explicitly, $J$ is concave because it is the pointwise infimum of a family of such functions according to the convexity of $\Lambda$. And similarly, $J$ is not increasing because it is the pointwise infimum of a family of such functions according to $J(\lambda)=\inf_{y\in\R}\big(y^{2}/(4\lambda)-\Lambda(y1)\big)$.
	\item The lower estimate in \eqref{J1_bounds} is obvious from $-\inf_{N\geq2}G_{N}/N\leq\Omega(\psi_{\lambda})\leq J(\lambda)$.
	The upper estimate in \eqref{J1_bounds} follows by restricting to $x=m$ in \eqref{J1} and using $\Lambda(\psi)\geq\langle\psi,\mu\rangle$ from \eqref{cumulant_bounds}.
	Estimate \eqref{J1_bound2} follows by restricting to $x=1$  and using \eqref{cumulant_bound2}.

	\item  In terms of the ${\lbrack0,1\rbrack}$-valued random variable  $\widetilde q\ceq\big(\int_{0}^{1}\dl{t}\sigma_{1}(t)\big)^{2}$  we have $\langle y 1,\xi_{1}\rangle=y\widetilde q$ 	and by an analogy to \eqref{elementary} therefore
	\begin{equation}
	\Lambda(y1)\leq \ln\big(1+m(\e^{y}-1)\big)\leq m|y|+(1-m)\,y^{2}/2\label{b2}
	\end{equation}
	for all $y\in\R$. The second inequality follows from \eqref{log_inequality}. This gives
	\begin{equation}
		y^{2}-4\lambda\Lambda\big(y1\big)\geq\big(1-2(1-m)\lambda\big)y^{2}-4\lambda m|y|\geq(1-r)y^{2}-4\lambda m|y|\label{b3}\,.
	\end{equation}
	The proof of \eqref{J1_bound3} is completed by completing the square in \eqref{b3} and dividing by $4\lambda$.

	\item The strong-disorder limit follows from \eqref{J1_bound2}, the lower estimate in \eqref{J1_bounds}, and \eqref{g_limits}. For the weak-disorder limit we start from $m^{2}/(r-1)\leq \liminf_{\lambda\downarrow0}J(\lambda)/\lambda$ by \eqref{J1_bound3}, take the supremum over $r\in{\rbrack0,1\lbrack}$, and observe the upper estimate in \eqref{J1_bounds}.\qed

\end{enumerate}
\end{proof}

\begin{remark}\label{remcor_J1}

\renewcommand{\labelenumi}{(\roman{enumi})}
\renewcommand{\theenumi}{(\roman{enumi})}
\begin{enumerate}

	\item The small-$\lambda$ estimate \eqref{J1_bound3} is not only useful for the proof of the weak-disorder limit in \eqref{J1_limits}, which differs from the true limit in \eqref{asymptotic} since $m^{2}<p$, but it also implies that $J(\lambda)$ is strictly larger than the minimum $\Omega(\psi_{\lambda})$ in \eqref{gauss} for sufficiently small $\lambda>0$. More precisely, by choosing $r<1-(m^{2}/p)$ we get from \eqref{f_ann_inf_bounds}, \eqref{gauss} with \eqref{varadhan}, and \eqref{J1_bound3} that
	\begin{equation}
		\beta f^{\ann}_{\infty}+\ln\big(2\cosh(\beta\bb)\big)=\Omega(\psi_{\lambda})\leq -p\lambda<J(\lambda)
	\end{equation}
	for all $\lambda<(p-m^{2})/\big(2p(1-m)\big)\,\,\big[\!\!<(m-p)/(2p)<1/2\big]$. This upper bound on $\lambda$ is a continuous and strictly increasing function of $\beta\bb>0$ and approaches its extreme values $0$ and $1/2$ in the limiting cases $\beta\bb\downarrow0$ and $\beta\bb\to\infty$, respectively. It attains the value $1/4$ approximately at $\beta\bb=3$. On the other hand, the true strong-disorder limit in \eqref{J1_limits} implies that $J(\lambda)\leq-p\lambda$ for sufficiently large $\lambda$.
	\item Our proof of Corollary\,\ref{constant_psi} is based on rather crude estimates of $\Lambda(y1)$ that easily follow from its definition. Additional information on $J(\lambda)$, for intermediate values of $\lambda$, may be obtained from the formula
	\begin{equation}\label{Lambda_gauss}
		\Lambda(y1)=\ln\Big(\int_{\R}\dll{z}\e^{-\pi z^{2}}\cosh\Big(\sqrt{(\beta\bb)^{2}+ 4\pi y z^{2}}\Big)\Big/\cosh(\beta\bb)\Big)\qquad (y\geq0)\,. 
	\end{equation}
	It follows from a {\scshape Gauss}ian linearization and a consequence of the {\scshape PFK} formula, see Remark\,\ref{rem_pfk}\,\ref{Laplace_transform} in Appendix\,\ref{pfk_proof}. The restriction to $y\geq 0$ in \eqref{Lambda_gauss} causes no problem, because one may restrict to $x\geq 0$ in \eqref{J1} without losing generality. This follows from $\Lambda(\psi)\leq\Lambda(|\psi|)$ in \eqref{cumulant_bounds}.
	\item For related models without disorder one may restrict to constant variational functions without losing generality as has been shown in \cite{D2009,CCIL2008}. In particular, for the quantum {\scshape CW} model (defined by \eqref{H_N} with non-random  $g_{ij}=1/\sqrt{N}$) this observation provides one, but not the simplest, rigorous approach to its well-known macroscopic free energy and to the equation $\tanh(\beta\bb)=\bb/\vv$ of its critical line \cite{BMT1966,S1986}.
\end{enumerate}
\end{remark}
\section{The macroscopic free energy and absence of spin-glass order for weak disorder}

In this section we are going to prove that for weak disorder, more precisely for any $4\lambda(=\beta^{2}\vv^{2})$ in the open unit interval ${\rbrack0,1\lbrack}$ and any $\beta\bb>0$, the free energy $f_N$ coincides almost surely with the annealed free energy $f^\ann_N$ in the macroscopic limit $N\to\infty$. We begin by comparing the first and the second moment of the partition function with respect to the {\scshape Gauss}ian disorder average.
By the positivity of general variances we know that $\big(\EE[Z_N]\big)^2\leq \EE\big[(Z_N)^2\big]$. 
In the present case of \eqref{H_N} we also have 
\begin{equation}
		\EE\big[(Z_{N})^{2}\big]\leq c\big(\EE[Z_{N}]\big)^{2}\quad\T{ with }\quad c\ceq\frac{\e^{-2\lambda}}{\sqrt{1-4\lambda}}>1\label{Z_N_secmom_upbound}
\end{equation}
provided that $4\lambda<1$. This is a special case of the following lemma, which in its turn is an extension of \cite[Lem.\,11.2.3]{T2011b} for the zero-field {\scshape SK} model to the present (quantum) case with a transverse field. For its formulation we recall definition \eqref{Q_N} and introduce three ``tensor expectations''. We write $\langle\,(\cdot)\,\rangle^{\otimes}_{\beta\bb}$ for the joint (conditional) expectation with respect to the given set $\{{\cal N}_{1},\dots,{\cal N}_{N}\}$ of {\scshape Poisson} processes and an independent copy (or replica) $\{\widehat{\cal N}_{1},\dots,\widehat{\cal N}_{N}\}$ thereof.
The joint {\scshape Gibbs} expectation $\langle\,(\cdot)\,\rangle^{\otimes}$ corresponding to the duplicated quantum {\scshape SK} model with {\scshape Hilbert} space $\C^{2^{N}}\!\!\!\otimes\C^{2^{N}}$ and {\scshape Hamilton}ian defined as the sum of $H_{N}\otimes\um$, see \eqref{H_N}, and a copy $\um\otimes H_{N}$ thereof (with spin operators $\widehat{S^{\alpha}_{i}}$, but the \emph{same} random variables $(g_{ij})_{1\leq i<j\leq N}$ and parameters $\bb$, $\vv$) then, in the spin-flip process representation, takes the form
\begin{equation}\label{Gibbs_expectation}
\frac{\big(Z_{N}\big)^{2}\llangle[\big]\,(\cdot)\,\rrangle[\big]^{\!\!\otimes}_{\beta\bb}}{\big(\cosh(\beta\bb)\big)^{2N}}\ceq\sum_{s,\widehat{s}}\Big\langle\exp\Big(-\beta\int_0^1\dl{t} \big[h_{N}\big(s\sigma(t)\big)+h_{N}\big(\widehat{s}\,\widehat{\sigma}(t)\big)\big]\Big)\big(\,\cdot\,\big)\Big\rangle_{\beta\bb}^{\!\otimes}\,.
\end{equation}
Two simple examples for the expectation $\llangle[\big]\,(\cdot)\,\rrangle[\big]^{\!\!\otimes}_{\beta\bb}$, revealing the (dynamical) independence and symmetry between the original  {\scshape SK} model and its copy, are given by
\begin{equation} \label{eg1}
\llangle[\Big]\int_{0}^{1}\dl{t}\int_{0}^{1}\dl{t'} s_{i}\sigma_{i}(t)\widehat{s}_{i}\widehat{\sigma}_{i}(t')\rrangle[\Big]^{\!\!\otimes}_{\beta\bb}=\langle S^{\z}_{i}\widehat{S^{\z}_{i}}\rangle^{\otimes}=  \langle S^{\z}_{i}\rangle\langle\widehat{S^{\z}_{i}}\rangle                 =\big(\langle S^{\z}_{i}\rangle\big)^{2}=0
\end{equation}
and
\begin{equation}\label{eg2}
\llangle[\Big]\int_{0}^{1}\dl{t}\int_{0}^{1}\dl{t'} s_{i}\sigma_{i}(t)s_{j}\sigma_{j}(t)\widehat{s}_{i}\widehat{\sigma}_{i}(t')\widehat{s}_{j}\widehat{\sigma}_{j}(t')\rrangle[\Big]^{\!\!\otimes}_{\beta\bb	}=\langle S^{\z}_{i}S^{\z}_{j}\widehat{S^{\z}_{i}}\widehat{S^{\z}_{j}}\rangle^{\otimes}=\big(\langle S^{\z}_{i}S_{j}^{\z}\rangle\big)^{2}
\end{equation}
for $i,j\in\{1,\dots,N\}$. The last equality in \eqref{eg1} is due to the identities $U_{N}S^{\z}_{i}U_{N}^{*}=-S^{\z}_{i}$ and $U_{N}H_{N}U_{N}^{*}=H_{N}$ with the unitary operator $U_{N}\ceq\exp\big(\i\frac{\pi}{2}\sum_{n=1}^{N}S^{\x}_{n}\big)=(\i S^{\x})^{\otimes N}$ on $\C^{2^{N}}$.
\begin{lemma}[Controlling a generalized second moment of\,\,$Z_{N}$ by its first moment]\label{comparison}

\noindent
For any $N\geq2$, $\lambda>0$, and $a\geq0$ with $4a\lambda<1$ we have
\begin{equation}\label{Q_N_bound}
	\EE\Big[\big(Z_{N}\big)^{2}\llangle[\Big]\exp\big(2N(a-1)\lambda R_{N}\big)\rrangle[\Big]^{\!\!\otimes}_{\beta\bb}\Big]\leq\frac{\e^{-2\lambda}}{\sqrt{1-4a\lambda}}\,\big(\EE[Z_{N}]\big)^{2}\,,
\end{equation}
with the ${[0,1]}$-valued random variable $R_{N}\ceq\int_{0}^{1}\dl{t}\int_{0}^{1}\dl{t'}\big[Q_{N}\big(s\sigma(t),\widehat{s}\,\widehat{\sigma}(t')\big)\big]^{2}$, see \eqref{Q_N}.
\end{lemma}

\begin{proof}
Throughout the proof we will, without mention, repeatedly interchange the order of various integrations according to the {\scshape Fubini--Tonelli} theorem.
In a first step, we observe the following identity for the {\scshape Gauss}ian disorder average
\begin{equation}
	\e^{4\lambda}\EE\Big[\exp\Big(-\beta\int_0^1\dl{t} \big[h_{N}\big(s\sigma(t)\big)+h_{N}\big(\widehat{s}\,\widehat{\sigma}(t)\big)\big]\Big)\Big]
	={\cal Z}_{N}(\sigma){\cal Z}_{N}(\widehat{\sigma})\exp(2N\lambda R_{N})\,.\label{Z_N_secmom2}
\end{equation}
Here we have used \eqref{h_N_mean}, \eqref{h_N_cov}, \eqref{Phi_N3}, and \eqref{Phi_N1}. The left hand-side (\textsf{LHS}) of \eqref{Q_N_bound} can therefore be written as
\begin{equation}\label{R_N}
\textsf{LHS}=\e^{-4\lambda}\big(\cosh(\beta\bb)\big)^{2N}\Big\langle{\cal Z}_{N}(\sigma){\cal Z}_{N}(\widehat{\sigma})\sum\limits_{s,\widehat{s}}\exp\big(2Na\lambda R_{N}\big)\Big\rangle^{\!\otimes}_{\!\beta\bb}\,.{}
\end{equation}
In a second step, we use the {\scshape Jensen} inequality
\begin{equation}\label{jensen2}
	\exp\big(2Na\lambda R_{N}\big)\leq\int_0^1\dll{t}\int_0^1\dll{t'}\exp\Big(2Na\lambda\big[Q_{N}\big(s\sigma(t),\widehat{s}\,\widehat{\sigma}(t')\big)\big]^{2}\Big)
\end{equation}
and the linearization formula
\begin{align}
\exp\Big(2Na\lambda\big[Q_{ N}\big(s\sigma(t),\widehat{s}\,\widehat{\sigma}(t')\big)\big]^{2}\Big)&=\int_{\R}\dl{x}w_{N}(x)\exp\Big(x\sqrt{8a\lambda}\sum\limits_{i=1}^Ns_{i}\sigma_{i}(t)\widehat{s}_{i}\widehat{\sigma}_{i}(t')\Big)\label{linear}
\end{align}
with the  {\scshape Gauss}ian probability density $w_{N}$ given by $w_{N}(x)=\sqrt{N/\pi}\,\exp\big(-Nx^{2}\big)$, as in Remark\,\ref{remthm1}\ref{rem_W_N}.
By combining \eqref{R_N}, \eqref{jensen2}, and \eqref{linear} we get
\begin{align}
	\textsf{LHS}&\leq\e^{-4\lambda}\big(\cosh(\beta\bb)\big)^{2 N}\int_{\R}\dl{x}w_{N}(x)\int_0^1\dll{t}\int_0^1\dll{t'}\nonumber\\
	&\quad\times\Big\langle{\cal Z}_{N}(\sigma){\cal Z}_{N}(\widehat{\sigma})\sum_{s,\widehat{s}}\prod_{i=1}^{N}\exp\Big(x\sqrt{8a\lambda}\,s_{i}\sigma_{i}(t)\widehat{s}_{i}\,\widehat{\sigma}_{i}(t')\Big)\Big\rangle^{\!\otimes}_{\!\beta\bb}\,.\label{Z_N_secmom_upbound1}
\end{align}
By observing the identities
\begin{equation}
\sum_{s,\widehat{s}}\prod_{i=1}^{N}\exp\big(\dots\big)=\prod_{i=1}^{N}\sum_{s_{i},\widehat{s}_{i}}\exp\big(\dots\big)=\prod_{i=1}^{N}4\cosh(x\sqrt{8a\lambda})=\big(4\cosh(x\sqrt{8a\lambda})\big)^{N}
\end{equation}
and \eqref{Z_N_mean}, the inequality \eqref{Z_N_secmom_upbound1} takes the simpler form
\begin{align}
	\textsf{LHS}&\leq \e^{-4\lambda}\big(2\cosh(\beta\bb)\big)^{2 N}\label{simpler}
	\big(\big\langle{\cal Z}_{N}(\sigma)\big\rangle_{\!\beta\bb}\big)^{2}\int_{\R}\dl{x}w_{N}(x)\big(\cosh(x\sqrt{8a\lambda})\big)^{N}\\
	&=\e^{-2\lambda}\big(\EE\big[Z_{N}\big]\big)^{2}\int_{\R}\dl{x}w_{N}(x)\big(\cosh(x\sqrt{8a\lambda})\big)^{N}\,.
\end{align}
By the crude inequalities $(\cosh(y))^{N}\leq\exp(N|y|)\leq\exp(Ny)+\exp(-Ny)$ for $y\in\R$ the last integral is seen to be bounded from above by $2\exp(2Na\lambda)$ for arbitrary $a\geq0$. If $4a\lambda<1$, then it even has the $N$-independent upper bound $1/\sqrt{1-4a\lambda}$ as claimed in \eqref{Q_N_bound}. It is due to the inequality $\cosh(y)\leq \exp\big(y^{2}\!/2\big)$ mentioned at the end of the proof of Lemma\,\ref{quasi_classical_lemma}.\qed
\end{proof}
\begin{remark}
For the zero-field {\scshape SK} model equality holds in \eqref{jensen2} and hence in \eqref{simpler}, because $\bb=0$ implies $\sigma_{i}(t)=1$ for all $t\in{\lbrack0,1\rbrack}$ and all $i\in\{1,\dots,N\}$, see Remark \ref{remthm1}\,\ref{zero-field}. 
\end{remark}
Inequality \eqref{Q_N_bound} will be applied with a suitable  $a > 1$ in the proof of Corollary\,\ref{no_sg} below. The choice $a=0$ leads to equality in \eqref{Q_N_bound}, see \eqref{R_N} and/or \eqref{jensen2}.The special case $a=1$, see \eqref{Z_N_secmom_upbound}, is the main ingredient for the proof of the next theorem. This theorem and its two corollaries extend two of the pioneering results of {\scshape Aizenman}, {\scshape Lebowitz}, and {\scshape Ruelle} in \cite{ALR1987} for the zero-field {\scshape SK} model, see also \cite {FZ1987,CN1995} and \cite[Ch.\,11]{T2011b}, to the present (quantum) model with a transverse field of arbitrary strength $\bb>0$.

\begin{theorem}[The macroscopic quenched free energy for weak disorder]\label{difference}

\noindent
If $4\lambda<1$, then the macroscopic limit of the quenched free energy exists and is given by that of the annealed free energy, in symbols
	\begin{equation}\label{diff_zero}
		\lim\limits_{N\to\infty}\EE[f_{N}]=\lim\limits_{N\to\infty} f_{N}^{\ann}=f^{\ann}_{\infty}\,.
	\end{equation}
	\end{theorem}

\begin{proof}
By Theorem\,\ref{f_ann_exist} it is sufficient to show that the (positive) difference $\Delta_N\ceq\EE[f_N]- f^{\ann}_N$ tends to 0 as $N\to\infty$. In order to do so we adopt the so-called second-moment method as applied in \cite[Ch.\,11]{T2011b} to the zero-field {\scshape SK} model. For this method to work we build on the large-deviation estimate of Lemma\,\ref{gc} in Appendix\,\ref{wsa} and on the elementary {\scshape Paley}--{\scshape Zygmund} inequality \cite{PZ1932} (see also \cite[Lem.\,4.1]{K2002})
	\begin{equation} \label{pz}
		 (1-q)^{2}\frac{\big(\EE[X]\big)^{2}}{\EE\big[X^{2}\big]}\leq \PP\big\{X\geq q\EE[X]\big\}
	\end{equation}
for any $[0,\infty[$-valued random variable $X$ with $\EE[X]\in{\rbrack0,\infty\lbrack}$\, and for any $q\in{\lbrack0,1\rbrack}$. Here $\PP$ denotes the probability measure underlying the (disorder) expectation $\EE$.

We begin by rewriting the given (non-random) difference as the sum of two random differences
\begin{equation}
	0\leq\Delta_N=f_N-f^{\ann}_N+\EE[f_N]-f_N\leq f_N-f^{\ann}_N+\big|f_N-\EE[f_N]\big|\label{delta1}\,.
\end{equation}
Next we show that there exist constants $\varepsilon>0$ (independent of $N$) and $\gamma_N>0$ (with $\gamma_N\downarrow 0$ as $N\to\infty$) such that the probability of finding the right-hand side of \eqref{delta1} to be smaller than $\gamma_N$, is larger than $\varepsilon$. This then yields $\Delta_N\leq\gamma_N$ and hence $\lim_{N\to\infty}\Delta_N=0$.

In fact, with an (initially) arbitrary energy $\delta>0$ we estimate as follows:
\begin{align}
	1/(4c)-&2\exp\big(-\delta^2/(2\vv^2)\big)\nonumber\\
\leq&\,\PP\big\{f_N-f^{\ann}_N\leq\ln(2)/(\beta N)\big\}+\PP\big\{\big|f_N-\EE[f_N]\big|\leq\delta/\sqrt{N}\big\}-1\label{pzwsa}\\
\leq&\,\PP\big\{f_N-f^{\ann}_N\leq\ln(2)/(\beta N)\T{ and }\big|f_N-\EE[f_N]\big|\leq\delta/\sqrt{N}\big\}\label{inex}\\
\leq&\,\PP\big\{f_N-f^{\ann}_N+\big|f_N-\EE[f_N]\big|\leq\ln(2)/(\beta N)+\delta/\sqrt{N}\big\}\,.\label{subset}
\end{align}
Here \eqref{pzwsa} is due to \eqref{pz} with $X=Z_{N}$ and $q=1/2$, combined with \eqref{Z_N_secmom_upbound}, and due to the large-deviation estimate \eqref{gc_e} in Appendix\,\ref{wsa} using $N-1<N$ and replacing $\delta$ by $\beta\delta/\sqrt{N}$. For \eqref{inex} we have used the inclusion-exclusion formula for two sets/events and the fact that probabilities do not exceed the value $1$. Inequality\,\eqref{subset} is just the monotonicity of (probability) measures. Finally we choose $\delta$ so large that $\varepsilon\ceq1/(4c)-2\exp\big(-\delta^2/(2\vv^2)\big)>0$ and put $\gamma_N\ceq\ln(2)/(\beta N)+\delta/\sqrt{N}$.\qed
\end{proof}
Theorem\,\ref{difference} has two important consequences.
\begin{corollary}[Absence of spin-glass order for weak disorder]\label{no_sg}

\noindent
If $4\lambda<1$, then we have
	\begin{equation}\label{absence}
		\lim_{N\to\infty}\EE\big[\big(\langle S^{\z}_{1} S^{\z}_{2}\rangle\big)^{2}\big]=0\,.
	\end{equation}
\end{corollary}
\begin{proof}
By applying the {\scshape Jensen} inequality to the left-hand side of \eqref{Q_N_bound} with respect to the joint expectation $\EE[\llangle\,(\cdot)\,\rrangle^{\!\!\otimes}_{\beta\bb}]$ we obtain (for any $N\geq2$)
\begin{align}
	\ln\Big(&\EE\Big[\big(Z_{N}\big)^{2}\llangle[\big]\exp\big(2N(a-1)\lambda R_{N}\big)\rrangle[\big]^{\!\!\otimes}_{\beta\bb}\Big]\Big)-2\EE\big[\ln(Z_{N})\big]-2(a-1)\lambda\nonumber\\
	&\geq2N(a-1)\lambda\EE\big[\llangle R_{N}\rrangle^{\!\!\otimes}_{\beta\bb}\big]-2(a-1)\lambda\label{LHS_inequality}\\
	&=2(N-1)(a-1)\lambda\EE\Big[\llangle[\Big]\int_{0}^{1}\dl{t}\int_{0}^{1}\dl{t'}s_{1}\sigma_{1}(t)\widehat{s}_{1}\widehat{\sigma}_{1}(t')s_{2}\sigma_{2}(t)\widehat{s}_{2}\widehat{\sigma}_{2}(t')\rrangle[\Big]^{\!\!\otimes}_{\beta\bb}\Big]\label{order_parameter1}\\
	&=2(N-1)(a-1)\lambda\EE\big[\big(\langle S^{\z}_{1}S^{\z}_{2}\rangle\big)^{2}\big]\,.\label{order_parameter2}
\end{align}
For \eqref{order_parameter1} we have used spin-index symmetry and for \eqref{order_parameter2} we refer to the example \eqref{eg2}.
By combining this with \eqref{Q_N_bound} we get 
\begin{equation}\label{combi}
	(a-1)\lambda\EE\big[\big(\langle S^{\z}_{1} S^{\z}_{2}\rangle\big)^{2}\big]\leq \frac{N}{N-1}\big(\EE[\beta f_{N}]-\beta f^{\ann}_{N}\big)-\frac{4a\lambda+\ln\big(1-4a\lambda\big)}{4(N-1)}
\end{equation}
under the assumption $4a\lambda\in{\lbrack0,1\lbrack}$ of Lemma\,\ref{comparison}. For given $4\lambda\in {]0,1[}$ we now choose an arbitrary $a\in{\rbrack1,1/(4\lambda)\lbrack}\,\neq\,\emptyset$. Then the claim \eqref{absence} follows from \eqref{combi} by observing $\EE\big[\big(\langle S_{1}^{\z}S_{2}^{\z}\rangle\big)^{2}\big]\geq 0$ and Theorem\,\ref{difference}.\qed
\end{proof}
\begin{remark}\label{rem_order_parameter}
Following \cite{ALR1987}, see also \cite{PS1991,WB2004}, the left-hand side of \eqref{absence} is the mean of the \emph{spin-glass order para\-meter} in the macroscopic limit, because the pre-limit $\EE\big[\big(\langle S^{\z}_{1} S^{\z}_{2}\rangle\big)^{2}\big]$ is, by spin-index symmetry, identical to the disorder average of the ${\lbrack0,1\rbrack}$-valued random variable
\begin{equation}
	q_{N}\ceq\frac{2}{N(N-1)}\sum\limits_{1\leq i<j\leq N}\big(\langle S^{\z}_{i}S^{\z}_{j}\rangle\big)^{2}=\frac{N}{N-1}\Big\langle\Big(\frac{1}{N}\sum_{i=1}^{N}S^{\z}_{i}\widehat{S}^{\z}_{i}\Big)^{2}\Big\rangle^{\!\!\otimes}-\frac{1}{N-1}\, ,
\end{equation}
using the (squared) quantum analog of the overlap \eqref{Q_N} and its {\scshape Gibbs} expectation $\big\langle(\,\cdot\,)\big\rangle^{\!\!\otimes}$  induced by the model \eqref{H_N} upon duplication. In the spin-flip-process representation this identity takes the form \eqref{order_parameter1}  (for  $a\neq 1$). 
By $q_{N}^2\leq q_{N}$ and Corollary\,\ref{no_sg} also the variance $\EE\big[\big(q_{N}-\EE[q_{N}]\big)^2\big]$ of $q_{N}$ is seen to vanish as $N\to\infty$. Since $|q_{N}-\EE[q_{N}]|\leq 1$, the dominated-convergence theorem establishes the $\PP$-almost-sure relation $\lim_{N\to\infty}q_{N}=0$, that is, the self-averaging property of the sequence $(q_{N})_{N\geq 2}$. In particular, the distribution of $q_{N}$ converges weakly to the {\scshape Dirac} distribution at $0\in\R$ as $N\to\infty$ if $4\lambda <1$, confer \cite[Thm.\,5.1]{B1996}. 
In the physics literature this asymptotic concentration is often interpreted as absence of replica-symmetry breaking for high temperatures.
\end{remark} 
\begin{corollary}[The macroscopic free energy for weak disorder]\label {weak_disorder_free_energy}

\noindent
If $4\lambda<1$, then the sequence of random variables $(f_N)_{N\geq2}$ defined by \eqref{f_N} converges almost surely to the macroscopic limit \eqref{f_ann_inf} of the annealed free energy, in symbols
	\begin{equation}\label{free_energy_limit}
		\lim\limits_{N\to\infty}f_{N}=\lim\limits_{N\to\infty}f_{N}^{\ann}=f^{\ann}_{\infty}\qquad (\PP\T{-almost surely}).
	\end{equation}
\end{corollary}
\begin{proof}
	By Theorem\,\ref{f_ann_exist} it sufficies to show that $\lim_{N\to\infty}|f_N-f^{\ann}_N|=0$, almost surely. By the triangle inequality and by \eqref{jensen} we have
	\begin{equation}\label{triangle}
	|f_N-f^{\ann}_N|\leq \EE[f_N]-f^{\ann}_N+\big|f_N-\EE[f_N]\big|
	\end{equation}
By Theorem\,\ref{difference} the first difference on the right-hand side tends to $0$ as $N\to\infty$. Moreover, the large-deviation estimate \eqref{gc_e} in Appendix\,\ref{wsa} implies the summability
\begin{equation}\label{summability}
\sum\limits_{N=2}^{\infty}\PP\big\{\big|\beta f_N-\EE[\beta f_N]\big|> \delta\big\}\leq \frac{2a^{2}}{1-a}<\infty\,,\qquad a\ceq\e^{-\delta^{2}/(8\lambda)}
\end{equation}
for any $\delta>0$. A simple and standard application \cite[\S11, Example\,1]{B1996} of the easy part of the {\scshape Borel--Cantelli} lemma (see \cite[Lem.\,11.1]{B1996} or \cite[Thm.\,3.18]{K2002}) now shows that also the second difference in \eqref{triangle} tends to zero, $\PP$-almost surely.\qed 
\end{proof}

\begin{remark}
In the above proof we have used the fact that the summability \eqref{summability} implies the almost-sure relation $\lim_{N\to\infty}(f_{N}-\EE[f_{N}])=0$. Clearly, the summability and hence the relation hold for arbitrary $\lambda>0$. It may be dubbed as ``self-averaging in the mean'' of the sequence $(f_{N})_{N\geq 2}$. The physically indispensable \emph{self-averaging} (or \emph{ergodicity}) in the sense of the almost-sure relation $\lim_{N\to\infty}f_{N}= \lim_{N\to\infty}\EE[f_{N}]$ additionally requires the existence of one of the latter limits. 
Until now the model \eqref{H_N} seems to be the only quantum mean-field spin-glass model for which the second limit is known to exist. For arbitrary $\lambda >0$ this is due to {\scshape Crawford} \cite{C2007}. Theorem\,\ref{difference} above provides for $4\lambda<1$ a (variational) formula for the limit and therefore its existence for the weak-disorder regime as a by-product, similarly as in \cite{ALR1987} for the case $\bb=0$. In view of the complexity of the {\scshape Parisi} formula \cite{P1980a,P1980b,D1981,T2006,P2009,T2011b,P2013,AC2015}\footnote{The last three references also contain results for classical {\scshape SK} models with multi-spin interactions.}, even for vanishing longitudinal field, we conjecture a more complicated (variational) formula to hold for $4\lambda\geq 1$ and $\bb>0$. 
\end{remark}

\section{Concluding remarks and open problems}\label{conclusion}
The present paper contains the first rigorous explicit results on the thermostatics of the quantum {\scshape Sherrington--Kirkpatrick} spin-glass model \eqref{H_N} for the regime $\beta\vv<1$. Unfortunately, the opposite (and more important) regime remains not nearly as well understood as in the ``classical limit'' $\bb\downarrow 0$, at least from a rigorous point of view. Over the 35 years of research several investigators have provided stimulating and possibly correct results by approximate arguments and/or numerical methods. But for low temperatures these results are typically less reliable, for example due to the unjustified interchange of various limits and/or because of too small ``{\scshape Lie--Trotter} numbers''. Therefore one should find rigorous arguments for the actual shape of the (red) dashed line in Figure\,\ref{regionplot}. In view of Remark\,\ref{remthm3}\,\ref{landau} it is tempting to conjecture that the assertions \eqref{diff_zero}, \eqref{absence}, and \eqref{free_energy_limit} remain true under the ($\bb$-dependent) condition $\beta\vv<1/m$. This would enlarge the heavy gray region in Figure\,\ref{regionplot} slightly beyond the vertical line $\beta\vv=1$ and help to ``localize'' the critical line somewhat further. In any case, the precise determination of this line is a demanding problem.
A similar challenge is to aspire after the analog of the {\scshape Parisi} formula for the macroscopic (quenched) free energy of the quantum {\scshape SK} model \eqref{H_N}. A first step in this direction has been achieved recently by {\scshape Adhikari} and {\scshape Brennecke} \cite{AB2020}, see the end of Section\,\ref{introduction}. At present we do not know, how their variational formula reduces to our \eqref{varadhan} if $\beta\vv<1$.  Unfortunately, for $\beta\vv\geq 1$ we only have the inequalities \eqref{quasi_classical} which may be used to bound the free energy of \eqref{H_N} from below and above in terms of the zero-field {\scshape Parisi} formula. Nevertheless, it could be that the quantum analog of the {\scshape Parisi} formula is in certain respects simpler than the classical one \emph{because} of quantum fluctuations, confer \cite{RCC1989,BU1990b,MRC2018}.

\appendix

\section{The positivity of certain {\scshape Poisson}-process covariances}\label{poisson_appendix}

For the proofs of \eqref{mu}, \eqref{multipoint_inequality}, and related facts it is convenient to consider {\scshape Poisson} (point) processes being more general than the one used in the main text (see, for example, \cite{K2002,LP2018,K1993}). A {\scshape Poisson} process in a general sigma-finite measure space $(\Gamma,{\cal A},\rho)$ is a random measure $\nu$ on $(\Gamma,{\cal A})$. The distribution of $\nu$ is uniquely defined, in terms of the (positive) measure $\rho$, by the elegant and powerful formula 
\begin{equation}\label{campbell}
	\Big\langle\exp\Big(-\int_{\Gamma}\nu(\dll{x})f(x)\Big)\Big\rangle=\exp\Big(-\int_{\Gamma}\rho(\dll{x})\big(1-\e^{-f(x)}\big)\Big)
\end{equation}
for its {\scshape Laplace} functional, which dates back to {\scshape Campbell} \cite{C1909}. Here and in Appendix \ref{pfk_proof} the angular brackets $\langle\,(\cdot)\,\rangle$ denote the expectation with respect to the probability measure steering the randomness of $\nu$ and $f:\Gamma\to{[0,\infty[}$ is an arbitrary measurable function into the positive half-line. For $f=a\chi_{A}$ with $a\in {[0,\infty[}$, $A\in\cal{A}$, and  $\rho(A)<\infty$ the right-hand side of \eqref{campbell} equals the {\scshape Laplace} transform of the {\scshape Poisson} distribution with mean $\rho(A)$. Hence the random variable $\nu(A)$ is $\N_{0}$-valued and {\scshape Poisson} distributed with mean $\langle\nu(A)\rangle=\rho(A)$. In words, the mean number of {\scshape Poisson} points lying in $A$ equals its  $\rho$-measure. By choosing $f=\sum_{j=1}^{m}a_{j}\chi_{A_{j}}$ it also follows from \eqref{campbell} that the random variables $\nu(A_{1}),\dots,\nu(A_{m})$ are independent for (pairwise) disjoint sets $A_{1},\dots,A_{m}\in{\cal A}$ of finite $\rho$-measures for all  $m\in\N$. Finally, we note that \eqref{campbell} remains valid when $f$ is replaced by the imaginary function $\i f$ with $f:\Gamma\to\R$ obeying $\int_{\Gamma}\rho(\dll{x})\min\{|f(x)|,1\}<\infty$. 

In the main text we are mainly interested in $\{-1,1\}$-valued random variables corresponding to $\sigma(A)\ceq(-1)^{\nu(A)}$ with $A\in{\cal A}$ obeying $\rho(A)<\infty$. By choosing $f=\i\pi\sum_{j=1}^{m}\chi_{A_{j}}$ with an arbitrary collection of $m\in\N$ such sets, $\widehat{A}\ceq\{A_{1},\dots,A_{m}\}\subset{\cal A}$, we get from \eqref{campbell}
\begin{equation}\label{prod_exp}
		\big\langle\sigma(\widehat{A})\big\rangle=\exp\Big(-\int_{\Gamma}\rho(\dll{x})\big(1-\tau_{\widehat{A}}(x)\big)\Big)
\end{equation}
in terms of the $\{-1,1\}$-valued products $\sigma(\widehat{A})\ceq\prod_{j=1}^{m}\sigma(A_{j})$ and  $\tau_{\widehat{A}}\ceq\prod_{j=1}^{m}\big(1-2\chi_{A_{j}}\big)$. In particular, we have $\langle\sigma(A_{j})\rangle=\exp\big(-2\rho(A_{j})\big)$ by choosing $A_{k}=\emptyset$ for all $k\neq j$. If $\widehat{B}\ceq\{B_{1},\dots,B_{n}\}\subset{\cal A}$ is another arbitrary collection of $n\in\N$ such sets, we obtain the positive covariance
\begin{equation}\label{uncond_positivity}
	\big\langle\sigma(\widehat{A})\sigma(\widehat{B})\big\rangle\geq\big\langle\sigma(\widehat{A})\big\rangle\big\langle\sigma(\widehat{B})\big\rangle
\end{equation}
by \eqref{prod_exp}, the pointwise inequality $\tau_{\widehat{A}}\tau_{\widehat{B}}\geq\tau_{\widehat{A}}+\tau_{\widehat{B}}-1$, and the functional equation of the exponential. A simple consequence of \eqref{uncond_positivity} by iteration is
\begin{equation}
	\big\langle\sigma(\widehat{A})\big\rangle\geq\prod_{j=1}^{m}\big\langle\sigma(A_{j})\big\rangle=\exp\Big(-2\sum_{j=1}^{m}\rho(A_{j})\Big)>0\,.
\end{equation}

As in the main text we are going to introduce a conditional {\scshape Poisson} expectation. For a fixed $\Lambda\subseteq\Gamma$ with $\Lambda\in{\cal A}$ and $\rho(\Lambda)<\infty$ we write the two {\scshape Kronecker} deltas $\delta_{\sigma(\Lambda),\pm1}$ as $\delta_{\sigma(\Lambda),\pm1}=\big(1\pm\sigma(\Lambda)\big)/2$. The {\scshape Poisson} expectation conditional on $\sigma(\Lambda)=1$, equivalently on even $\nu(\Lambda)$, can therefore be written as
\begin{equation}\label{cond_expectation}
	\big\langle\,(\,\cdot\,)\,\big\rangle_{\!\Lambda}\ceq\frac{\langle\delta_{\sigma(\Lambda),1}\,(\,\cdot\,)\rangle}{\langle\delta_{\sigma(\Lambda),1}\rangle}=\frac{\langle\,(\,\cdot\,)\,\rangle+\langle\sigma(\Lambda)\,(\,\cdot\,)\rangle}{1+\e^{-2\rho(\Lambda)}}\,.
\end{equation}
By \eqref{cond_expectation} and  \eqref{uncond_positivity} we immediately see that $\big\langle\sigma(\widehat{A})\big\rangle_{\!\Lambda}\geq\big\langle\sigma(\widehat{A})\big\rangle$. The ``conditional analog" of \eqref{uncond_positivity} is 
\begin{lemma}
	\begin{equation}\label{cond_positivity}
		\big\langle\sigma(\widehat{A})\sigma(\widehat{B})\big\rangle_{\!\Lambda}\geq\big\langle\sigma(\widehat{A})\big\rangle_{\!\Lambda}\big\langle\sigma(\widehat{B})\big\rangle_{\!\Lambda}\,.
	\end{equation}
\end{lemma}

\begin{proof}
From \eqref{prod_exp} and  \eqref{cond_expectation} we get the ``conditional analog" of \eqref{prod_exp}
\begin{equation}\label{prod_exp_cond}
\big\langle\sigma(\widehat{A})\big\rangle_{\!\Lambda}=\frac{\cosh\big(I_{\Lambda}(\widehat{A})\big)}{\cosh\big(\rho(\Lambda)\big)}\exp\Big(-\int_{\Gamma\setminus\Lambda}\rho(\dll{x})\big(1-\tau_{\widehat{A}}(x)\big)\Big)
\end{equation}
with $I_{\Lambda}(\widehat{A})\ceq\int_{\Lambda}\rho(\dll{x})\tau_{\widehat{A}}(x)$. Corresponding formulas hold for $ \big\langle\sigma(\widehat{B})\big\rangle_{\!\Lambda}\,$
and $\big\langle\sigma(\widehat{A})\sigma(\widehat{B})\big\rangle_{\!\Lambda}$. In the latter case $\tau_{\widehat{A}}$ has to be replaced with the product $\tau_{\widehat{A}}\tau_{\widehat{B}}$ and $I_{\Lambda}(\widehat{A})$ with $I_{\Lambda}(\widehat{A},\widehat{B})\ceq\int_{\Lambda}\rho(\dll{x})\big(1-\tau_{\widehat{A}}(x)\tau_{\widehat{B}}(x)\big)$.
For the proof of \eqref{cond_positivity} we firstly employ again the above inequality $\tau_{\widehat{A}}\tau_{\widehat{B}}\geq\tau_{\widehat{A}}+\tau_{\widehat{B}}-1$. Then it remains to show that $\cosh\big(\rho(\Lambda)\big)\cosh\big(I_{\Lambda}(\widehat{A},\widehat{B})\big)\geq\cosh\big(I_{\Lambda}(\widehat{A})\big)\cosh\big(I_{\Lambda}(\widehat{B})\big)$. To this end, we use the hyperbolic relation $2\cosh\big(\rho(\Lambda)\big)\cosh\big(I_{\Lambda}(\widehat{A},\widehat{B})\big)=\cosh\big(\rho(\Lambda)+I_{\Lambda}(\widehat{A},\widehat{B})\big)+\cosh\big(\rho(\Lambda)-I_{\Lambda}(\widehat{A},\widehat{B})\big)$. Therefore the two inequalities $\rho(\Lambda)\pm I_{\Lambda}(\widehat{A},\widehat{B})\geq|I_{\Lambda}(\widehat{A})\pm I_{\Lambda}(\widehat{B})|$, based on the pointwise equalities $1\pm\tau_{\widehat{A}}\tau_{\widehat{B}}=|\tau_{\widehat{A}}\pm\tau_{\widehat{B}}|$ and the triangle inequality, combined with the relation $\cosh(|y|)=\cosh(y)$ for $y\in\R$ complete the proof.\qed
\end{proof}
\begin{remark}\label{appendix_A_remarks}

\renewcommand{\labelenumi}{(\roman{enumi})}
\renewcommand{\theenumi}{(\roman{enumi})}

\begin{enumerate}
	\item If $A_{j} \subseteq\Lambda$ for all $j$, then the exponential factor in \eqref{prod_exp_cond} takes its maximum value $1$. For example, in the case of two such sets, say $A$ and $B$, we simply have 
	\begin{equation}\label{two_sets}   
		\big\langle\sigma(A)\sigma(B)\big\rangle_{\!\Lambda}=\cosh \big(\rho(\Lambda)-2\rho(A)-2\rho(B)+4\rho(A\cap B)\big)/\cosh\big(\rho(\Lambda)\big).            
	\end{equation}  
	The numerator further simplifies to $\cosh\big(\rho(\Lambda)-2\big|\rho(A)-\rho(B)\big|\big)$ if $A\subseteq B$ or $B\subseteq A$.
 	\item \label{special_case_main_text}In the main text and in Appendix\,\ref{pfk_proof} we only consider the special case corresponding to $\Gamma = {[0,\infty[}$, ${\cal A} ={\cal B}\big({[0,\infty[}\big)\ceq${\scshape Borel} sigma-algebra in ${[0,\infty[}$, $\rho=\beta\bb\times${\scshape Lebesgue} measure, and $\Lambda = {[0,1]}$. There we write ${\cal N}(t)$ and $\sigma(t)=(-1)^{{\cal N}(t)}=(-1)^{-{\cal N}(t)}$  instead of $\nu({[0,t]})$ and $\sigma ({[0,t]})$, respectively, for any $t\in {[0,\infty[}$. We also write  $\langle\,(\,\cdot\,)\,\rangle_{\beta\bb}$  instead of $\langle\,(\,\cdot\,)\,\rangle_{{[0,1]}}$. It is well-known that the stochastic process $\big\{{\cal N}(t):t\in{[0,\infty[}\big\}$ has independent and, in distribution, time-homogeneous increments ${\cal N}(t+u)-{\cal N}(u)=\nu\big({]t,t+u]}\big)$ for $u\geq0$. This implies that it is a {\scshape Markov} process, more precisely, a continuous-time homogeneous {\scshape Markov} chain with transition probabilities
	\begin{equation}
 	p_{n,n'}(t,t')\ceq\e^{-\beta\bb(t-t')}\frac{\big(\beta\bb(t-t')\big)^{n-n'}}{(n-n')!}\qquad \big(0< t'\leq t,\quad n,n'\in\N_{0},\quad n'\leq n\big).
	\end{equation}
 	Also the spin-flip process $\big\{\sigma(t):t\in{[0,\infty[}\big\}$ is such a {\scshape Markov} process. Its transition probabilities are $p_{s,s'}(t,t')\ceq\big[1+ss'\exp\big(-2\beta\bb(t-t')\big)\big]\big/2$ with $s,s'\in{\{-1,1\}}$.
	
\end{enumerate}
\end{remark}

\section{The {\scshape Poisson--Feynman--Kac} formula}\label{pfk_proof}

In this appendix we consider an independent collection of $N\in\N$ {\scshape Poisson} processes in the positive half-line ${[0,\infty[}$ with common rate $\beta\bb>0$ in the sense and notation of Remark\,\ref{appendix_A_remarks}\,\ref{special_case_main_text}. We begin with the case of a single spin. Here $\gg\in\R$ is an arbitrary parameter. The rest of the notation has been introduced in Section\,\ref{introduction}.

\begin{lemma}[Operator-valued {\scshape PFK} formula for a single spin]\label{pfk_1_lemma}

\noindent
For a single spin we have the operator identity
\begin{equation}\label{PFK_single_spin}
	\exp\big(\beta\bb S^{\x}_{i}+\beta \gg S^{\z}_{i}\big) = \e^{\beta\bb}\Big\langle\big(S^{\x}_{i}\big)^{{\cal N}_{i}(1)}\exp\Big(\beta \gg S^{\z}_{i}\int_{0}^{1}\dl{t}\sigma_{i}(t)\Big)\Big\rangle\qquad \big(i \in\{1,\dots, N\}\big)
\end{equation}	
on the $N$-spin {\scshape Hilbert} space $(\C^2)^{\otimes N}\cong\C^{2^N}$.
\end{lemma}
\begin{proof}
It is enough to prove \eqref{PFK_single_spin} on the single-spin {\scshape Hilbert} space $\C^{2}$. So we suppress the spin index $i$. For the auxiliary operator $K_{\gg}(u,t)\ceq\exp\big(\int_{t}^{u}\,\dl{t'}\beta\gg S^{\z} \sigma(t')\big)$ with $0\leq t\leq t'\leq u$ we write
\begin{equation}\label{K_g_alt}
	K_{\gg}(u,0)=\um-\int_{0}^{u}\dl{t}\frac{\dl{}}{\dl{t}}K_{\gg}(u,t)=\um+\int_{0}^{u}\dl{t}K_{\gg}(u,t)\beta\gg S^{\z}\sigma(t)
\end{equation}
and define the operator
\begin{equation}\label{T_g_alt}
T_{\gg}(u)\ceq(S^{\x})^{{\cal N}(u)}K_{\gg}(u,0)=T_{0}(u)+\int_{0}^{u}\dl{t}(S^{\x})^{{\cal N}(u)}K_{\gg}(u,t)\beta\gg S^{\z}\sigma(t)\,.
\end{equation}
The last integrand can be rewritten as follows
\begin{equation}\label{integrand}
   (S^{\x})^{{\cal N}(u)-{\cal N}(t)}(S^{\x})^{{\cal N}(t)}K_{\gg}(u,t)\beta\gg S^{\z} \sigma(t)=(S^{\x})^{{\cal N}(u)-{\cal N}(t)}K_{\gg\sigma(t)}(u,t)\beta\gg S^{\z}(S^{\x})^{{\cal N}(t)}\,.
\end{equation}
Here we have moved $\big(S^{\x}\big)^{{\cal N}(t)}$ to the utmost right by using ${\cal N}(t)$ times the relation $S^{\x}f(S^{\z})=f(-S^{\z})S^{\x}$, based on \eqref{dirac} and the elementary spectral identity $f(A)=\um\big(f(1)+f(-1)\big)/2+A\big(f(1)-f(-1)\big)/2$ for any selfadjoint operator $A$ with $A^{2}=\um$ and any {\scshape Borel}-measurable function $f: \{-1,1\}\to\C$.
Now we assert that the {\scshape Poisson} expectation of \eqref{T_g_alt} leads to
\begin{align}\label{DDP1}
	\big\langle T_{\gg}(u)\big\rangle&=\big\langle T_{0}(u)\big\rangle + \int_{0}^{u}\dl{t}\big\langle\big(S^{\x}\big)^{{\cal N}(u)-{\cal N}(t)} K_{\gg\sigma(t)}(u,t)\big\rangle \beta\gg S^{\z}\big\langle T_{0}(t)\big\rangle\\\label{DDP2}
	&=\big\langle T_{0}(u)\big\rangle + \int_{0}^{u}\dl{t}\big\langle T_{\gg}(u-t)\big\rangle \beta\gg S^{\z}\big\langle T_{0}(t)\big\rangle\qquad (u\geq0)\,.
\end{align}
Equation\,\eqref{DDP1} relies on the fact that the increments ${\cal N}(u)-{\cal N}(t)$ as well as ${\cal N}(t')-{\cal N}(t)$ occurring in $K_{\gg\sigma(t)}(u,t)=\exp\big(\int_{t}^{u}\dl{t'}(-1)^{{\cal N}(t')-{\cal N}(t)}\beta\gg S^{\z}\big)$ are independent of ${\cal N}(t)-{\cal N}(0)= {\cal N}(t)$. For \eqref{DDP2} we recall 
that ${\cal N}(t')-{\cal N}(t)$ has the same distribution as ${\cal N}(t'-t)-{\cal N}(0)=  {\cal N}(t'-t)$ by time-homogeneity.
Since
	\begin{equation}\label{Sx_formula_alt}
		(S^{\x})^{{\cal N}(u)}=\um\delta_{\sigma(u),1}+S^{\x}\delta_{\sigma(u),-1}=\frac{1}{2}(\um+S^{\x})+\frac{1}{2}(\um-S^{\x})\sigma(u),\quad  \langle\sigma(u)\rangle= \e^{-2\beta\bb u}\,,
	\end{equation}
we see that the mapping $u\mapsto\langle T_{0}(u)\rangle=\exp\big(u\beta\bb(S^{\x}-\um)\big)$ is the ``free'' or ``unperturbed" semigroup on $\C^{2}$ (corresponding to $\gg=0$ and up to the factor $\e^{u\beta\bb}$). Consequently, the combination of \eqref{DDP1} and \eqref{DDP2} implies that $u\mapsto\langle T_{\gg}(u)\rangle$ satisfies the same {\scshape Duhamel--Dyson--Phillips} integral equation as the ``perturbed" semigroup $u\mapsto\exp\big(u\beta\bb(S^{\x}-\um)+u\beta\gg S^{\z}\big)$.
Actually, this equation is equivalent to the differential equation $\frac{\partial}{\partial u}\langle T_{\gg}(u)\rangle=\langle T_{\gg}(u)\rangle\big(\beta\bb(S^{\x}-\um)+\beta\gg S^{\z}\big)$ with the initial condition $ \langle T_{\gg}(0)\rangle=\um$. Since the solution is unique, the proof is completed by considering $\langle T_{\gg}(1)\rangle$.
\qed
\end{proof}
\begin{remark}\label{rem_pfk}

\renewcommand{\labelenumi}{(\roman{enumi})}
\renewcommand{\theenumi}{(\roman{enumi})}

	\begin{enumerate}
		\item Formula \eqref{PFK_single_spin} dates back to {\scshape Kac} \cite{K1974}. There he has not written down it explicitly, but it is the backbone of his {\scshape PFK} formula for the solution of the telegraph (or damped-wave) equation. For a modern account of this genre see \cite{KR2013} and also \cite{CD2006}.
		\item We learned the {\scshape  PFK} formula \eqref{PFK_single_spin} for a single-spin semigroup from {\scshape Gaveau} and {\scshape Schulman} \cite{GS1989} who proved it by a suitable {\scshape Lie--Trotter}(--{\scshape P.\,R.~Chernoff}) product formula. Our proof avoids time-slicing and is in the spirit of {\scshape Simon}'s ``second proof'' of the ({\scshape Wiener}--){\scshape Feynman--Kac} formula for {\scshape Schr\"odinger} semigroups \cite[Thm.\,6.1]{S2005a}, see also \cite[Sec.\,2.2]{R1994}. It easily extends to time-dependent integrable $\gg:{[0,1]}\to\R$. So does the alternative, slightly more direct, proof at the end of this appendix.
		\item\label{Laplace_transform} The PFK formula  \eqref {PFK_single_spin} is equivalent to the explicit formulas 
		\begin{equation}
			{\cal L}_{-1}(\gg)=\frac{\bb}{\ww}\sinh(\beta \ww)\,,\quad {\cal L}_{1}(\gg)=\cosh(\beta \ww)+\frac{\gg}{\ww}\sinh(\beta \ww)\,,\quad  \ww\ceq\sqrt{\bb^{2}+\gg^{2}}
		\end{equation}
		for the {\scshape Laplace} transforms of the two conditional distributions of the random variable $\beta\int_{0}^{1}\dl{t}\sigma_{i}(t)$,
		\begin{equation}
			{\cal L}_{s_{i}}(\gg)\ceq\e^{\beta\bb}\Big\langle\delta_{\sigma_{i}(1),\,s_{i}}\exp\Big(\beta \gg\int_{0}^{1}\dl{t}\sigma_{i}(t)\Big)\Big\rangle\qquad\big(s_{i}\in\{-1,1\}\,,\, i\in\{1,\dots,N\}\big)\,,
		\end{equation}
		up to a $\gg$-independent factor. This follows from \eqref{Sx_formula_alt} with $u=1$ and the special case 
		\begin{equation}\label{H}
			\exp(-\beta H)=\um\cosh(\beta\ww)-H\frac{\sinh(\beta\ww)}{\ww}\,,\quad H\ceq-\bb S^{\x}_{i}-\gg S^{\z}_{i}\,,\quad H^{2}=\um\ww^{2}\,
		\end{equation}
		of the elementary spectral identity mentioned between \eqref{integrand} and \eqref{DDP1}.
		\item Clearly, we have $\tr\exp(-\beta H)=2\cosh(\beta\ww)$ by \eqref{H}. On the other hand, by evaluating this trace in the eigenbasis of $S^{\z}_{i}$ and using 		\eqref{PFK_single_spin} we get its spin-flip representation
		\begin{equation}\label{PFK_trace}
			\tr{\e^{-\beta H}}=\e^{\beta\bb}\sum\limits_{s_{i}}\Big\langle\delta_{\sigma_{i}(1),1}\exp\Big(\beta\gg s_{i}\int_{0}^{1}\dl{t}\sigma_{i}(t)\Big)\Big\rangle=\cosh(\beta\bb)\sum\limits_{s_{i}}\Big\langle\exp\Big(\beta\gg\int_{0}^{1}\dl{t}s_{i}\sigma_{i}(t)\Big)\Big\rangle_{\beta\bb}\,.
		\end{equation}
		\item \label{autocorrelation} As a simple example of a {\scshape Gibbs} expectation value in the spin-flip representation we consider the case $\gg=0$ and the (non-selfadjoint) operator
		\begin{equation}
		S_{i}^{\z}(t)\ceq\e^{-\beta\bb\,tS_{i}^{\x}}S_{i}^{\z}\e^{\beta\bb\,tS_{i}^{x}}= 
		S_{i}^{\z}\e^{2\beta\bb\,tS_{i}^{\x}} \qquad\big(t\in{[0,\infty[}\big)\,,
		\end{equation}
		that is, the operator $S_{i}^{\z}$ in the ``imaginary-time'' {\scshape Heisenberg}\, picture. Then we have for the {\scshape Duhamel--Kubo}\, auto-correlation function of the $\z$-component of a single spin the formula
		\begin{equation}\label{Duhamel_Kubo}
			\tr\big({\e^{\beta\bb S_{i}^{\x}}S_{i}^{\z}(t)S_{i}^{\z}(t')}\big)\big/\tr{\e^{\beta\bb S_{i}^{\x}}}=\big\langle\sigma_{i}(t)\sigma_{i}(t')\big\rangle_{\beta\bb}\qquad\big(0\leq t'\leq t\leq1\big)\,.
		\end{equation}
		It is due to the (positive) product $S_{i}^{\z}(t)S_{i}^{\z}(t')=\exp{\big(2\beta\bb(t'-t)S_{i}^{\x}\big)}$ and \eqref{two_sets} for the special case used in the main text. If $t<t'$, then the factor order of the two spin operators in \eqref{Duhamel_Kubo} has to be reversed. It also follows that \eqref{multipoint_inequality} may be viewed as an inequality for multi-time correlation functions of the $\z$-component of a single spin which interacts with a transverse field only. For an odd number of instants these functions vanish by $S_{i}^{\z}$-reversal symmetry, confer the argument immediately above Lemma\,\ref{comparison}.
	\end{enumerate}
\end{remark}
	\begin{corollary}[Operator-valued {\scshape PFK} formula for several spins]\label{PFK_N_corollary}

\noindent
		Let $S^{\alpha}_{1},\dots,S^{\alpha}_{N}$ be a collection of $N\in\N$ (pairwise commuting) spin operators with component  $\alpha\in\{\x,\y,\z\}$ and
		let ${\cal N}_{1},\dots,{\cal N}_{N}$ be associated mutually independent {\scshape Poisson} processes in the positive half-line ${[0,\infty[}$ with the common rate $\beta\bb>0$. Moreover, let $\varv_{N}:\{-1,1\}^{N}\to\R$ be a {\scshape Borel}-measurable function and consider the {\scshape Boltzmann--Gibbs} operator
		\begin{equation}
			W_{N}(\beta)\ceq\exp\Big(\beta\bb\sum_{i=1}^{N}S^{\x}_{i}+\beta \varv_{N}(S_{1}^{\z},\dots,S_{N}^{\z}\big)\Big)\,.
		\end{equation}
		Then the operator identity
		\begin{equation}\label{PFK_multi_spins}
			W_{N}(\beta)=\e^{N\beta\bb}\Big\langle\Big(\prod_{i=1}^{N}\big(S_{i}^{\x}\big)^{{\cal N}_{i}(1)}\Big)\exp\Big(\beta\int_{0}^{1}\dl{t} \varv_{N}\big(S_{1}^{\z}\sigma_{1}(t),\dots,S_{N}^{\z}\sigma_{N}(t)\big)\Big)\Big\rangle
		\end{equation}
		holds on the $N$-spin {\scshape Hilbert} space $\C^{2^{N}}$.
\end{corollary}
\begin{proof}
Apply the arguments used in the proof of Lemma\,\ref{pfk_1_lemma} for the $i$-th spin also for all other spins ($j\neq i$).
	\qed
\end{proof}
\begin{remark}
\begin{enumerate}
		\item Two examples: $\quad\varv_{1}(s_{1})=\gg s_{1}\,,\quad\varv_{N}(s_{1},\dots,s_{N})=\frac{\vv}{\sqrt{N}}\sum_{1\leq i<j\leq N} g_{ij}s_{i}s_{j}\quad(N\geq2)\,.$
		\item By \eqref{PFK_multi_spins} formula~\eqref{PFK_trace} naturally extends to
		\begin{equation}\label{PFK_trace_multi}
			\tr{W_{N}(\beta)}=\big(\cosh(\beta\bb)\big)^{N}\sum\limits_{s_{1},\dots, s_{N}}\Big\langle\exp\Big(\beta\int_{0}^{1}\dl{t}\varv_{N}\big(s_{1}\sigma_{1}(t),\dots,s_{N}\sigma_{N}(t)\big)\Big)\Big\rangle_{\beta\bb}\,.
		\end{equation}
\end{enumerate}
\end{remark}

\subsection*{Alternative proof of Lemma\,\ref{pfk_1_lemma}}

\begin{proof}
Using the auxiliary operator $K_{\gg}(u,t)$ introduced in the first proof we want to show that
\begin{equation}\label{T_beta}
	u\mapsto \overline{T}_{\gg}(u)\ceq\big\langle\big(S^{\x})^{{\cal N}(u)}K_{g}(u,0)\big\rangle
\end{equation}
is an operator semigroup on $\C^{2}$ with generator $\beta\bb(S^{\x}-\um)+\beta\gg S^{\z}$. In the first step, we pick $u,t\geq0$ and get
\begin{align}\label{semigroup1}
	\overline{T}_{\gg}(u+t)&=\big\langle\big(S^{\x}\big)^{{\cal N}(u+t)}K_{\gg}(u+t,0)\big\rangle=\big\langle\big(S^{\x}\big)^{{\cal N}(u+t)-{\cal N}(t)}\big(S^{\x}\big)^{{\cal N}(t)}K_{\gg}(u+t,t)K_{g}(t,0)\big\rangle\\\label{semigroup2}
	&=\big\langle\big(S^{\x}\big)^{{\cal N}(u+t)-{\cal N}(t)}K_{\gg\sigma(t)}(u+t,t)\big(S^{\x}\big)^{{\cal N}(t)}K_{\gg}(t,0)\big\rangle\\\label{semigroup3}
	&=\big\langle\big(S^{\x}\big)^{{\cal N}(u+t)-{\cal N}(t)}K_{\gg\sigma(t)}(u+t,t)\big\rangle\overline{T}_{\gg}(t)=\overline{T}_{\gg}(u)\overline{T}_{\gg}(t)\,.
\end{align}
The first two equations are obvious. For \eqref{semigroup2} we have moved $\big(S^{\x}\big)^{{\cal N}(t)}$ to the right of $K_{\gg}(u+t,t)$ by using ${\cal N}(t)$ times the relation $S^{\x}f(S^{\z})=f(-S^{\z})S^{\x}$ known from the first proof. The first equality in \eqref{semigroup3} relies on the fact that the increments ${\cal N}(u+t)-{\cal N}(t)$ and ${\cal N}(t')-{\cal N}(t)$ occurring in $K_{\gg\sigma(t)}(u+t,t)=\exp\big(\int_{t}^{u+t}\dl{t'}(-1)^{{\cal N}(t')-{\cal N}(t)}\beta\gg S^{\z}\big)$ are independent of ${\cal N}(t)-{\cal N}(0)={\cal N}(t)$. For the second equality in \eqref{semigroup3} we recall that ${\cal N}(t')-{\cal N}(t)$ has the same distribution as ${\cal N}(t'-t)-{\cal N}(0)={\cal N}(t'-t)$ by time-homogeneity.
In the second step of this proof, we combine the definition \eqref{T_beta} with \eqref{K_g_alt} and \eqref{Sx_formula_alt} to get
\begin{equation}
	\overline{T}_{\gg}(u)-\um=\frac{1}{2}\big(1-\e^{-2\beta\bb u}\big)\big(S^{\x}-\um\big)+u\int_{0}^{1}\dl{t}\big\langle(S^{\x})^{{\cal N}(u)}K_{\gg}(u,u\,t)\sigma(u\,t)\big\rangle\beta\gg S^{\z}\,.
\end{equation}
Due to $\big\|(S^{\x})^{{\cal N}(u)}K_{\gg}(u,u\,t)\sigma(u\,t)\big\|\leq \e^{u\beta|\gg|}$ the dominated-convergence theorem gives (in the operator-norm sense) $\lim_{u\downarrow0}(\overline{T}_{\gg}(u)-\um)/u=\beta\bb(S^{\x}-\um)+\int_{0}^{1}\dl{t}\langle\um\rangle \beta\gg S^{\z}=\beta\bb(S^{\x}-\um)+\beta\gg S^{\z}$ as claimed. This completes the proof of \eqref{PFK_single_spin} by considering $\overline{T}_{\gg}(1)$.
\qed
\end{proof}

\section{Large-deviation estimate for the free energy}\label{wsa}
For the reader's convenience we begin by quoting Theorem\,1.3.4 in \cite{T2011a} without proof.
\begin{proposition}[{\scshape Gauss}ian concentration estimate]\label{talagrand}

\noindent
For $d\in\N$ let $F:\R^{d}\to\R$ be a (globally) {\scshape Lipschitz} continuous function, $|F(x)-F(x')|\leq L|x-x'|$, with some constant $L>0$, arbitrary $x,x'\in\R^{d}$and $|\,(\,\cdot\,)\,|$ denoting the {\scshape Euclid}ean norm on $\R^{d}$.
Moreover, let $g\ceq\{g_{1},\dots,g_{d}\}$ be an independent collection of $d$ {\scshape Gauss}ian random variables with common mean $0$ and variance $1$. Then
\begin{equation}
	\PP\big\{\big|F(g)-\EE[F(g)]\big|> \delta\big\}\leq 2\exp\Big(-\frac{\delta^{2}}{4L^{2}}\Big)
\end{equation}
for any $\delta>0$.
\end{proposition}
It is the basis of
\begin{lemma}[Large-deviation estimate for the free energy]\label{gc}

\noindent
For the (random) free energy $f_{N}$ defined in \eqref{f_N} with $\vv >0$ we have
\begin{equation} \label{gc_e}
	\PP\big\{\big| \beta f_{N}-\EE[\beta f_{N}]\big|> \delta\big\}\leq 2\exp\Big(-\frac{N^{2}\delta^{2}}{2(N-1)(\beta\vv)^{2}}\Big)
\end{equation}
for any total number of spins $N\geq2$ and any $\delta>0$.
\end{lemma}

\begin{proof}
	We interprete the coefficients $(g_{ij})_{1\leq i<j\leq N}$ in the quantum {\scshape Hamilton}ian\, $H_{N}$, defined in \eqref{H_N}, as the components of a non-random vector $g\in\R^{d}$ with $d=N(N-1)/2$, and write more explicitly $H_{N}(g)$ and $f_{N}(g)$ for its (specific) free energy. In view of Proposition\,\ref{talagrand} we then only have to show that the function $g\mapsto\beta f_{N}(g)$ is {\scshape Lipschitz} continuous on $\R^{d}$ with constant $L=\beta\vv\sqrt{N-1}/(N\sqrt{2})$. To this end, we introduce the {\scshape Gibbs} expectation $\langle\,(\,\cdot\,)\,\rangle_{g}\ceq \e^{N\beta f_{N}(g)}\tr{\e^{-\beta H_{N}(g)}(\,\cdot\,)}$ induced by $H_{N}(g)$. Then the {\scshape Jensen--Peierls--Bogolyubov} inequality, see for example \cite{S2005b}, gives
\begin{align}
	\beta f_{N}(g')-\beta f_{N}(g)&\leq\frac{\beta}{N}\big\langle H_{N}(g')-H_{N}(g)\big\rangle_{g}=\frac{\beta\vv}{N^{3/2}}\sum\limits_{1\leq i< j\leq N}\big(g_{ij}-g'_{ij}\big)\langle S^{\z}_{i}S^{\z}_{j}\rangle_{g}\qquad\big(g,g'\in\R^{d}\big)\\
	&\leq\frac{\beta\vv}{N^{3/2}}\sum\limits_{1\leq i< j\leq N}\big|g_{ij}-g'_{ij}\big|\big|\langle S^{\z}_{i}S^{\z}_{j}\rangle_{g}\big|\leq\frac{\beta\vv}{N^{3/2}}|g-g'|_{1}\leq L|g-g'|\,.\label{triangle_jensen}
\end{align}
For \eqref{triangle_jensen} we have used the triangle inequality, the operator inequalities $-\um\leq S^{\z}_{i}S^{\z}_{j}\leq \um$, and the ({\scshape Jensen}) inequality $|x|_{1}\leq\sqrt{d}\, |x|$ between the $1$-norm and the $2$-norm of $x=(x_{1},\dots,x_{d})\in\R^{d}$. By considering the last chain of inequalities also with $g$ and $g'$ interchanged we get the desired {\scshape Lipschitz} continuity.
\qed
\end{proof}
\begin{remark}\label{rem_lde}

A similar result was already given by {\scshape Crawford} \cite{C2007}. We include the lemma for two reasons. First, it serves to make the present paper reasonably self-contained. Second, the above proof is simpler than the one in \cite{C2007}. It does not need the {\scshape PFK} spin-flip representation and can easily be extended to quantum spin-glass models with additional mean-field type interactions between the spins, for example to the quantum mean-field {\scshape Heisenberg} spin-glass model with or without an external magnetic field \cite{BM1980,S1981}.
\end{remark}

\section*{Acknowledgements}
We are grateful to Peter EICHELSBACHER (Bochum, Germany) for useful discussions and pointing out to us Reference\,\cite{CET2005}. We also thank Bikas K. CHAKRABARTI (Kolkata, India) for a clarifying correspondence and one of the referees for constructive
remarks.


\end{document}